%% file: main.tex
\documentclass[11pt]{article}
\usepackage[pdftex]{graphicx}
\usepackage[T1]{fontenc}
\usepackage{lmodern}

\usepackage{fullpage}

\usepackage{amsmath,amsthm,amssymb}
\usepackage{verbatim}
\usepackage{mathtools}
\usepackage{algorithm}
\usepackage{algorithmicx}
\usepackage{algpseudocode}
\usepackage{subfig}
\usepackage[toc,page]{appendix}
\usepackage[usenames,dvipsnames]{color}

\usepackage{caption}
\usepackage{thmtools}
\usepackage{thm-restate}

\usepackage{hyperref}

\usepackage{cleveref}

\declaretheorem[name=Theorem,numberwithin=section]{thm}
\declaretheorem[name=Theorem,numberlike=thm]{theorem}
\declaretheorem[name=Lemma,numberlike=thm]{lemma}
\declaretheorem[name=Claim,numberlike=thm]{claim}
\declaretheorem[name=Claim,numberlike=thm]{clm}
\declaretheorem[name=Corollary,numberlike=thm]{cor}
\declaretheorem[name=Definition,numberlike=thm,style=definition]{defn}

\makeatletter
\newcommand{\multiline}[1]{%
  \begin{tabularx}{\dimexpr\linewidth-\ALG@thistlm}[t]{@{}X@{}}
    #1
  \end{tabularx}
}
\makeatother

\DeclarePairedDelimiter{\ceil}{\lceil}{\rceil}
\DeclarePairedDelimiter{\floor}{\lfloor}{\rfloor}
 
\DeclareMathOperator*{\argmax}{arg\,max}
\DeclareMathOperator*{\argmin}{arg\,min}

\global\long\def\norm#1{\big\|#1\big\|}
 \global\long\def\normFull#1{\left\Vert #1\right\Vert }

\newcommand{\runningtimedelta}{O\left(\btt^2\boo\log^{O(1)}(\frac{\boo\btt}{\delta}) +\btt^3\log^{O(1)}(\frac{\boo\btt}{\delta}) \right)}

\newcommand{\boo}{b_1}
\newcommand{\btt}{b_2}
\newcommand{\otbtt}{[1,\btt]}
\newcommand{\ztbtt}{[0,\btt]}
\newcommand{\otboo}{[1,\boo]}

\newcommand{\bttpo}{(\btt+1)}

\newcommand{\Tn}{\Phih^{n}}

\newcommand{\F}{\phi}
\newcommand{\Fd}{\phid}
\newcommand{\1}{\overrightarrow{1}}
\newcommand{\R}{\mathbb{R}}

\newcommand{\bbP}{\mathbb{P}}
\newcommand{\E}{\mathbb{E}}
\newcommand{\Z}{\mathbb{Z}_{+}}

\newcommand{\simplex}{\Delta^{\bX}}
\newcommand{\psimplex}{\Delta_{\mathrm{pseudo}}^{\bX}}
\newcommand{\dsimplex}{\Delta_{\mathrm{discrete}}^{\bX}}

\newcommand{\nd}{n'}

\newcommand{\lt}{n^2}
\newcommand{\pp}{\pvec}
\newcommand{\nn}{\mvec}
\newcommand{\probpml}{\bbP}
\newcommand{\probsdpml}{\bw_{\mathrm{sdpml}}}

\newcommand{\sd}{c}

\newcommand{\phid}{\phi '}
\newcommand{\phih}{\lphi}
\newcommand{\phihd}{\phih '}
\newcommand{\Phihd}{\Phih '}
\newcommand{\Phih}{\Psi}
\newcommand{\Phid}{\Phi '}
\newcommand{\bw}{\textbf{w}}
\newcommand{\bD}{\textbf{D}}
\newcommand{\bd}{\textbf{d}}
\newcommand{\bA}{\textbf{A}}

\newcommand{\bh}{\textbf{h}}
\newcommand{\bg}{\textbf{g}}

\newcommand{\ba}{\textbf{a}}

\newcommand{\bk}{X}

\newcommand{\bp}{\textbf{p}}
\newcommand{\bq}{\textbf{q}}
\newcommand{\bff}{\textbf{f}}
\newcommand{\ff}{\textbf{f}}
\newcommand{\bX}{\mathcal{D}}
\newcommand{\bK}{\textbf{K}}
\newcommand{\bH}{\textbf{H}}
\newcommand{\bP}{\textbf{P}}

\newcommand{\bZ}{\textbf{Z}}

\newcommand{\bY}{\textbf{Y}}

\newcommand{\epssec}{\epsilon_{1}}
\newcommand{\epsone}{\epsilon}
\newcommand{\epstwo}{\epsilon}
\newcommand{\bS}{\textbf{S}}
\newcommand{\bM}{\textbf{M}}

\newcommand{\expo}[1]{\exp \left(#1 \right)}
\newcommand{\exps}[1]{\exp (#1 )}

\newcommand{\boundx}{\frac{\|\vd\|_2^2}{4\lambda^2} + \frac{n^9}{\reg^3}}
\newcommand{\boundalpha}{\max(\sqrt{\frac{\|d\|^2+\sd^2}{\lambda}+\delta\lambda+|\sd|} ~, ~\frac{1}{4\lambda^2 n^2}(\|d\|^2+\sd^2)+\lambda n^2+|\sd|+\frac{\delta}{n^2}~,~1)}

\newcommand{\vC}{\textbf{C}}
\newcommand{\opt}{\textbf{OPT}}
\newcommand{\reg}{\gamma}
\newcommand{\fnH}{\textbf{H}}
\newcommand{\fnHa}{\textbf{H}}
\newcommand{\fnf}{\hat{\textbf{f}}}
\newcommand{\fnh}{\textbf{h}}
\newcommand{\fng}{\textbf{g}}
\newcommand{\vx}{X}
\newcommand{\vvx}{x}

\newcommand{\vxe}{X^{(\epsilon)}}

\newcommand{\ste}{\textbf{t}^{(\epsilon)}}

\newcommand{\vxa}{X^{(\alpha)}}
\newcommand{\vya}{Y^{(\alpha)}}
\newcommand{\vyk}{Y^{(K)}}
\newcommand{\vy}{Y}
\newcommand{\vxopt}{X^{*}}
\newcommand{\vxo}{X^{(1)}}
\newcommand{\vxt}{X^{(2)}}
\newcommand{\vxstar}{X^{*}}
\newcommand{\ao}{\alpha^{(1)}}
\newcommand{\astar}{\alpha^{*}}
\newcommand{\at}{\alpha^{(2)}}
\newcommand{\st}{\textbf{t}}
\newcommand{\ststar}{\textbf{t}^{*}}
\newcommand{\sto}{\textbf{t}^{(1)}}
\newcommand{\stopt}{\textbf{t}^{*}}
\newcommand{\stt}{\textbf{t}^{(2)}}
\newcommand{\sta}{\textbf{t}^{\left(\alpha\right)}}
\newcommand{\vd}{D}
\newcommand{\setK}{\textbf{K}}

\newcommand{\imply}{\Rightarrow}

\newcommand{\runningtime}{O\left(\btt^2\boo\log^{O(1)}(\boo\btt) +\btt^3\log^{O(1)}(\boo\btt) \right)}
\newcommand{\epsrunningtime}{O\left(\phi_{size}+\frac{1}{\epst^2  \epso} \log^{O(1)}(\frac{1}{\epso \epst})+\frac{1}{\epst^3} \log^{O(1)}(\frac{1}{\epso \epst}) \right)}
\newcommand{\epsnrunningtime}{O\left(\sqrt{n}+\frac{1}{\epst^2  \epso} \log^{O(1)}(\frac{1}{\epso \epst}) +\frac{1}{\epst^3} \log^{O(1)}(\frac{1}{\epso \epst}) \right)}
\newcommand{\otilde}{\widetilde{O}}
\newcommand{\timeoracle}{O(\boo\cdot \btt \log(\frac{\bxmax}{\epsilon_{2}})\log(\bamaxo))}

\newcommand{\dxij}{\Delta_{\vx_{ij}}}
\newcommand{\dxt}{\Delta_{\st}}
\newcommand{\sa}{\textbf{a}}
\newcommand{\sK}{\textbf{K}}

\newcommand{\sKai}{\textbf{K}^{\alpha,i}}
\newcommand{\sbb}{\textbf{b}}
\newcommand{\lb}{\ell}
\newcommand{\ub}{u}

\newcommand{\fnga}{\textbf{h}^{\left(\alpha \right)}}

\newcommand{\fngo}{\textbf{h}^{\left(1 \right)}}
\newcommand{\fngt}{\textbf{h}^{\left(2 \right)}}

\newcommand{\Cphid}{\binom{\nd}{\phid}}

\newcommand{\cphi}{C_{\phi}}
\newcommand{\cphid}{C_{\phid}}
\newcommand{\pmin}{2n^2}
\newcommand{\pk}{\bq_{\bk}}
\newcommand{\lphi}{\psi}

\newcommand{\levelq}{\ell^{\bq}}
\newcommand{\conset}{\textbf{K}}
\newcommand{\ma}{\textbf{A}}
\newcommand{\mb}{\textbf{B}}
\newcommand{\mzero}{\textbf{0}}
\newcommand{\vb}{\textbf{b}}

\newcommand{\pvec}{\zeta}
\newcommand{\vvec}{\textbf{v}}
\newcommand{\logpvec}{\log(\pvec)}
\newcommand{\mvec}{\mathrm{m}}
\newcommand{\epso}{\epsilon_1}
\newcommand{\epst}{\epsilon_2}
\newcommand{\discloss}{\epso n +7\epst n \log n}
\newcommand{\odiscloss}{\epso n +\epst n \log n}
\newcommand{\grouploss}{ \frac{\log^3 n}{\epso \epst} }
\newcommand{\extrel}{\expo{-O\left( \frac{\log^3 n}{\epso \epst}  \right)}}

\newcommand{\eqdef}{\stackrel{\mathrm{def}}{=}}
\newcommand{\defeq}{\eqdef}
\newcommand{\tr}{\mathrm{tr}}
\newcommand{\onevec}{\mathrm{1}}
\newcommand{\slo}{\sum_{i}p_i}
\newcommand{\slt}{\sum_{i}{p_{i}^2}}

\newcommand{\vecc}{c}
\newcommand{\vecb}{b}
\newcommand{\vecy}{y}
\newcommand{\vecx}{x}
\newcommand{\vecz}{z}
\newcommand{\del}{\delta}

\newcommand{\dimn}{\btt}
\newcommand{\timeora}{OO^{(2)}_{\lambda.\del}(\conset)}
\newcommand{\bamax}{\textbf{B}_{\alpha,\epsone}}
\newcommand{\bamaxo}{\textbf{B}_{\alpha,\delta}}
\newcommand{\bxmax}{\textbf{B}_{\vx}}

\newcommand{\poly}{\mathrm{poly}}

\usepackage[utf8]{inputenc} % allow utf-8 input
\usepackage[T1]{fontenc}    % use 8-bit T1 fonts
\usepackage{hyperref}       % hyperlinks
\usepackage{url}            % simple URL typesetting
\usepackage{booktabs}       % professional-quality tables
\usepackage{amsfonts}       % blackboard math symbols
\usepackage{nicefrac}       % compact symbols for 1/2, etc.
\usepackage{microtype}      % microtypography

\title{Efficient Profile Maximum Likelihood for \\
Universal Symmetric Property Estimation}

\author{
  Moses Charikar\\
  Stanford University\\
  \texttt{moses@cs.stanford.edu}
  \thanks{Supported by NSF grant CCF-1617577, a Simons Investigator Award and a Google Faculty Research Award.}
   \\
   \and
Kirankumar Shiragur\\
  Stanford University\\
  \texttt{shiragur@stanford.edu} \\
     \and
Aaron Sidford\\
  Stanford University\\
  \texttt{sidford@stanford.edu}
  \thanks{Supported by NSF CAREER Award CCF-1844855.} \\
}

\begin{document}

\maketitle

\begin{abstract}
Estimating symmetric properties of a distribution, e.g. support size, coverage, entropy, distance to uniformity, are among the most fundamental problems in algorithmic statistics. While each of these properties have been studied extensively and separate optimal estimators are known for each, in striking recent work, Acharya et al. \cite{ADOS16} showed that there is a single estimator that is competitive for all symmetric properties. 
This work proved that computing the distribution that approximately maximizes \emph{profile likelihood (PML)}, i.e. the probability of observed frequency of frequencies, and returning the value of the property on this distribution is sample competitive with respect to a broad class of estimators of symmetric properties. Further, they showed that even computing an approximation of the PML suffices to achieve such a universal plug-in estimator. 
Unfortunately, prior to this work there was no known polynomial time algorithm to compute an approximate PML and it was open to obtain a polynomial time universal plug-in estimator through the use of approximate PML. 

In this paper we provide a algorithm (in number of samples) that, given $n$ samples from a distribution, computes an approximate PML distribution up to a multiplicative error of $\exps{n^{2/3} \poly \log(n)}$ in time nearly linear in $n$. Generalizing work of \cite{ADOS16} on the utility of approximate PML we show that our algorithm provides a nearly linear time universal plug-in estimator for all symmetric functions up to accuracy $\epsilon = \Omega(n^{-0.166})$. Further, we show how to extend our work to provide efficient polynomial-time algorithms for computing a $d$-dimensional generalization of PML (for constant $d$) that allows for universal plug-in estimation of symmetric relationships between distributions.
\end{abstract}

\newpage

\input{intro.tex}
\input{preliminaries.tex}

\input{results.tex}

\input{structure.tex}
\input{symmetric.tex}
\bibliographystyle{alpha}
\bibliography{PML}
\input{appendix.tex}
\input{appendixprofile.tex}
\input{appendixstructure.tex}

\input{oracle_temp.tex}
\input{multipmlappendix.tex}

\end{document}

%% file: intro.tex
\section{Introduction}
Estimating a symmetric property of a distribution given a small number of samples is a fundamental problem in algorithmic statistics.
Formally, a property is symmetric if it is invariant to permutation of the labels, i.e. it is a function only of the multiset of probabilities and does not depend on the symbol labels. 
For many natural properties, including support size, coverage, distance from uniform and entropy, there has been extensive work that has led to designing efficient estimators both with respect to computational time and sample complexity \cite{HJW17, HJM17, AOST14, RVZ17, ZVVKCSLSDM16, WY16a, RRSS07, WY15, OSW16, VV11a, WY16, JVHW15, JHW16, VV11b}. In many cases these estimators are tailored to the particular property of interest.
This paper is motivated by the goals of unifying the development of efficient estimators of symmetric properties of distributions and designing a single efficient universal algorithm for estimating arbitrary symmetric properties of distributions. 

Our approach stems from the observation that a sufficient statistic for the problem of estimating a symmetric property from a sequence of samples is the profile of the sequence, i.e. the multiset of the frequencies (i.e multiplicities) of symbols in the sequence, e.g. the profile of $ababc$ is $\{2,2,1\}$.
Profiles are also called histograms of histograms, histogram order statistics, or fingerprints.
Our approach to obtaining a universal estimator is based on the elegant problem of \emph{profile maximum likelihood (PML)} introduced by Orlitsky et al.~\cite{OSSVZ04}: Given a sequence of $n$ samples, find the distribution that maximizes the probability of the observed profile. This problem has been studied in several papers since, applying heuristic approaches such as Bethe approximation \cite{Von12, Von14}, the EM algorithm \cite{OSSVZ04}, and some algebraic approaches \cite{ADMOP10} to calculate the PML. Recently Pavlichin, Jiao and Weissman~\cite{PJW17} introduced an efficient dynamic programming heuristic for PML that can be computed in linear time. While there are no approximation guarantees for the solution they produce, their approach was the initial impetus for our work.

A recent paper of Acharya et al.~\cite{ADOS16} showed that a distribution that optimizes the PML objective can be used to obtain a plug-in estimator for various symmetric properties of distributions. In fact it suffices to compute a distribution that approximates the PML objective to within a factor $\exps{n^{1-\delta}}$ for constant $\delta > 0$ where $n$ is the size of the sample. Unfortunately, no polynomial time computable PML estimator with such an approximation guarantee was known previously. In this paper, we provide an estimator with an approximation factor of $\exps{n^{2/3} \poly\log(n)}$, leading to a universal estimator for a host of symmetric properties. Moreover, our estimator is computable in time nearly linear in $n$.
Our techniques extend to computing a $d$-dimensional generalization of PML, where we have access to samples from multiple distributions on a common domain.
This allows for universal plug-in estimation of various symmetric relationships between multiple distributions.

\subsection{Overview of approach}
The bulk of our work is dedicated to find a distribution that approximately maximizes the PML objective within an $\exps{n^{1-\delta}}$ factor for a constant $\delta>0$.
We call such a  distribution an \emph{approximate PML distribution}.
Given a sequence $y^n$ and its corresponding profile $\phi$, the PML optimization problem is a maximization problem over all distributions $\bp \in \simplex$.
The objective function of the PML optimization problem is the probability of observing profile $\phi$ with respect to a distribution $\bp \in \simplex$, which in turn is equal to the summation of probabilities of sequences (with respect to $\bp$) that have $\phi$ as their corresponding profile. The distribution that maximizes this objective is called a \emph{profile maximum likelihood (PML) distribution}. (See \Cref{sec:prelimsmain} for formal definitions.)

To efficiently compute an approximate PML distribution, we first restrict ourselves to maximizing the PML objective for a discretized version of the profile over a class of distributions we call \emph{discrete pseudo-distributions} (See \Cref{sec:structuremain}). 
Here, the probability values of the distribution are restricted to belong to a small set $\bP$ of permissible values (See \Cref{sec:prob_discrete})), and the frequencies in the profile are similarly restricted to belong to a small set $\bM$ (See \Cref{sub:mult_discrete}). We call the resulting maximizing distribution, a \emph{discrete PML (DPML) distribution} and the corresponding optimization problem as \emph{DPML optimization} (See \Cref{subsec:DPMLmain}). 
%These discrete pseudo-distributions take probability values from some restricted set $\bP$.

There are two main features of the DPML optimization problem. Firstly, the maximizing distribution DPML is an approximate PML distribution with an approximation guarantee that we can control (as a function of the sizes of $\bP$ and $\bM$). Secondly, the DPML optimization problem has a simpler equivalent formulation, in which sequences that have the same associated probability value with respect to a discrete pseudo-distribution are combined together into sub groups and the whole summation is written as a summation over a small number of subgroups. The number of these subgroups is a function of the sizes of $\bP$ and $\bM$ which we control (See \Cref{subsec:DPMLmain} for both these results).

As an illustration of DPML,  
%let $\bX=\{a,b,c,b \}$ and $\bP=\{ 1/4\}$. Given sequence \emph{ab}, its profile is $\{1,1 \}$ and set of all sequences that have $\{1,1\}$ as their profile is $\{ab,ac,ad,bc,bd,cd \}$ union their permutations. The probability of profile $\{1,1 \}$ is $2(\bp_{a}\bp_{b} + \bp_{a}\bp_{c} + \bp_{a}\bp_{d}+ \bp_{b}\bp_{c} + \bp_{b}\bp_{d} + \bp_{c}\bp_{d})$, where probability term $\bp_{a}\bp_{b}$ corresponds to sequence \emph{ab} and multiplicative factor $2$ is because of its permutation \emph{ba}. Since our restricted set is $\bP=\{1/4\}$, the only discrete pseudo-distribution possible is $\bp_{a}=\bp_{b}=\bp_{c}=\bp_{d}=1/4$ and probability terms corresponding to all sequences $\{\emph{ab,ac,ad,bc,bd,cd} \}$ and their permutations are all the same and equal to $1/16$. Restricting probabilities values to $\bP$ grouped all summations terms in this example to just one subgroup.
consider the profile $\{2,1,1\}$ and a probability distribution on 5 elements: two with a value of $\frac{1}{4}$ and three with a value of $\frac{1}{6}$. 
Note that the probability values come from the set $\bP=\{1/4,1/6\}$.
One way to get the profile $\{2,1,1\}$ is to have an element of probability $1/4$ appear twice and two elements of probability $1/6$ appear once. There are ${2 \choose 1} {3 \choose 2}$ choices of such elements and for each such choice, $\frac{4!}{2!\cdot 1!\cdot 1!}$ sequences of length 4 with these elements.
The probability of any such sequence is the same: $\left( \frac{1}{4} \right)^2 \left(\frac{1}{6}\right)\left(\frac{1}{6}\right)$.
We consider the set of all these sequences as one subgroup.
Different subgroups are identified by specifying, for each permissible probability value, the frequencies with which elements of that probability value are seen in the sample.
The DPML objective then sums up the contributions of each such subgroup.

%In the DPML formulation above, each subgroup is a set of sequences that share the same probability value. Moreover, the computation of this probability has a similar structure  across subgroups. 
Reformulating the problem in terms of summation over a small number of subgroups is crucial to our approach.  It allows us to focus on the subgroup that gives the largest contribution to the objective instead of summing over all the subgroups. We call the optimization problem that optimizes the contribution of a single subgroup (instead of summing over all terms) as \emph{single discrete PML (SDPML)}. 
We show that the SDPML optimization problem has a convex relaxation and can be solved efficiently. Since there were a small number of these subgroups in the summation, the optimizing discrete pseudo-distribution that optimizes over just one subgroup has objective function value that is lower by at most the number of subgroups.
Hence the maximizing discrete pseudo-distribution for this new objective function approximately optimizes the earlier objectives (PML and DPML) with bounded loss (See \Cref{subsec:DPMLmain}).
%is an approximate PML distribution with a good approximation guarantee. 

Ultimately, our algorithm first solves this convex relaxation to the SDPML optimization problem to obtain a fractional solution (in some representation space of these discrete pseudo-distributions) (See \Cref{subsec:convexrelaxmain}). Then we apply a rounding algorithm that finds a distribution which maintains the approximation guarantee need to obtain an approximate PML distribution (See \Cref{subsec:algmain}). 

\subsection{Related work}
As discussed in the introduction, PML was introduced by Orlitsky et al.~\cite{OSSVZ04} in 2004. Many heuristic approaches such as Bethe approximation \cite{Von12, Von14}, the EM algorithm \cite{OSSVZ04}, algebraic approaches \cite{ADMOP10} and a dynamic programming approach~\cite{PJW17} have been proposed to calculate the approximate PML.

The connection between PML and universal estimators was first studied in~\cite{ADOS16}. There have been several other approaches for designing universal estimators for symmetric properties. Valiant and Valiant~\cite{VV11a} adopted and rigorously analyzed a linear programming based approach for universal estimators proposed by \cite{ET76} and showed that it is sample complexity optimal in the constant error regime for estimating certain symmetric properties (namely, entropy, support size, support coverage, and distance to uniformity). Recent work of Han, Jiao and Weissman~\cite{HJW18} applied a local moment matching based approach in designing efficient universal symmetric property estimators for a single distribution. \cite{HJW18} achieves the optimal sample complexity in all error regimes for estimating the power sum function, support and entropy.

Estimating symmetric properties of a distribution is a rich field and extensive work has been dedicated to studying their optimal sample complexity for estimating each of these properties. Optimal sample complexities for estimating many symmetric properties were resolved in the past few years, including all the properties studied here: support~\cite{VV11a, WY15}, support coverage~\cite{OSW16,ZVVKCSLSDM16}, entropy~\cite{VV11a, WY16} and distance from uniform~\cite{VV11b,JHW16}.

Symmetric properties for distribution pairs have been studied in the literature as well. For instance, optimal sample complexity for estimation of KL divergence between two distributions were given by \cite{BZLV16, HJW16}.

\subsection{Paper organization}

The rest of the paper is structured as follows. In \Cref{sec:prelimsmain}, we provide definitions and notations. In \Cref{sec:results}, we state our main results of the paper. Our main contribution is to provide an algorithm that efficiently compute an approximate PML and in \Cref{sec:structuremain} we prove this result. In this section, we also present an almost linear time algorithm based on cutting plane methods for solving our convex relaxation to SDPML; however we defer all of its analysis to the appendix. Finally, in \Cref{app:universal}, we provide the connection between approximate PML distribution and a universal estimator for symmetric property estimation. The proof presented in \cite{ADOS16} showed this connection for an $\exp(\sqrt{n})$-approximate PML estimator and we show it for an $\exp(n^{2/3})$-approximate PML estimator. However it is easy to see the proof presented in \cite{ADOS16} works for any $\exp(n^{1-\delta})$-approximate PML estimator for constant $\delta >0$. In \Cref{app:multipmlstart} we show that the techniques presented here generalize to a higher dimensional version of PML. 

%% file: preliminaries.tex
\newcommand{\setd}{\textbf{D}}
\newcommand{\eled}{d}

\section{Preliminaries}\label{sec:prelimsmain}
Let $[a,b]$ and $[a,b]_{\R}$ denote the interval of integers and reals $\geq a$ and $\leq b$ respectively and let $[a] \defeq [1,a]$.  Let $\simplex \subset [0,1]_{\R}^{\bX}$ be the set of all distributions supported on domain $\bX$ and let $N$ be the size of the domain. We use the word distribution to refer to discrete distributions. Throughout this paper we assume that we receive a sequence of $n$ independent samples from an underlying distribution $\bp \in \simplex$. Let $\bX ^n$ be the set of all length $n$ sequences and $y^n \in \bX^n$ be one such sequence with $y^n_{i}$ denoting its $i$th element.
%\sidford{Since we never change $n$, can we just write $y$ instead of $y^n$ ? Would that make things more compact? Whenever we want to emphasize the cardinality of $y$ we can just write $y \in \bX^n$?}\kiran{We need $x^n$ for couple of proofs and we already used $x$ to denote domain element. In that place having $x^n$,$y$ to denote sequences would be confusing. Delete if agree } 
The probability of observing sequence $y^n$ is:
$$\bbP(\bp,y^n) \defeq \prod_{x \in \bX}\bp_x^{\bff(y^n,x)}$$
where $\bff(y^n,x)= |\{i\in [n] ~ | ~ y^n_i = x\}|$ is the frequency (multiplicity) of symbol $x$ in sequence $y^n$ and $\bp_x$ is the probability of domain element $x\in \bX$.

We extend and use the definition for $\bbP(\vvec,y^n)$ to any vector $\vvec \in \R^{\bX}$ by letting $\bbP(\vvec,y^n) \defeq \prod_{x \in \bX}\vvec_x^{\bff(y^n,x)}$. Further, for functions of probability distributions $\bp$, we assume those expressions are also defined for any vector $\vvec \in \R^{\bX}$ just by replacing $\bp_{x}$ by $\vvec_{x}$ everywhere.

For any given sequence one could define its \emph{type} (histogram) and \emph{profile} (histogram of a histogram or fingerprint) that are sufficient statistics for symmetric property estimation. The histogram of histogram perspective comes from viewing type as a histogram and profile as histogram of type.
\begin{defn}[Type]
A \emph{type} $\phih=\Phih(y^n) \in \Z^{\bX}$ of a sequence $y^n \in \bX^{n}$ is the vector of frequencies $\phih_x\defeq \bff(y^n,x)$ of domain elements in $y^n$.
We call $n$ the length of type $\phih$ and use $\Tn$ to represent the set of all types of length $n$.
\end{defn}
 To simplify notation we use just $\phih$ to denote type and the associated sequence will be clear from context. For a distribution $\bp \in \simplex$, the probability of a type $\phih \in \Tn$ is:  
\[
\bbP(\bp,\phih) \defeq \sum_{\{y^n \in \bX^n~|~ \Phih (y^n)=\phih \}}\bbP(\bp,y^n)=\binom{n}{\phih}\prod_{x \in \bX}\bp_x^{\phih_x},
\]
where $\binom{n}{\phih}\defeq\frac{n!}{\prod_{x \in \bX}\phih_x!}$ and $0!\defeq1$.
%
%For a type $\phih$, note that the quantity $\binom{n}{\lphi}$ depends just on the multiset of frequencies of domain elements (without regard to their identities and probabilities). 
\begin{defn}[Profile]
For any sequence $y^n \in \bX^n$, let $\setd=\{ \bff(y^n,x) \}_{x \in \bX}$ be the set of all its distinct frequencies and $\eled_1,\eled_2,\dots, \eled_{|\setd|}$ be elements of the set $\setd$. The \emph{profile} of a sequence $y^n \in \bX^{n}$ denoted $\phi=\Phi(y^n) \in \Z^{|\setd|}$ 
is $\phi\defeq(\F_j)_{j=1\dots |\setd|} $ where $\F_j=\F_j(y^n)\defeq|\{x\in \bX ~|~\bff(y^n,x)=\eled_{j} \}|$ is the number of domain elements with frequency $\eled_{j}$ in $y^{n}$. We call $n$ the length of profile $\F$ and as a function of profile $\phi$, $n = \sum_j \eled_{j} \cdot \F_j$. We let $\Phi^n$ denote the set of all profiles of length $n$.
\footnote{The number of unseen domain elements is not part of the profile, because the domain size is unknown.}
\end{defn}
For any distribution $\bp \in \simplex$, the probability of a profile $\phi \in \Phi^n$ is defined as:
\begin{equation}\label{eqpml1main}
\probpml(\bp,\phi)\defeq\sum_{\{y^n \in \bX^n~|~ \Phi (y^n)=\phi \}} \bbP(\bp,y^n)\\
\end{equation}
One can also define the profile of a type $\phih$. We overload notation and use $\phi=\Phi(\phih)$ to denote the profile associated with type $\phih$
and $\F_j=\F_j(\phih)\defeq|\{x\in \bX ~|~\phih_x = \eled_{j} \}|$. 

For future use, we also write the probability of a profile $\phi \in \Phi^{n}$ in terms of its types. All types $\lphi$ with $\Phi(\lphi)=\phi$ have the same $\binom{n}{\lphi}$ value and we use notation $\cphi$ to represent this quantity. The explicit expression for $\cphi$ is written below:
\begin{equation}\label{cphi}
\cphi \defeq \frac{n!}{\prod_{j=1}^{|\setd|}(\eled_{j}!)^{\phi_{i}}}, 
\mbox{\ \  where\ } n=\sum_j \eled_{j} \cdot \phi_i
%\cphi \defeq \frac{n!}{\prod_{x \in \bX}\phih_{x}!}=\frac{n!}{\prod_{i=1}^{n}\phi_{i} \cdot i!}
\end{equation}
We next derive an expression for the probability of a profile in terms of its types:
\begin{equation}\label{eqlabeled}
\begin{split}
\probpml(\bp,\phi)&=\sum_{\{y^n \in \bX^n~|~ \Phi (y^n)=\phi \}} \bbP(\bp,y^n)
%=\sum_{\{\phih \in \Tn ~|~ \Phi(\phih)=\phi\}}~\sum_{\{y^n \in \bX^n~|~ \Phih (y^n)=\phih \}}\bbP(\bp,y^n) \\
=\sum_{\{\phih\in \Tn ~|~ \Phi(\phih)=\phi\}}\bbP(\bp,\phih)
%=\sum_{\{\phih\in \Tn ~|~ \Phi(\phih)=\phi\}} \binom{n}{\phih}\prod_{x \in \bX}\bp_x^{\phih_x}
=\cphi \sum_{\{\phih \in \Tn ~|~ \Phi(\phih)=\phi\}} \prod_{x \in \bX}\bp_x^{\phih_x}
\end{split}
\end{equation}

The distribution which maximizes the probability of a profile $\phi$ is called a profile maximum likelihood distribution.
\begin{defn}[Profile maximum likelihood] For any profile $\phi \in \Phi^{n}$, a \emph{profile maximum likelihood} (PML) distribution $\bp_{pml,\phi} \in \simplex$ is:
 $$\bp_{pml,\phi} \in \argmax_{\bp \in \simplex} \probpml(\bp,\phi)$$
 and $\probpml(\bp_{pml,\phi},\phi)$ is the maximum PML objective value.
\end{defn}

The central goal of this paper is to define efficient algorithms for computing approximate PML distributions defined as follows.
 
\begin{defn}[Approximate PML]
 For any profile $\phi \in \Phi^{n}$, a distribution $\bp^{\beta}_{pml,\phi} \in \simplex$ is a $\beta$-\emph{approximate PML} distribution if $$\probpml(\bp^{\beta}_{pml,\phi},\phi)\geq \beta \cdot \probpml(\bp_{pml,\phi},\phi)$$
\end{defn}
Throughout this paper we use the phrase \emph{approximate PML} to denote a $\beta$-approximate PML distribution for some non-trivial $\beta$.

\subsection{Representation of a profile}\label{subsec:representation}
For any profile $\phi \in \Phi^n$, we represent $\phi$ using the set of $(frequency,count)$ tuples, where a tuple $(a,b)$ denotes that $b$ number of domain elements have frequency $a$ in the sequence. We use $\phi_{size}$ to denote the size of profile $\phi$ in this representation. It is not hard to see that for any length $n$ profile $\phi_{size} \in O(\sqrt{n})$. Further it takes $O(n)$ time to write the profile in this representation.

For all our algorithmic results, when we are given a profile, we assume the above representation. We will explicitly state running times when we start with a sequence instead of a profile.

%% file: results.tex
\newcommand{\bpalg}{\textbf{p}_{\mathrm{approx}}}
\newcommand{\condo}{\textbf{C1}}
\newcommand{\condt}{\textbf{C2}}
\newcommand{\condth}{\textbf{C3}}
\newcommand{\n}{\textbf{n}}
\newcommand{\nkk}{\n(k)}
\newcommand{\epsok}{\epsilon(k)}
\newcommand{\epstk}{\gamma(k)}
\newcommand{\odisclossd}{}
\newcommand{\md}{d}
\newcommand{\grouplossd}{\prod_{k=1}^{d}\frac{\log^3 \nkk}{\epsok \epstk}}
\newcommand{\disclossd}{\sum_{k=1}^{\md}\epsok \nkk+\left(\sum_{k=1}^{\md}\epstk \nkk \right)\log \left( \sum_{k=1}^{\md} \epstk \nkk \right)}
\newcommand{\otdisclossd}{\sum_{k=1}^{\md}\epsok \nkk+\sum_{k=1}^{\md}\epstk \nkk }
\newcommand{\grouplossepsd}{\prod_{k=1}^{\md}\frac{\log^3 \nkk}{\epsok \epstk}}
\newcommand{\otgrouplossepsd}{\prod_{k=1}^{\md}\frac{1}{\epsok \epstk}}
\newcommand{\runningtimed}{\prod_{k=1}^{\md}\frac{1}{\epsok (\epstk)^2} + \prod_{k=1}^{\md}\frac{1}{(\epstk)^3}}
\newcommand{\epsnrunningtimed}{run}
\newcommand{\writeprofile}{\sum_{k=1}^{\md}\nkk}
\section{Results}\label{sec:results}
Here we state the main results of this paper. Our first main theorem provides an algorithm to efficiently compute an approximate PML distribution. Our approximation guarantee in this result is something that depends on the running time itself and we can achieve sub-linear running times (in size of the sample) if we allow for weaker approximation guarantees.
\begin{restatable}[Efficient and approximate PML distribution]{theorem}{thmapproxpml}
\label{thm:approxpml}
Given a profile $\phi \in \Phi^n$, let $\bp_{pml}$ be its corresponding PML distribution. There is an algorithm that for any $\frac{1}{\poly(n)}<\epso,\epst<1$, computes an $\exps{-O(\odiscloss+\grouploss)}$-approximate PML distribution $\bp_{approx}$, i.e.
$$\probpml(\bp_{approx},\phi) \geq \expo{-O\left(\odiscloss+\grouploss \right)}\probpml(\bp_{pml,\phi},\phi)$$
in $\epsrunningtime$ time. Using $\phi_{size} \in O(\sqrt{n})$ this running time simplifies to $\epsnrunningtime$.
\end{restatable}

In the above result, the best approximation is achieved for $\epso,\epst=n^{-1/3}$ and we get an $\exps{-O(n^{2/3}\log^3n)}$-approximate PML distribution in nearly linear time (in the number of samples). This result is summarized below.
\begin{cor}[Nearly linear time $\exps{-O(n^{2/3}\log^3n)}$- approximate PML distribution]\label{cor:twothirdpml}
Let $y^n \in \bX^{n}$ be a sequence and $\phi=\Phi(y^n)$ be its corresponding profile. There is an algorithm that computes an $\exps{-O(n^{2/3}\log^3n)}$-approximate PML distribution in time $\otilde(n)$.
\end{cor}
This results constitutes the first  polynomial time algorithm to compute an $\exps{-n^{1-\delta}}$ -approximate PML for any constant $\delta>0$.
In the corollary above we start with a sequence instead of a profile; in this case our algorithm still runs in $\otilde(n)$ because we only need $O(n)$ time to compute the profile of a sequence in the representation discussed in \Cref{subsec:representation}.

Our next result relates an approximate PML distribution to a universal plug-in estimator that is sample complexity optimal for support size, coverage, entropy and distance from uniform. In \Cref{app:universal}, we prove this result. However it is easy to see the proof presented in \Cref{app:universal} proves a more general result that approximate PML is sample complexity optimal for a broad class of symmetric properties $\bff(\cdot)$ satisfying certain conditions. One such set of conditions (informally) is the existence of an estimator $\widehat{\bff}$ for $\bff(\cdot)$ with following properties: $(1)$ the estimator $\widehat{\bff}$ is sample complexity optimal, $(2)$ the estimator $\widehat{\bff}$ has low bias, and $(3)$ the output of the estimator is not changed by much when we change any individual sample. This result was already shown in \cite{ADOS16} for an $\exps{-n^{0.5}}$-approximate PML distribution. Using the same proof with slight modifications we get the following result.

\begin{restatable}[Universal estimator using approximate PML]{theorem}{theoremADOSuniversal}
\label{theoremADOSuniversal}
Let $n$ be the optimal sample complexity of estimating entropy, support, support coverage and distance to uniformity and $c$ be a large positive constant. Let $\epsilon \geq \frac{3c}{n^{1/6-\eta}}$ for any constant $\eta>0$, then for any $\beta>\exps{-O(n^{2/3}\log^3n)}$, the $\beta$-approximate PML estimator estimates entropy, support, support coverage, and distance to uniformity to an accuracy of $4\epsilon$ with probability at least $1-\exps{-n^{2/3}}$.
\end{restatable}

Setting $\eta=1/6-0.166$ in the theorem above and combined with \Cref{cor:twothirdpml}, we obtain the following  result.

\begin{theorem}[Efficient universal estimator using approximate PML]
Let $n$ be the optimal sample complexity of estimating entropy, support, support coverage and distance to uniformity. If $\epsilon \geq \frac{3c}{n^{0.166}}$, then there exists a PML based universal plug-in estimator that runs in time $\otilde(n)$ and is sample complexity optimal for estimating entropy, support, support coverage and distance to uniformity to accuracy $4\epsilon$.
\end{theorem}

%Our techniques for PML are general and can be extended to generalization of PML to multiple dimensions (multidimensional PML). To be specific, we provide a polynomial time (in number of samples) algorithm to compute approximate PML in multiple dimensions when the dimension is constant. Approximate PML in multiple dimensions can be used to estimate symmetric properties of multiple distributions. For instance, in this paper we use approximate PML in two dimensions to estimate KL divergence between two unknown distributions. We  next formally define and state our main results for PML in higher dimensions. However we recommend readers to read the one dimensional case first. 
Our techniques for PML are general and can be extended to a generalization of PML to multiple dimensions (multidimensional PML). We provide a polynomial time (in number of samples) algorithm to compute approximate PML in multiple dimensions when the number of dimensions is constant. This allows for universal plug-in estimation of various symmetric relationships between multiple distributions. We next formally define and state our main results for multidimensional PML. 

\subsection{Results for multidimensional PML}\label{subsec:higherPML}
\newcommand{\cc}{\textbf{c}}

\newcommand{\freq}{\textbf{e}}
\newcommand{\freqk}{\textbf{e}(k)}
\newcommand{\freqj}{\textbf{e}_{j}}
\newcommand{\Freq}{\textbf{f}}
\newcommand{\Freqn}{\textbf{F}^{\n}}
\newcommand{\Freqnd}{\textbf{F}^{\n(\md)}}
\newcommand{\otmd}{[1,\md]}

\newcommand{\simplexd}{\Delta^{\bX,\md}}

\newcommand{\bpd}{\textbf{p}}
\newcommand{\bqd}{\textbf{q}}

\newcommand{\yn}{\textbf{y}^{\n}}
\newcommand{\xn}{\textbf{x}^{\n}}
\newcommand{\ynk}{\textbf{y}^{\nkk}}
\newcommand{\yni}{\textbf{y}^{\nii}}
\newcommand{\ynj}{\textbf{y}^{\njj}}
\newcommand{\yno}{\textbf{y}^{\n(1)}}
\newcommand{\ynt}{\textbf{y}^{\n(2)}}
\newcommand{\ynd}{\textbf{y}^{\n(\md)}}
\newcommand{\bffi}{\textbf{f}_{i}}
\newcommand{\bffk}{\textbf{f}_{k}}
%\ne wcommand{\vvec}{\textbf{v}}
\newcommand{\vveco}{\vvec{(1)}}
\newcommand{\vvecd}{\vvec{(\md)}}
\newcommand{\vveck}{\vvec{(k)}}
\newcommand{\vveci}{\vvec{(i)}}
\newcommand{\vvecj}{\vvec{(j)}}
\newcommand{\vsimplexd}{\Delta_{\mathrm{vector}}^{\bX,\md}}

%Here we introduce and state our results for multidimensional PML. We introduce this general setting here but defer all the proofs to \Cref{app:multipmlstart}. Throughout this paper we assume the dimension is constant. 
%Here we introduce and state our results for multidimensional PML. 
First we describe the multidimensional setting, then we define multidimensional PML, and then state our main results.
Throughout this paper we assume the number of dimensions is constant.

\paragraph{Multidimensional setup:} For each $k \in \otmd$, we receive a sequence $\ynk$ that consists of $\nkk$ independent samples drawn from an underlying distribution $\bpd(k)$ supported on same domain $\bX$ ($N\defeq |\bX|$), further $\ynk$ is independent of other sequences $\textbf{y}^{\n(k')}$ for $k' \in \otmd$ and $k' \neq k$. We call $\yn=(\yno,\dots \ynd)$ a $\md$-sequence and $\n=(\n(1),\dots, \n(\md))$ its $\md$-length. Let $\bX^{\n}$ be the set of all $\md$-sequences of $\md$-length equal to $\n$. We use $\bpd_{x}(k)$ to denote the probability of domain element $x$ in distribution $\bpd(k)$. We also refer to $\bpd=(\bpd(1),\dots,\bpd(\md))$ as a $\md$-distribution and let $\simplexd$ denote the set of all $\md$-distributions.

For any $\md$-distribution $\bp \in \simplexd$, the probability of a $\md$-sequence $\yn$ is defined as: $$\bbP(\bpd,\yn) \defeq \prod_{k=1}^{\md}\prod_{x \in \bX}(\bpd_{x}(k))^{\bff(\ynk,x)}~.$$ Recall that for each $k \in \otmd$, $\bff(\ynk,x)$ is the frequency of domain element $x$ in sequence $\ynk$. For any $\md$-sequence $\yn$, we call $\bff(\yn,x)=(\bff(\yno,x),\dots, \bff(\ynd,x))$ the $\md$-frequency of domain element $x$ in $\yn$. Let $\Freqn$ be the set of all $\md$-frequencies generated by different domain elements in all possible $\md$-sequences in $\bX^{\n}$ and we let $\freqj \in \Freqn$ denote its $j$th element. We next define multidimensional generalizations of profile, PML, and approximate PML.

\paragraph{$\md$-Profile:} For any $\md$-sequence $\yn \in \bX^{\n}$, we call $\phi=\Phi(\yn)$ a $\md$-profile if $\phi=(\F_j)_{j=1\dots |\Freqn|}$ and $\F_j=|\{x\in \bX ~|~\Freq(\yn,x)=\freqj \}|$ is the number of domain elements with $\md$-frequency $\freqj$. We call $\n$ the $\md$-length of $\F$ and use $\Phi^n$ to denote the set of all $\md$-profiles of $\md$-length equal to $\n$. For any $\md$-distribution $\bp \in \simplexd$, the probability of a $\md$-profile $\phi \in \Phi^n$ is defined as:
\begin{equation}\label{eqpml1d}
\probpml(\bp,\phi)\defeq\sum_{\{\yn \in \bX^{\n}~|~ \Phi (\yn)=\phi \}} \bbP(\bp,\yn)~.
\end{equation}

\paragraph{Profile maximum likelihood:} For any $\md$-profile $\phi \in \Phi^{n}$, a \emph{Profile Maximum Likelihood} $\md$-distribution $\bp_{pml,\phi} \in \simplexd$ is: $$\bp_{pml,\phi} \in \argmax_{\bp \in \simplexd} \probpml(\bp,\phi)$$ and $\probpml(\bp_{pml,\phi},\phi)$ is the maximum PML objective value.

\paragraph{Approximate profile maximum likelihood:} For any $\md$-profile $\phi \in \Phi^{n}$, a $\md$-distribution $\bp^{\beta}_{pml,\phi} \in \simplexd$ is a $\beta$-\emph{approximate PML} $\md$-distribution if $$\probpml(\bp^{\beta}_{pml,\phi},\phi)\geq \beta \cdot \probpml(\bp_{pml,\phi},\phi)$$.

We next state our results for approximate PML $\md$-distributions. In \Cref{thm:PMLd}, we give a algorithm to efficiently compute an approximate PML $\md$-distribution. Then, we substitute $\md=2$ in this result to get \Cref{PMLtwo}. 

\begin{restatable}[Efficient and approximate multidimensional PML]{theorem}{thmPMLd}
Let $\yn$ be a $\md$-sequence of $\md$-length $\n=(\n(1),\dots,\n(\md))$. There is an algorithm that computes an $\expo{-\otilde\left(\sum_{k=1}^{\md}\nkk^{1-1/(2\md+1)}\right)}$-approximate PML $\md$-distribution $\bpalg$ in $\otilde(\sum_{k=1}^{\md} \nkk+\prod_{k=1}^{\md} \nkk^{3/(2\md+1)})$ time\footnote{Here $\otilde$ notation hides all $\prod_{k=1}^{\md}\log^{O(1)} \nkk$ terms and therefore $O(d)$ term as well.}.
%Let $\yn$ be a $\md$-sequence of $\md$-length $\n=(\n(1),\dots,\n(\md))$. We can\sidford{The word efficient is somewhat informal and often used to mean polynomial or nearly linear (depending on context). In a formal theorem where you give a formal running time, I would then not use this word ``efficient'' as it is not clear what it is adding to the formal runtime; in short, I would just delete this word. I removed it and rearranged the text in the early standard PML theorems, if you would do the same here that would be great.}
%	I changed this in the analogous lemmas for PML efficiently\footnote{Here $\otilde$ notation hides all $\prod_{k=1}^{\md}\log^{O(1)} \nkk$ terms and therefore $O(d)$ term as well.} compute an $\exps{-\otilde\left(\sum_{k=1}^{\md}\nkk^{1-1/(2\md+1)}\right)}$-approximate PML $\md$-distribution $\bp_{approx}$ in time $\otilde(\sum_{k=1}^{\md} \nkk+\prod_{k=1}^{\md} \nkk^{3/(2\md+1)})$.
\label{thm:PMLd}
\end{restatable}

\begin{cor}[Efficient and approximate PML for two dimensions]\label{PMLtwo}
For $\md=2$, let $\yn$ be a $\md$-sequence of $\md$-length $\n=(\n(1),\n(2))$. There is an algorithm that computes an $\exps{-\otilde\left(\n(1)^{4/5} + \n(2)^{4/5} \right)}$-approximate PML $\md$-distribution $\bpalg$ in $\otilde(\n(1)+\n(2)+ \n(1)^{3/5}\n(2)^{3/5} )$ time.
\end{cor}

As mentioned before, one of the important applications of approximate multidimensional PML is in estimating symmetric properties for $\md$-distributions. A symmetric property is a function of $\md$-distributions that is invariant to a permutation of the labels. Here we study one such symmetric property for $\md=2$ called KL divergence that is studied in the context of PML. Estimation of KL divergence between two distributions is well studied and estimators that achieve optimal sample complexity were given by \cite{BZLV16, HJW16}.   In \Cref{theoremADOSuniversald}, we show that approximate PML is sample complexity optimal for estimating KL divergence. A similar result was already shown in \cite{Acharya18} (Theorem 6) for exact PML and we use the same proof with slight modification to prove our result. In \Cref{thm:effKL}, we give an efficient version of \Cref{theoremADOSuniversald} by combining it with \Cref{PMLtwo}.  
%One of the important applications of approximate PML in higher dimensions is in estimating symmetric properties for $\md$-distributions. A symmetric property is a function of $\md$-distributions that is invariant to a permutation of the labels. Here we study one such symmetric property for $\md=2$ called KL divergence that is studied in the setting of PML. In \Cref{theoremADOSuniversald}, we show that approximate PML is sample complexity optimal for estimating KL divergence. A similar result was already shown in \cite{Acharya18} (Theorem 6) for exact PML and we use the same proof with slight modification to prove our result. In \Cref{thm:effKL}, we give an efficient version of \Cref{theoremADOSuniversald} by combining it with \Cref{PMLtwo}. We next state two conditions under which our result holds. Let $B$ be such that, $\forall x \in \bX$, $\frac{\bpd(1)_{x}}{\bpd(2)_{x}} \leq B$. The two conditions are: $\bullet$ $\condo$ The estimation error $\epsilon$, satisfies $\epsilon > \frac{\log^3 N}{N}$. $\bullet$ $\condt$ $B \leq \epsilon^{2.24}N^{0.24}$. \sidford{I think I would just explicitly state these conditions inside the corollaries instead of defining them here; the conditions are short and standard things stated in theorems and corollaries.}

\begin{restatable}[Optimal sample complexity for KL divergence]{theorem}{theoremADOSuniversald}
\label{theoremADOSuniversald}
Let $B$ be such that, $\forall x \in \bX$, $\frac{\bpd(1)_{x}}{\bpd(2)_{x}} \leq B$ and let $\n=(\n(1),\n(2))$ be the optimal sample complexity for estimating KL divergence between $\bpd(1)$ and $\bpd(2)$ to an accuracy $\epsilon$. If~\footnote{Recall $N$ here is the size of domain $\bX$.} $\epsilon > \frac{\log^3 N}{N}$ and $B \leq \epsilon^{2.24}N^{0.24}$, then $\beta$-approximate PML $\md$-distribution (for $\md=2$) with $\beta>\exps{-\otilde\left(\n(1)^{4/5} + \n(2)^{4/5} \right)}$ is sample complexity optimal for estimating KL divergence to an accuracy $4\epsilon$.
%Suppose $\condo$ and $\condt$ hold, then $\beta$-approximate PML $\md$-distribution (for $\md=2$) with $\beta>\exps{-\otilde\left(\n(1)^{4/5} + \n(2)^{4/5} \right)}$ is sample complexity optimal \sidford{I think we need to be clearer what optimal means} for estimating KL divergence.
\end{restatable}
Theorem 6 in \cite{Acharya18} also requires $\epsilon > \frac{\log^3 N}{N}$ and a slightly weaker version of the other condition ($B^{3/2} \leq \epsilon^{0.99}N^{0.49}$). 

\begin{cor}[Efficient estimator for KL divergence]
\label{thm:effKL}
Let $B$ be such that, $\forall x \in \bX$, $\frac{\bpd(1)_{x}}{\bpd(2)_{x}} \leq B$ and let $\n=(\n(1),\n(2))$ be the optimal sample complexity for estimating KL divergence between $\bpd(1)$ and $\bpd(2)$ to an accuracy $\epsilon$. If $\epsilon > \frac{\log^3 N}{N}$ and $B \leq \epsilon^{2.24}N^{0.24}$, then there exists a PML based universal plug-in estimator that runs in $\otilde(\n(1)+\n(2)+ \n(1)^{3/5}\n(2)^{3/5} )$ time and is sample complexity optimal for estimating KL divergence to an accuracy $4\epsilon$.
%Suppose $\condo$ and $\condt$ hold, we give an efficient approximate PML based estimator that is sample complexity optimal for estimating KL divergence. 
\end{cor}

%% file: structure.tex
\newcommand{\Phidset}{\Phi_{\mathrm{discrete}}^{n}}
\newcommand{\disc}{\mathrm{disc}}

\section{Existence of Structured Approximate PML for One Dimension}\label{sec:structuremain}
Here we provide the proof for \Cref{thm:approxpml}. First, we show the existence of an approximate PML distribution with a nice structure in  Sections \ref{sec:prob_discrete}, \ref{sub:mult_discrete} and \ref{subsec:DPMLmain}. Then, we exploit this structure in \Cref{subsec:convexrelaxmain} to give an algorithm that returns a fractional solution with running time ranging from nearly linear to sub linear depending on the desired approximation factor. Finally, in \Cref{subsec:algmain} we present a rounding algorithm that takes the fractional solution from the previous step as input and returns an approximate PML distribution within the desired approximation factor.

First, we show the existence of a distribution with minimum non-zero probability value $\Omega(\frac{1}{n^2})$ that is an $\expo{-6}$-approximate PML distribution.
% In this section, we show the existence of an approximate PML distribution with a nice structure over the next several lemmas. We exploit this structure in \Cref{subsec:convexrelax} to give an algorithm with running time ranging from nearly linear to sub linear depending on the desired approximation factor. In this direction, we first show that one can assume the minimum non-zero probability of the PML distribution is $\Omega(\frac{1}{n^2})$ by only loosing a constant factor in the PML objective value. 

\begin{restatable}[{Minimum probability lemma}]{lemma}{lemmamin}
	\label{lemminmain}
	For any profile $\phi \in \Phi^{n}$, there exists a distribution $\bp'' \in \simplex$ such that $\bp''$ is a $\expo{-6}$-approximate PML distribution and $\min_{x \in \bX:\bp''_x \neq 0}\bp''_x \geq \frac{1}{2n^2}$.
\end{restatable} 
\begin{proof}
	See \Cref{app:minprob}.
\end{proof}

\begin{comment}
\[\min_{x \in \bX:\bp''_x \neq 0}\bp''_x \geq \frac{1}{2n^2}
\text{ and }
\probpml(\bp'',\phi) \geq \frac{1}{100} \probpml(\bp_{pml,\phi},\phi) ~.
\]
\end{comment}

This lemma allows us define a region in which our approximate PML takes all its probability values and we use this fact throughout the paper. In \Cref{sec:prob_discrete} and \Cref{sub:mult_discrete,} we show how we can further simplify the problem of computing an approximate PML by discretizing the probability and the frequency spaces respectively.

\subsection{Probability discretization}
\label{sec:prob_discrete}
Let $\bP\defeq\{(1+\epso)^{1-i}:i=1,\dots \boo\}$ where $\boo=O(\frac{\log n}{\epso})$ is such that $(1+\epso)^{1-\boo}\leq \frac{1}{2n^2}$ for some $\epso \in (0,1)$. $\bP$ is the set representing discretization of probability space and discretization introduces a technicality of probability values not summing up to one and we define pseudo-distributions and discrete pseudo-distribution to handle it.
\begin{defn}[Pseudo-distribution] $\bq \in [0,1]^{\bX}_{\R}$ is a \emph{pseudo-distribution} if $\|\bq\|_1 \leq 1$ and a \emph{discrete pseudo-distribution} if all its entries are in $\bP$ as well. We use $\psimplex$ and $\dsimplex$ to denote the set of all such pseudo-distributions respectively.
	\footnote{ As discussed in \Cref{sec:prelimsmain} we extend all functions of distributions as functions defined for any general vector in $\R^{\bX}$ and therefore to pseudo-distributions as well. For convenience we refer to $\probpml(\bq,\phi)$ for any pseudo-distribution $\bq$ as the ``probability'' of profile $\phi$ or PML objective value with respect to $\bq$.}
\end{defn}
%As discussed in \Cref{sec:prelims} we extend all functions taking distributions as input to functions defined for any general vector in $\R^{\bX}$ and therefore to pseudo-distributions as well. For convenience we refer to $\probpml(\bq,\phi)$ for any pseudo-distribution $\bq$ as the ``probability'' of profile $\phi$ or PML objective value with respect to $\bq$.

One of the important structural properties we prove here is the following: there exists a discrete pseudo-distribution $\bq$ that when converted to a distribution by dividing all its entries by its $\ell_1$ norm ($\frac{\bq}{\|\bq\|_1}$) is an approximate PML distribution. Even stronger, the discrete pseudo-distribution $\bq$ itself has $\probpml(\bq,\phi)$ value that approximates $\probpml(\bp_{pml,\phi},\phi)$ within a good factor and converting $\bq$ into a distribution by its $\ell_1$ norm is only going to help us in this probability because $\|\bq \|_1 \leq 1$. In the rest of the paper we refer to such a discrete pseudo-distribution as an approximate PML pseudo-distribution and for the earlier reason we focus on finding an approximate PML pseudo-distribution.

%We overload and use notations for discrete pseudo-distribution same as distributions.

The way we show the existence of such a discrete pseudo-distribution that is an approximate PML pseudo-distribution is by taking the PML distribution and converting it into a discrete pseudo-distribution while still preserving the PML objective value to a desired approximation factor. Our next lemma formally proves a general version of this statement. In the remainder of this paper, for notational convenience, for a scalar $c$ and set $\bS$ we use the notation $\floor{c}_{\bS}$ and $\ceil{c}_{\bS}$ to denote:\\
$$ \floor{c}_{\bS}\defeq\max_{s \in \bS: s \leq c}s \quad \text{ and } \quad \ceil{c}_{\bS}\defeq\min_{s \in \bS: s \geq c}s$$

\begin{defn}[Discrete pseudo-distribution] For any distribution $\bp \in \simplex$, its \emph{discrete pseudo-distribution} $\bq=\disc(\bp) \in \dsimplex$ is defined as:
	$$\bq_x\defeq\floor{\bp_x}_{\bP} \quad \forall x \in \bX$$
\end{defn}
Note that $\floor{\bp_x}_{\bP} \geq \frac{\bp_x}{1+\epso}$. Further, for $\bp \in \simplex$, $\frac{1}{1+\epso} \leq ||\disc(\bp)||_1 \leq 1$.
We next state a result that captures the impact of discretizing the probability space.

\begin{lemma}[{Probability discretization lemma}]\label{lem:probdisc}
	For any profile $\phi \in \Phi^{n}$ and distribution $\bp \in \simplex$, its discrete pseudo-distribution $\bq=\disc(\bp) \in \dsimplex$ satisfies:
	$$\probpml(\bp,\phi) \geq \probpml(\bq,\phi) \geq \expo{-\epso n}\probpml(\bp,\phi)$$
\end{lemma}
\begin{proof}
	The first inequality is immediate because $\bq_x=\floor{\bp_x}_{\bP} \leq \bp_{x}$ for all $x \in \bX$. To show second inequality consider any sequence $y^n \in \bX^n$,
	\begin{align*}
	\bbP(\bq,y^n) &= \prod_{x \in \bX}\bq_x^{\bff(y^n,x)}=\prod_{x \in \bX}\floor{\bp_x}_{\bP}^{\bff(y^n,x)} \geq \prod_{x \in \bX} \left( \frac{\bp_x}{1+\epso} \right)^{\bff(y^n,x)} = \frac{1}{(1+\epso)^n}\bbP(\bp,y^n)\\ 
	&\geq \expo{-\epso n}\bbP(\bp,y^n)
	\end{align*}
	In the inequality above we use $\sum_{x \in \bX}\bff(y^n,x)=n$. Now,
	$$ \probpml(\bq,\phi)=\sum_{\{y^n \in \bX^n: \Phi (y^n)=\phi \}} \bbP(\bq,y^n)\geq \sum_{\{y^n \in \bX^n: \Phi (y^n)=\phi \}} \expo{-\epso n}\bbP(\bp,y^n)=\expo{-\epso n}\probpml(\bp,\phi)$$
\end{proof}
%Our previous lemma showed that we can work with the discretized probability space $\bP$ by not losing much in the PML objective value and next we show a similar result for discretization of multiplicities.
\subsection{Multiplicity discretization}
\label{sub:mult_discrete}
Let $\bM=\{\ceil{(1+\epst/2)^1},\ceil{(1+\epst/2)^2},\dots , \ceil{(1+\epst/2)^{k-1}},n \} \cup \{1,2,3,\dots ,\ceil{\frac{1}{\epst}} \}$ be the set representing discretization of multiplicities where $k=O(\frac{\log n}{\epst})$ is such that $\ceil{(1+\epst/2)^{k}}\geq n$, $\ceil{(1+\epst/2)^{k-1}}< n$ and as before $\epst \in(0,1)$ will be carefully choose later. Let $\btt = |\bM|=O(\frac{\log n}{\epst})$ and note the definition of $\bM$ keeps all positive integers $\leq \ceil{\frac{1}{\epst}}$. We use $\nn_j$ to denote elements of set $\bM$ and using this set $\bM$ we define an analogous quantity to profile called \emph{discrete profile}.

\begin{defn}[Discrete profile]
	For a sequence $y^n \in \bX^{n}$, its \emph{discrete} profile $\phid=\Phid(y^n) \in \Z^{\btt}$ is a profile and is defined as: $\phid\defeq(\Fd_j)_{j=1\dots \btt} $, where $\Fd_j=\Fd_j(y^n)\defeq|\{x\in \bX ~|~\ceil{\bff(y^n,x)}_{\bM}=\nn_j \}|$ and $\nd=\sum_{x\in \bX}\ceil{\bff(y^n,x)}_{\bM}=\sum_{j=1}^{\btt} \nn_{j}\Fd_{j}$ is the length of discrete profile $\phid$ with $\nd \leq (1+\epst)n$. We use $\Phidset$  to denote the set of all such discrete profiles.
\end{defn}

\paragraph{Note:} As mentioned in the definition, a discrete profile is also a profile. Note that in the representation of discrete profile we might have indices $i$ with $\phid_{i}=0$,  however we have defined profiiles so that there are no such zero entries. We keep these zero entries in our discrete profile $\phid$ for notational convenience and proof simplification. Further it only takes $O(\phi_{size})$ time to write a discrete profile from access to a profile $\phi$ in the representation discussed in \Cref{subsec:representation}.

A discrete profile $\phid$ is a profile of length $\nd$ and it correspond to profile of some sequences of length $\nd$. One such sequence can be obtained by appending $\ceil{\bff(y^n,x)}_{\bM}-\bff(y^n,x)$ of $x$ symbols to sequence $y^n$ itself.
%Does $y^{\nd}$ refer to only sequences where $x$ appears $\ceil{\bff(y^n,x)}_{\bM}$ times? If so, why not say this explicitly?}
The probability of $\phid$ with respect to a distribution $\bp$ is straightforward:
$$ \probpml(\bp,\phid)\defeq\sum_{\{y^{\nd} \in \bX^{\nd}~|~ \Phi (y^{\nd})=\phid \}} \bbP(\bp,y^{\nd})$$
%In the formulation above we used function $\Phi (\cdot)$ on $y^{\nd}$ (and not $\Phid (\cdot)$ on $y^n$). 
We next state a result that captures the impact of discretizing the multiplicity space. It is important to note that probability terms ($\probpml(\bp,\phi)$ and $\probpml(\bp,\phid)$) have different summation terms and yet we show their values approximate each other.

\begin{restatable}[{Profile discretization lemma}]{lemma}{lemmaprofiledisc}
	\label{lem:profiledisc}
	For any distribution $\bp\in \simplex$, and a sequence $y^n \in \bX^{n}$:
	$$\expo{-7\epst n \log n}\probpml(\bp,\phi)\leq \probpml(\bp,\phid) \leq \expo{7\epst n \log n}\probpml(\bp,\phi)$$
	where $\phi=\Phi(y^n)$ and $\phid=\Phid(y^n)$ are {the} profile and discrete profile of $y^n$ respectively.
\end{restatable}
\begin{proof}
	See \Cref{app:profiledisc}.
\end{proof}

Combining both \Cref{lem:probdisc} and \Cref{lem:profiledisc} we bound the impact of discretizing both probabilities and multiplicities.
\begin{cor}[{Discretization lemma}]\label{cordiscmain}
	For any distribution $\bp\in \simplex$, and a sequence $y^n \in \bX^{n}$. If $\bq=\disc(\bp)$ is {the} discrete distribution of $\bp$ then,
	$$\expo{-(\discloss)}\probpml(\bp,\phi)\leq \probpml(\bq,\phid) \leq \expo{\discloss}\probpml(\bp,\phi)$$ 
	where $\phi=\Phi (y^n)$ and $\phid=\Phid(y^n)$ are the profile and discrete profile of $y^n$ respectively.
\end{cor}

% \begin{proof}
% Corollary follows immediately by combining \Cref{lem:probdisc} and \Cref{lem:profiledisc}.
% \end{proof}

The discretization lemma above suggests that optimizing over over discrete pseudo-distributions with $\phid$ as input is approximately as good as as optimizing over distributions with $\phi$ as input. This result motivates the definition of a new objective function which we introduce and study next.
\subsection{Discrete PML Optimization}\label{subsec:DPMLmain}
Here we define a new optimization problem that admits convex relaxations and further returns an approximate PML pseudo-distribution\footnote{Note we call a pseudo-distribution $\bq$ an approximate PML pseudo-distribution if it satisfies $\probpml(\bq,\phid) \geq \beta \probpml(\bp_{pml,\phi},\phi)$, for some non-trivial $\beta$.}. First, we define a discrete profile maximum likelihood (DPML) which is just the PML objective maximized over discrete pseudo-distributions with discrete profile as input. In \Cref{corfinaldiscmain} we show the optimal discrete pseudo-distribution of this new objective is an approximate PML pseudo-distribution. In \Cref{DPMLequalitymain}, we rephrase the DPML optimization problem. Finally, using this DPML reformulation, we define a new optimization problem that we call a single discrete PML (SDPML) and in \Cref{lemsinglediscmain}, we show the maximizing discrete pseudo-distribution for the SDPML objective is an approximate PML pseudo-distribution.
\begin{defn}[Discrete profile maximum likelihood]
	Let $y^n \in \bX^{n}$ be any sequence, $\phi=\Phi (y^n)$ and $\phid=\Phid(y^n)$ be its profile and discrete profile respectively, a \emph{discrete profile maximum likelihood (DPML) pseudo-distribution} $\bq_{dpml,\phid} \in \dsimplex$ is:
	\begin{equation}\label{eqdpml}
	\bq_{dpml,\phid}\defeq\argmax_{\bq \in \dsimplex} \probpml(\bq,\phid),
	\end{equation}
	and $\probpml(\bq_{dpml,\phid},\phid)$ is {the} maximum objective value.
\end{defn}
\begin{cor}[DPML is an approximate PML]\label{corfinaldiscmain}
	For any sequence $y^n \in \bX^{n}$ if $\phi=\Phi (y^n)$ and $\phid=\Phid(y^n)$ are its profile and discrete profile respectively, then $$\probpml(\bq_{dpml,\phid},\phid) \geq \expo{-(\discloss)} \probpml(\bp_{pml,\phi},\phi)$$
\end{cor}
\begin{proof}
	Note that $\bq_{pml,\phi}=\disc(\bp_{pml,\phi})$ is a discrete pseudo-distribution.
	The result follows from \Cref{cordiscmain} applied to $\bp_{pml,\phi}$.
\end{proof}
In a approximate sense, our \Cref{cordiscmain} suggests that working with discrete profile and discrete pseudo-distributions is no different than original profile and distribution itself.

In {the} next two lemmas we rephrase the DPML optimization problem in forms that are amenable to convex relaxation. To do this, we introduce some new notation.\\
\begin{itemize}
	\item As before {let} $\bP$ and $\bM$ be sets representing discretization of probabilities and frequencies respectively. Recall that we used $1=\nn_1 < \dots < \nn_j\dots < \nn_{\btt}$ to denote {the} elements of set $\bM$ and we use $\pp_1 < \dots <\pp_i \dots  <\pp_{\boo}$ to denote {the} elements of set $\bP$. Let $\pvec \in \R^{\boo}$ be the vector with elements indexed from $1$ to $\boo$ and $i$th element equal to $\pp_i$. Also let $\mvec \in \R^{\bttpo}$ be the vector with elements indexed from $0$ to $\btt$. Its zeroth entry (denoted by $\nn_{0}$) is equal to $0$ and $j$th entry is equal to $\nn_j \in \bM$.
	\item Let $\bk \in \Z^{\boo \times \bttpo}$ be a variable matrix with entries $X_{ij}$ for $i \in \otboo,j \in \ztbtt$.
	%and we use $\bk_{ij}$ for $i \in \otboo,j \in \ztbtt$ to denote elements of this matrix. 
	As in the case for vector $\nn$, our second index $j$ of variable matrix $\bk$ starts at $0$ and not at $1$. Here the variable $\bk_{ij}$ counts the number of domain symbols $x \in \bX$ with probability value $\pp_i$ and frequency $\nn_j$. Further, $\bk_{i,0}$ counts %is basically counting 
	the number of unseen domain symbols $x \in \bX$ with probability value $\pp_i$.
	\item For any vector $\vvec$ and set $S$, we use $\vvec_{S}$ to denote the $|S|$ length vector corresponding to the portion of vector $\vvec$ associated with index set $S$.
	%<<<<<<< HEAD
	\item For a discrete profile $\phid=(\Fd_{j})_{j=1\dots \btt}$ (corresponding to sequence $y^n$), define\\
	$~~~~\bK_{\phid}\defeq \{\bk \in \Z^{\boo \times \bttpo}~\Big|~ ~(\bk^{T}\onevec)_{\otbtt} = \Fd, \text{ and } \pvec^{T} \bk \onevec \leq 1 \}$\\
	Note the constraint $(\bk^{T}\onevec)_{\otbtt} = \Fd$ does not involve $\bk_{0,j}$ variables that corresponds to unseen elements. These variables only appear in the constraint $\pvec^{T} \bk \onevec \leq 1$ which ensures our output is always a pseudo-distribution.
	\item For a discrete profile $\phid=(\Fd_{j})_{j=1\dots \btt}$ (of $y^n$) and a discrete pseudo-distribution $\bq$, also define\\
	$~~~~\bK_{\bq,\phid}\defeq \{\bk \in \Z^{\boo \times \bttpo}~\Big|~ ~(\bk^{T}\onevec)_{\otbtt} = \Fd, \text{ and } \bk \onevec=\levelq \}$
	where $\levelq \in \R^{\boo}$ and $\levelq_i$ denote the number of domain elements with probability value $\pp_i \in \bP$ in pseudo-distribution $\bq$. It will be clear from our next lemma why we define these constraint sets.
\end{itemize}
The advantage of probability and profile discretization we described earlier is that many types in the set $\{\phih ~|~ \Phi(\phih)=\phid\}$ share the same probability value of being observed and our goal is to group them using these $\bk_{ij}$ variables. Exploiting this idea, we next give a different formulation for {the} DPML objective.
\begin{lemma}[DPML objective reformulation]\label{DPMLequalitymain}
	For any discrete pseudo-distribution $\bq \in \simplex$ and discrete profile $\phid \in \Phidset$:
	\begin{equation}\label{eq:formulation2main}
	\probpml(\bq,\phid) = \cphid \sum_{\bk\in \bK_{\bq,\phid}} \prod_{i=1}^{\boo}\left(\pp_i^{(\bk \mvec)_i}\frac{(\bk\onevec)_i!}{\prod_{j=0}^{\btt}\bk_{ij}!}\right)
	\end{equation}
\end{lemma}
\begin{proof}
	Recall from \Cref{eqlabeled},
	$$\probpml(\bq,\phid) = \cphid \sum_{\{\phih ~|~ \Phi(\phih)=\phid\}} \prod_{x \in X}\bq_x^{\phih_x}~.$$
	For convenience, we call a type $\lphi$ valid if it belongs to set $\{\phih ~|~ \Phi(\phih)=\phid\}$. Recall that variable $\bk_{ij}$ represents the number of domain elements with probability value $\pp_i$ and frequency $\nn_j$. In this representation and for the discrete pseudo-distribution $\bq$, each valid type $\lphi$ corresponds to the following unique variable assignment $\bk \in \bK_{\bq,\phid}$:
	$\bk_{ij}=|\{x \in \bX ~| ~\bq_x=\pp_i \text { and } \lphi_x=\nn_j\}|$. Using the previous expression it is not hard to write the exact expression for the probability term associated with the valid type $\lphi$,
	\begin{equation}\label{eq:probtype}
	\prod_{x \in X}\bq_x^{\phih_x}=\prod_{i=1}^{\boo}\prod_{j=0}^{\btt}\pp_{i}^{\bk_{ij}\nn_{j}}=\prod_{i=1}^{\boo}\pp_{i}^{(\bk \mvec)_i}
	\end{equation}
	Previous discussion showed that every valid type corresponds to a unique variable assignment. However this uniqueness property no more holds in the reverse direction and multiple valid types might share the same variable assignment. This where our grouping occurs and is an interesting case that we study next.
	
	For any variable assignment $\bk$, it is clear from the middle term in Equation \ref{eq:probtype} that all valid types $\lphi$ associated with $\bk$ share the same probability value of being observed. With this observation, it is now enough to argue about the number of valid types associated with a variable assignment $\bk$ to prove our lemma. We make this argument next by constructing all valid types associated with $\bk$.
	
	First consider all domain elements with a fixed probability value $\pp_{i}$ and the number of these elements is equal to $\sum_{j=0}^{\btt}\bk_{ij}$. We can generate part of a valid type corresponding to probability value $\pp_{i}$ by picking any partition of these $\sum_{j=0}^{\btt}\bk_{ij}$ domain elements into groups of sizes $\{ \bk_{ij}\}_{j \in \ztbtt}$. This corresponds to a multinomial coefficient and the number of types associated with $\bk$ is just,
$$\frac{(\bk \onevec)_i!}{\prod_{j=0}^{\btt}\bk_{ij}!}~.$$
	Here we only generated partial valid types corresponding to probability value $\pp_{i}$. To generate a full valid type we just need to combine these partial valid types generated for each probability value $\pp_{i}$. Let $S_{\bk}$ denote all such full valid types associated with a variable assignment $\bk$ and generating a full valid type corresponds to groups (for each probability value $\pp_{i}$) of independent possibilities considered conjointly. Further the cardinality of set $S_{\bk}$ is just the multiplication of cardinalities of each of these groups and is explicitly written below,
	$$|S_{\bk}|=\prod_{i=1}^{\boo}\frac{(\bk\onevec)_i!}{\prod_{j=0}^{\btt}\bk_{ij}!}~.$$
	We are almost done with the proof and all we do next is formally derive the expression in our lemma statement to complete the proof.  From \Cref{eqlabeled},
	\begin{align*}
	\probpml(\bq,\phid) & =\cphid \sum_{\{\phih | \Phi(\phih)=\phid\}} \prod_{x \in X}\bq_x^{\phih_x}
	=\cphid \sum_{\bk\in \bK_{\bq,\phid}} \sum_{\{\phih \in S_{\bk}\}} \prod_{x \in X}\bq_x^{\phih_x}\\
	&=\cphid \sum_{\bk\in \bK_{\bq,\phid}} \sum_{\{\phih \in S_{\bk}\}} \prod_{i=1}^{\boo}\pp_{i}^{(\bk \mvec)_i}
	=\cphid \sum_{\bk\in \bK_{\bq,\phid}}  |S_{\bk}| \prod_{i=1}^{\boo}\pp_{i}^{(\bk \mvec)_i} \\
	& = \cphid \sum_{\bk\in \bK_{\bq,\phid}} \prod_{i=1}^{\boo}\left(\pp_i^{(\bk \mvec)_i}\frac{(\bk\onevec)_i!}{\prod_{j=0}^{\btt}\bk_{ij}!}\right)
	\end{align*}
\end{proof}

In the lemma above we wrote the $\probpml(\bq,\phid)$ in terms of constraint set $\bK_{\bq,\phid}$ and to use this definition we need access to pseudo-distribution $\bq$. We overcome this difficulty in our next lemma by giving an inequality that relates $\probpml(\bq,\phid)$ with constraint set $\bK_{\phid}$ that only depends on $\phid$ and not $\bq$ itself. 
\begin{lemma}[DPML objective relaxed]\label{lem:DPMLupperboundmain}
	For any sequence $y^n \in \bX^{n}$, and a discrete pseudo-distribution $\bq \in \simplex$ the DPML objective can be upper bounded by: 
	\begin{equation}\label{eq:w1main}
	\probpml(\bq,\phid) \leq \cphid \sum_{\bk\in \bK_{\phid}} \prod_{i=1}^{\boo}\left(\pp_i^{(\bk \mvec)_i}\frac{(\bk\onevec)_i!}{\prod_{j=0}^{\btt}\bk_{ij}!}\right)
	\end{equation}
	where $\phid = \Phid(y^n) \in \Phidset$ is discrete profile of $y^n$.
\end{lemma}
\begin{proof}
	The proof follows because $ \bK_{\bq,\phid}\subseteq \bK_{\phid} $ and invoking \Cref{DPMLequalitymain}.
\end{proof}
In the above lemma we only showed one side of the inequality and it not clear how working with RHS relates to the LHS. Inf \Cref{subsec:algmain} we present an algorithm to achieve the other side of the inequality. The cardinality of set $\bK_{\phid}$ in the above formulation is small and we formalize this next.
\begin{lemma}[Cardinality of $\bK_{\phid}$]\label{lemcardmain}
	For any sequence $y^n \in \bX^{n}$ and its associated discrete profile $\phid = \Phid(y^n)$: $$|\bK_{\phid}| \leq  \expo{(\boo \times \btt)O(\log n)}.$$
\end{lemma}
\begin{proof}
	$\bK_{\phid}$ is a set of vectors in $\Z^{\boo \times \bttpo}$ and each coordinate takes an integer value in $[0,2n^2]$ (\Cref{lemminmain} combined with the constraint $\pvec^{T}\bk \onevec\leq 1$ ensures this fact). The lemma statement follows because $\bK_{\phid} \leq (2n^2)^{\boo \bttpo} \in \expo{(\boo \times \btt)O(\log n)}$.
\end{proof}
In our final optimization problem we just optimize over one term in the set $\bK_{\phid}$ instead of working with summation over all the terms. Focusing on the largest of these terms, gives a $1/|\bK_{\phid}|$ approximation of the sum. Combining this with \Cref{lemcardmain} motivates us to consider the following objective, define:
%   In our final optimization problem we just optimize over one term in the set $\bK_{\phid}$ instead of working with summation over all the terms and next two lemmas serve as the motivation for working with single term over the summation of terms by showing that the optimizing pseudo-distribution of our final optimization problem is still an approximate PML pseudo-distribution.
%  
%  
%Recall in the \Cref{eq:w1main} we wrote an inequality relating $\probpml(\bq,\phid)$ with sum of terms, where each term is a function of $\bk\in \bK_{\phid}$ with a common multiplier that is independent of $\bk$. We focus on the largest of these terms, which is a $1/|\bK_{\phid}|$ approximation of the sum. Combining this with \Cref{lemcard} motivates us to consider the following objective, define:
$$\probsdpml(\bk)\defeq  \prod_{i=1}^{\boo}\left(\pp_i^{(\bk \mvec)_i}\frac{(\bk \onevec)_i!}{\prod_{j=0}^{\btt}\bk_{ij}!}\right)$$
%$$\max_{\bk\in \bK_{\phid}}\binom{\nd}{\hat{\phid}} \probsdpml(\bk)$$
It is important to note that there is a discrete $\md$-pseudodistribution $\bq_{\bk}$ that correspond to each variable assignment $\bk\in \bK_{\phid}$. The description of this distribution is as follows: For each $i \in \otboo$, the number of domain elements with probability value $\pp_{i}$ in $\bq$ is equal to $(\bk \onevec)_{i}$ \footnote{This description only provides non zero probability values and also does not provide any labels, however it is sufficient for estimating all symmetric properties mentioned in this paper.}. We now go ahead and define the optimization problem involving $\probsdpml(\bk)$ that also help us compute the term that is largest in the summation of terms in \Cref{eq:w1main}. After this definition, we provide a lemma relating the PML objective with this new optimization problem.
\begin{defn}[Single discrete profile maximum likelihood]
	For any sequence $y^n \in \bX^{n}$ and its associated discrete profile $\phid = \Phid(y^n)$, a \emph{single discrete profile maximum likelihood} (SDPML) distribution $\bq_{sdpml,\phid}$ is:
	\begin{equation}\label{eq:sdpml}
	\bk_{sdpml,\phid}\defeq \argmax_{\bk\in \bK_{\phid}}\cphid \probsdpml(\bk)=\argmax_{\bk\in \bK_{\phid}}\probsdpml(\bk)
	\end{equation}
	and $\bq_{sdpml,\phid}$ is the pseudo-distribution corresponding to $\bk_{sdpml,\phid}$.
\end{defn}

\begin{lemma}[SDPML relation to PML]\label{lemsinglediscmain}
	For any sequence $y^n \in \bX^{n}$,
	$$\Cphid\probsdpml(\bk_{sdpml,\phid}) \geq \expo{-O(\odiscloss +\grouploss) }\probpml(\bp_{pml,\phi},\phi)$$
	where $\phi = \Phi(y^n)$ and $\phid = \Phid(y^n)$ are the profile and discrete profile associated with $y^n$.
\end{lemma}
\begin{proof} $~~~~~\Cphid \probsdpml(\bk_{sdpml,\phid}) \geq \Cphid \probsdpml(\bk_{dpml,\phid}) \geq \expo{-(\boo \times \btt)\log n}\probpml(\bq_{dpml,\phid},\phid)~~~~~$\\
	$$\geq  \expo{-O(\odiscloss +\grouploss)}\probpml(\bp_{pml,\phi},\phi)$$
	The second inequality follows from \Cref{lemcardmain}, \ref{lem:DPMLupperboundmain} and last follows from \Cref{corfinaldiscmain}.
\end{proof}

To simplify and better understand the expression in \Cref{lemsinglediscmain} just substitute $\epso=\epst=\frac{1}{n^{1/3}}$ and note that $\bk_{sdpml,\phid}\in \bK_{\bq_{sdpml},\phid}$, and $\probsdpml(\bk_{sdpml,\phid})$ is just one term in the summation of terms in \Cref{eq:formulation2main}. Using \Cref{DPMLequalitymain} we know that $\Cphid\probsdpml(\bk_{sdpml,\phid}) \leq \probpml(\bq_{sdpml,\phid},\phid)$ and combining this with previous lemma we get that the discrete pseudo-distribution $\bq_{sdpml,\phid}$ is an $\exp(-\widetilde{O}(n^{2/3}))$-approximate PML pseudo-distribution. All we do next is provide a convex relaxation for function $\probsdpml(\bk)$ to  arrive at our final optimization problem. This relaxation produces a real valued $\bk$ and later we give a rounding algorithm to get an integral solution.

%To simplify and better understand the expression in \Cref{lemsingledisc} just substitute $\epso=\epst=\frac{1}{n^{1/3}}$ and observe that $\bk_{sdpml,\phid}\in \bK_{\bq_{sdpml},\phid}$, and $\probsdpml(\bk_{sdpml,\phid})$ is just one term in the summation of terms in \Cref{eq:formulation2}. Using \Cref{DPMLequality} we know that $\Cphid\probsdpml(\bk_{sdpml,\phid}) \leq \probpml(\bq_{sdpml,\phid},\phid)$ and combining this with previous lemma we get that the discrete pseudo-distribution $\bq_{sdpml,\phid}$ is an $\expo{-\tilde{O}(n^{2/3})}$-approximate PML pseudo-distribution. We are almost done \sidford{with what? :P} and all we do next is provide a convex relaxation for function $\probsdpml(\bk)$ and we will arrive at our final optimization problem. This relaxation produces a real valued $\bk$ and later we give a rounding algorithm to get an integral solution.

\subsection{Convex relaxation of SDPML}\label{subsec:convexrelaxmain}
In the previous subsection we showed that the SDPML objective is a good approximation to the PML objective. However the objective function of SDPML is defined only over the integers and in this subsection we present a convex relaxation of SDPML. 

First, we consider the feasible set $\bK_{\phid}$ of SDPML and relax the integer constraint on variables $\bk_{ij}$ to get the following new constraint set:
\begin{equation}\label{setrelax}
\bK^{f}_{\phid}\defeq\{\bk \in \R^{\boo \times \bttpo}~\big|~ (\bk^{T}\onevec)_{\otbtt} = \Fd, \text{ and } \pvec^{T} \bk \onevec \leq  1 \} ~.
\end{equation}
%\[
%\bK_{\phid}=\{\bk \in \Z^{\boo \times \bttpo}~\Big|~ ~(\bk^{T}\onevec)_{\otbtt} = \Fd, \text{ and } \pvec^{T} \bk \onevec \leq 1 \} ~.
%\]
%
%%Instead we work with the following natural relaxation:
%We relax the integer constraint on variables $\bk_{ij}$:
%\begin{equation}\label{setrelax}
%\bK^{f}_{\phid}\defeq\{\bk \in \R^{\boo \times \bttpo}~\big|~ (\bk^{T}\onevec)_{\otbtt} = \Fd, \text{ and } \pvec^{T} \bk \onevec \leq  1 \} ~.
In the later subsections, we show how we deal with these fractional solutions by presenting a rounding algorithm with a good approximation ratio.

Secondly, we relax the objective function of SDPML itself. The objective of SDPML is defined only on the integral set. We next define a continuous relaxation of this objective function which is also log-concave.
\begin{equation}\label{contsdpmlorigmain}
\begin{split}
\bg(\bk)&\defeq \prod_{i=1}^{\boo}\left(\pp_i^{(\bk \mvec)_i}\frac{\expo{(\bk \onevec)_i \log (\bk \onevec)_i-(\bk \onevec)_i}}{\prod_{j=0}^{\btt}\expo{\bk_{ij}\log \bk_{ij}-\bk_{ij}}}\right)\\
& =\expo{\log(\pvec)^{T}\bk \mvec + \sum_{i=1}^{\boo}(\bk \onevec)_i \log (\bk \onevec)_i-\sum_{i=1}^{\boo}\sum_{j=0}^{\btt}\bk_{ij}\log \bk_{ij}}
\end{split}
\end{equation}
The lemma below states that continuous version is not far from the actual SDPML objective. 
\begin{restatable}[$\bg(\cdot)$ approximates SDPML objective]{lemma}{lemmastirlingsapproxmain}
	\label{app:sterlingsapproxmain}
	For any sequence $y^n \in \bX^{n}$ and its associated discrete profile $\phid = \Phid(y^n)$. If $\bk \in \bK_{\phid}$, then
	$$\expo{-O(\log n) \boo \btt} \bg(\bk) \leq \probsdpml(\bk)\leq \expo{O(\log n) \boo } \bg(\bk)$$
\end{restatable}
\begin{proof}
	See \Cref{app:structuremain}.
\end{proof}

A key fact about function $\bg(\bk)$ is that it is log-concave, so we can apply optimization machinery from convex optimization to optimize it.

\begin{restatable}{lemma}{lemmalogconcave}
	Function $\bg(\bk)$ is log-concave in $\bk$.
\end{restatable}
\begin{proof}
	See \Cref{app:structuremain}.
\end{proof}

Maximizing log concave objective function $\bg(\cdot)$ over the relaxed convex set $\bK^{f}_{\phid}$ easily reduces to a convex optimization problem and can be solved efficiently. Below is the convex relaxation of our SDPML objective,
\begin{equation}\label{eq:convoptoriginalmain}\argmin_{\bk\in \bK^{f}_{\phid}}-\log \bg(\bk)~.
\end{equation}
Formulation above is in the form of a general optimization problem $(11.14 )$ in \cite{LSW15} that solves it using a cutting plane method. The algorithm in \cite{LSW15} requires to implement a $\del$-2nd-order-optimization oracle (defined later in the appendix) and we provide an algorithm to implement this $\del$-2nd-order-optimization oracle for our convex program. Further, to upper bound the number of calls to such an oracle we need to bound the singular values of our constraint matrix. Everything put together we get the following theorem.
\begin{restatable}[Solver for convex relaxation to SDPML]{theorem}{theoremcuttingplanemain}
	\label{thm:cutting_plane_use}
	There exists a cutting plane method based algorithm that outputs a feasible solution $\bk'$ to optimization problem \ref{eq:convoptoriginalmain}, i.e. $\bk' \in \bK^{f}_{\phid}$ and satisfies:
	$$-\log \bg(\bk) \leq \argmin_{\bk\in \bK^{f}_{\phid}}-\log \bg(\bk)+\delta$$
	in $\runningtimedelta$ time.
\end{restatable}
\begin{proof}
	See \Cref{app:cuttingplane}.
\end{proof}

\subsection{Algorithm and runtime analysis}\label{subsec:algmain}
Here we give the complete description of our final algorithm to find an approximate PML distribution. The analysis in previous sections suggests that it suffices to find a discrete pseudo-distribution that approximates SDPML objective, which we replaced by a convex relaxation. First, we give the complete algorithm. Then, we present the algorithm that takes an optimal solution to the convex proxy for SDPML and produces an approximate PML distribution.
Recall that $\bK^{f}_{\phid}\defeq\{\bk \in \R^{\boo \times \bttpo}~\big|~ (\bk^{T}\onevec)_{\otbtt} = \phid, \text{ and } \pvec^{T} \bk \onevec \leq  1 \}$.
\begin{algorithm}[H]
	\caption{Algorithm for approximate PML}\label{euclid}
	\begin{algorithmic}[1]
		\Procedure{Approximate PML}{}
		\State Solve $\bk'=\argmax_{\bk\in \bK^{f}_{\phid}}\bg(\bk)$. 
		\State Round fractional solution $\bk'$ to integral solution $\bk \in \bK_{\phid}$.
		\State Construct discrete pseudo-distribution $\bq_{X}$ corresponding to $\bk$.
		\State \textbf{return} $\frac{\bq_{X}}{\|\bq_{X}\|_1}$
		\EndProcedure
	\end{algorithmic}
\end{algorithm}
%Note: 
In the algorithm we first maximize over the set of fractional solutions $\bK^{f}_{\phid}$ instead of  $\bK_{\phid}$ and we {\em round} our solution $\bk'$ to an integral solution $\bk$ that belongs to extended set of $\bK_{\phid}$. The rounding algorithms is presented next.
\begin{comment}
The fractional solution $\bk'$ represents, for any $i \in [\boo], j \in [\btt]$, $\bk'_{ij}$ number of domain elements have probability $\pp_i$ and frequency $\nn_j$ in the observed sequence (after discretization).
Our rounding algorithm (below) proceeds by eliminating the fractional parts of $\bk'_{ij}$, $i \in [\boo]$ %($j$ is fixed) 
and assign them to a new level set with associated probability $\pp_{\boo+j}$ chosen so that the total probability mass is conserved.
\end{comment}
%Given an optimal fractional solution $\bk'$, we produce an integer solution $\bk$ via the following algorithm:
%to the following optimization problem:
%$$\bk'=\argmax_{\bk\in \bK^{f}_{\phi}}\bg(\bk)$$
%Below is our algorithm to round to the integral solution $\bk$:
\begin{algorithm}[H]
	\caption{Rounding algorithm}
	\begin{algorithmic}[1]
		\Procedure{Rounding}{$\bk'$}
		\State Define $\bk=\mzero^{(\boo+\btt)\times \bttpo}$.
		\State $\bk_{ij}= \floor{\bk'_{ij}} \in \Z \quad \forall i \in \otboo,j \in \ztbtt$ %\Comment{$\bk \notin \bK_{\phid}$ and we fix it next.}
		%\State $\br_j=\sum_{i=1}^{\boo}\bk'_{ij}-\sum_{i=1}^{\boo}\bk_{ij}=\F_{\nn_j}-\sum_{i=1}^{\boo}\bk_{ij} \in \Z \quad \forall j \in \otbtt$\Comment{$\br_j <\boo$}
		%\State $\bs_i=\sum_{j=0}^{\btt}\bk'_{ij}-\sum_{j=0}^{\btt}\bk_{ij} \quad \forall i \in \otboo$\Comment{$\bs_i < \btt$}.
		\For{$j \in \otbtt$}
		\State Create a new level set with probability value $\pvec_{\boo+j}=\frac{\sum_{i=1}^{\boo}(\bk'_{ij}-\bk_{ij})\pvec_{i}}{\sum_{i=1}^{\boo}(\bk'_{ij}-\bk_{ij})}$.
		\State Assign $\bk_{\boo+j,j}=\sum_{i=1}^{\boo}(\bk'_{ij}-\bk_{ij})=\Fd_{j}-\sum_{i=1}^{\boo}\bk_{ij} \in \Z$.
		\EndFor
		\State \textbf{return} $\bk$
		\EndProcedure
	\end{algorithmic}
	\label{alg:rounding}
\end{algorithm} 
\newcommand{\pvecext}{\pvec_{ext}}

The solution $\bk$ returned by the rounding procedure is defined on an extended discretized probability space $\bP'$, where $\bP' \defeq  \bP \cup \{\pvec_{\boo+j} \}_{j \in \otbtt}$. To derive the relation between solution $\bk$ and PML objective value we need to extend some definitions studied earlier. First, we define $\pvecext$ as the vector whose entries are exactly the elements of $\bP'$. Note we still use $\pvec_{i}$ for all $i\in [1,\boo+\btt]$ to refer to elements of $\pvecext$. Further, for any pseudo-distribution $\bq$ with all its probability values in set $\bP'$ (we call it an \emph{extended discrete pseudo-distribution)} and discrete profile $\phid$, we first define following extensions of sets $\bK_{\bq,\phid}$ and $\bK_{\phid}$,
$$\bK^{ext}_{\bq,\phid}\defeq \{\bk \in \Z^{(\boo+\btt) \times \bttpo}~\Big|~ ~(\bk^{T}\onevec)_{\otbtt} = \Fd, \text{ and } \bk \onevec=\levelq \}~,$$
$$\bK^{ext}_{\phid}\defeq \{\bk \in \Z^{(\boo+\btt) \times \bttpo}~\Big|~ ~(\bk^{T}\onevec)_{\otbtt} = \Fd, \text{ and } \pvecext^{T} \bk \onevec \leq  1 \}~,$$
where $\levelq \in \R^{\boo+\btt}$ and $\levelq_i$ denote the number of domain elements with probability value $\pp_i \in \bP'$. 

Further by \Cref{DPMLequalitymain}, for any extended discrete pseudo-distribution $\bq$ and a discrete profile $\phid$, the following equality holds,
\begin{equation}\label{eq:formulation3}
\probpml(\bq,\phid) = \cphid \sum_{\bk\in \bK^{ext}_{\bq,\phid}} \prod_{i=1}^{\boo+\btt}\left(\pp_i^{(\bk \mvec)_i}\frac{(\bk\onevec)_i!}{\prod_{j=0}^{\btt}\bk_{ij}!}\right)
\end{equation}
Similarly for any $\bk \in \bK^{ext}_{\bq,\phid}$, below are the natural extension of definitions of functions $\probsdpml(\cdot)$ and $\bg(\cdot)$,
$$\probsdpml(\bk)\defeq \prod_{i=1}^{\boo+\btt}\left(\pp_i^{(\bk \mvec)_i}\frac{(\bk \onevec)_i!}{\prod_{j=0}^{\btt}\bk_{ij}!}\right)$$
$$ \bg(\bk)\defeq\prod_{i=1}^{\boo+\btt}\left(\pp_i^{(\bk \mvec)_i}\frac{\expo{(\bk \onevec)_i \log (\bk \onevec)_i-(\bk \onevec)_i}}{\prod_{j=0}^{\btt}\expo{\bk_{ij}\log \bk_{ij}-\bk_{ij}}}\right)$$

We are ready to analyze our rounding algorithm. First, we provide some interesting properties solution $\bk$ returned by our rounding procedure.

\begin{clm}\label{clmrounding} The solution $\bk \in \Z^{(\boo+\btt) \times \bttpo}$ returned by rounding procedure (\ref{alg:rounding}) above satisfies:
	\begin{enumerate}
		%\item $\bk'_{ij}-1\leq \bk_{ij} \leq \bk'_{ij}+1 \quad \forall  i \in \otboo \text{ and } j \in \ztbtt$.
		\item $(\bk' \onevec)_{i}-\bttpo \leq (\bk \onevec)_{i} \leq (\bk' \onevec)_{i} \quad \forall i \in \otboo$
		\item $\bk \in \bK^{ext}_{\phid}$.
	\end{enumerate}
\end{clm}
\begin{proof}
	Claim (1) follows because $\bk'_{ij}-1\leq \bk_{ij} \leq \bk'_{ij}$ for all $i \in \otboo,j \in \ztbtt$. Now note $\sum_{i=1}^{\boo+\btt}\bk_{ij}=\sum_{i=1}^{\boo}\bk'_{ij}=\Fd_{\nn_j} \quad \forall j \in \otbtt$ because of the adjustments made by new level sets. Further, 
	\begin{align*}
	\pvecext^{T} \bk \onevec=\sum_{i=1}^{\boo+\btt}\pvec_{i}(\bk \onevec)_i & = \sum_{i=1}^{\boo}\pvec_{i}(\bk \onevec)_i+\sum_{j=1}^{\btt}\pvec_{\boo+j}(\bk \onevec)_{\boo+j}\\
	&=\sum_{i=1}^{\boo}\pvec_{i}(\bk \onevec)_i+\sum_{j=1}^{\btt}\sum_{i=1}^{\boo}(\bk_{ij}'-\bk_{ij})\pvec_i\\
	&=\sum_{j=1}^{\btt}\sum_{i=1}^{\boo}\bk_{ij}'\pvec_i =\pvec^{T} \bk' \onevec \leq  1
	\end{align*}
	The final inequality follows because $\bk' \in \bK^{f}_{\phid}$ and therefore $\bk \in \bK^{ext}_{\phid}$ and Claim (2) follows.
\end{proof}

We next show that for any solution $\bk$ returned by our rounding algorithm (\ref{alg:rounding}), the values $\probsdpml(\bk)$ and $\bg(\bk)$ are close to each other and we summarize this next.
\begin{restatable}{lemma}{lemroundster} 
	\label{app:roundster}
	For any $\bk \in \bK^{ext}_{\phid}$ returned by rounding procedure above satisfies:
	\begin{equation}\label{relation}
	\expo{-O(\log n) \boo \btt} \bg(\bk) \leq \probsdpml(\bk)\leq \expo{O(\log n) (\boo+\btt) } \bg(\bk)
	\end{equation}
\end{restatable}
\begin{proof}
	See \Cref{app:structuremain}.
\end{proof}

Further using \Cref{eq:formulation3}, for any $\bk \in \bK^{ext}_{\phid}$, if $\bq_{\bk}$ is its corresponding extended discrete pseudo-distribution, then
\begin{equation}\label{otherside}
\probpml(\bq_{\bk},\phid) \geq \Cphid \probsdpml(\bk)
\end{equation}

In our next lemma, we show that the solution $\bk \in \bK^{ext}_{\phid}$ returned by the rounding procedure approximates $\probsdpml(\bk_{sdpml})$. Note from \Cref{lemsinglediscmain}, we know that $\probsdpml(\bk_{sdpml})$ is a good approximation to the PML objective.

\begin{lemma}\label{lem:roundedsoln}
	The solution $\bk \in \bK^{ext}_{\phid}$ returned by rounding procedure above satisfies:
	$$\probsdpml(\bk) \geq \extrel\probsdpml(\bk_{sdpml})$$
\end{lemma}
\begin{proof}
	For any $\bk' \in \bK^{f}_{\phid}$ and $\bk \in \bK^{ext}_{\phid}$ returned by our rounding procedure below are the explicit expressions for $\bg(\bk)$ and $\bg(\bk')$: 
	$$ \bg(\bk)=\left(\prod_{i=1}^{\boo}\pp_i^{(\bk \mvec)_i}\frac{\expo{(\bk \onevec)_i \log (\bk \onevec)_i}}{\prod_{j=0}^{\btt}\expo{\bk_{ij}\log \bk_{ij}}}\right)\left(\prod_{j=1}^{\btt} \pp_{\boo+j}^{\nn_j\bk_{\boo+j,j} }\cdot 1\right)$$
	$$ \bg(\bk')=\prod_{i=1}^{\boo}\left(\pp_i^{(\bk' \mvec)_i}\frac{\expo{(\bk ' \onevec)_i \log (\bk ' \onevec)_i}}{\prod_{j=0}^{\btt}\expo{\bk'_{ij}\log \bk'_{ij}}}\right)$$
	We first bound the probability term:
	\begin{equation}\label{eq:probterm}
	\begin{split}
	\prod_{i=1}^{\boo}&\pp_i^{(\bk' \mvec)_{i}}=\left(\prod_{i=1}^{\boo}\pp_i^{(\bk \mvec)_i}\right)\left(\prod_{i=1}^{\boo}\pp_i^{\sum_{j=1}^{\btt}\nn_j(\bk'_{ij}-\bk_{ij})}\right)\\
	&=\left(\prod_{i=1}^{\boo}\pp_i^{(\bk \mvec)_i}\right)\left(\prod_{j=1}^{\btt} \prod_{i=1}^{\boo}\pp_i^{\nn_j(\bk'_{ij}-\bk_{ij})}\right)\\
	&=\left(\prod_{i=1}^{\boo}\pp_i^{(\bk \mvec)_i}\right)\left(\prod_{j=1}^{\btt} \left(\prod_{i=1}^{\boo}\pp_i^{(\bk'_{ij}-\bk_{ij})}\right)^{\nn_j}\right)\\
	&\leq \left(\prod_{i=1}^{\boo}\pp_i^{(\bk \mvec)_i}\right)\left(\prod_{j=1}^{\btt} \left(\frac{\sum_{i=1}^{\boo}\pp_i(\bk'_{ij}-\bk_{ij})}{\sum_{i=1}^{\boo}(\bk'_{ij}-\bk_{ij})}\right)^{\nn_j\sum_{i=1}^{\boo}(\bk'_{ij}-\bk_{ij})}\right)\\
	&\leq \left(\prod_{i=1}^{\boo}\pp_i^{(\bk \mvec)_i}\right)\left(\prod_{j=1}^{\btt} \pp_{\boo+j}^{\nn_j\bk_{\boo+j,j}}\right)\\
	\end{split}
	\end{equation}
	The first inequality follows because $\mvec_{0}=0$. The fourth inequality follows from AM-GM inequality. The final expression above is the probability term associated with $\bk$ and the equation above shows that our rounding procedure only increases the probability term and all that matters is to bound the counting term that we do next.
	\begin{equation}\label{eq:countterm}
	\begin{split}
	\frac{\bg(\bk)}{\bg(\bk')} & \geq \prod_{i=1}^{\boo}\frac{\expo{(\bk \onevec)_i \log (\bk \onevec)_i-(\bk' \onevec)_{i} \log (\bk' \onevec)_{i}}}{\prod_{j=0}^{\btt}\expo{\bk_{ij}\log \bk_{ij}-\bk'_{ij}\log \bk'_{ij}}} \geq \prod_{i=1}^{\boo}\expo{(\bk \onevec)_i \log (\bk \onevec)_i-(\bk' \onevec)_{i} \log (\bk' \onevec)_{i}}\\
	&\geq \prod_{i=1}^{\boo} \expo{-\bttpo \log n} \geq \expo{-(\boo \times \bttpo)\log n}
	\end{split}
	\end{equation}
	In the derivation above we used (1) in Claim \ref{clmrounding}. It remains now to lower bound $\probsdpml(\bk)$:
	\begin{align*}
	\probsdpml(\bk)& \geq \extrel\bg(\bk) \geq \extrel\bg(\bk')\\
	& \geq \extrel\bg(\bk_{sdpml}) \geq \extrel\probsdpml(\bk_{sdpml})
	\end{align*}
	The first and second inequality follow from \Cref{app:roundster} and \Cref{eq:countterm} respectively. In the third inequality we used $\bg(\bk') \geq \bg(\bk_{sdpml})$ because $\bk'$ is the optimal solution over the relaxed constraint set $\bK^{f}_{\phid}$ and finally invoked \Cref{app:sterlingsapproxmain} to relate $\probsdpml$ and $\bg$.
\end{proof}

Now construct the extended discrete pseudo-distribution $\pk$ corresponding to the solution $\bk$ returned by \Cref{alg:rounding} by assigning $(\bk \onevec)_{i}$ elements with a probability value of $\pp_{i}$ $(\forall i \in [\boo+\btt])$. We next provide the proof for our main theorem that proves the distribution $\frac{\pk}{\|\pk\|_1}$ is an approximate PML distribution.
Our next theorem proves that the distribution $\frac{\pk}{\|\pk\|_1}$ is an approximate PML distribution.

\thmapproxpml*
\begin{proof}
	Let $\pk$ be the pseudo-distribution corresponding to solution $\bk$ returned by \Cref{alg:rounding}. Set $\bp_{approx}=\frac{\pk}{\|\pk\|_1}$, then:
	\begin{align*}
	\probpml(\bp_{approx},\phi) &\geq \probpml(\pk,\phi) 
	\geq \expo{-7\epst n \log n} \probpml(\pk,\phid) \geq  \expo{-7\epst n \log n} \Cphid \probsdpml(\bk) \\
	&\geq \extrel\Cphid \probsdpml(\bk_{sdpml})\\
	& \geq \expo{-O\left(\odiscloss +\grouploss\right) }\probpml(\bp_{pml,\phi},\phi)
	\end{align*}
	The first inequality follows because $\|\pk\|_1\leq 1$, second inequality from \Cref{cordiscmain}, third inequality follows because $\bk \in \bK^{ext}_{\pk,\phid}$ (because we constructed $\pk$ from $\bk$) and $\probsdpml(\bk)$ computes just one term in the summation over $\bK^{ext}_{\pk,\phid}$ (look at the representation of $\probpml(\pk,\phid)$ as summation over $\bK^{ext}_{\pk,\phid}$ from \Cref{otherside}), fourth inequality comes from \Cref{lem:roundedsoln} and last inequality follows from \Cref{lemsinglediscmain}.
	
	We bound the total running time as follows. Given a profile $\phi$, it takes $O(\phi_{size})$ to write down the discrete profile $\phid$, then we need to solve the convex optimization problem \ref{eq:convoptoriginalmain} which further takes $O\left(\frac{1}{\epst^2 \times \epso} \log^{O(1)}(\frac{1}{\epso \epst})+\frac{1}{\epst^3} \log^{O(1)}(\frac{1}{\epso \epst}) \right)$ and our final rounding algorithm can be implemented in time $O(\frac{\log^2n}{\epso \epst})$ ($=O(\boo\btt)$). The claimed running time follows by combining these bounds.
\end{proof}

%% file: symmetric.tex
\section{Unified optimal sample complexity for symmetric properties}\label{app:universal}
Here we study the connection between a universal estimator and approximate PML. We first recall the following theorem in~\cite{ADOS16}.
 \begin{thm}[Theorem 4 of \cite{ADOS16}]\label{thmbeta}
For a symmetric property $\bff$, suppose there is an estimator $\hat{\bff}:\Phi^n\rightarrow \R$, such that for any $\bp$ and observed profile $\phi$, 
$$\bbP(|\bff(\bp)-\hat{\bff}(\phi)|\geq \epsilon) \leq \delta$$
any $\beta$-approximate PML distribution satisfies:
$$\bbP(|\bff(\bp)-\bff(\bp^{\beta}_{pml,\phi})|) \geq 2 \epsilon) \leq \frac{\delta |\Phi^n|}{\beta}$$
\end{thm}

Our goal here is to prove \Cref{theoremADOSuniversal} that shows the following: computing an $\exp(\otilde(n^{2/3}))$-approximate PML distribution is sufficient to get a plug-in universal estimator that is sample competitive for estimating support size, coverage, entropy and distance from uniform. The proof presented in \cite{ADOS16} showed this connection for an $\exp(\sqrt{n})$-approximate PML estimator and it is easy to see the proof presented in \cite{ADOS16} works for any $\exp(n^{1-\delta})$-approximate PML estimator for constant $\delta >0$. We will need the following two lemmas from \cite{ADOS16, HR18}.
 \begin{lemma}[Lemma 2 of \cite{ADOS16}]\label{lemchange}
 Let $\alpha>0$ be a fixed constant. 
  For entropy, support, support coverage, and distance to uniformity there exist profile based
  estimators that use the optimal number of samples, have bias $\epsilon$ and if we change
  any of the samples, changes by at most $c \cdot \frac{n^{\alpha}}{n}$, where $c$ is a positive constant.
\end{lemma}

\begin{lemma}[\cite{HR18}]\label{lem:partition} $|\Phi^n| \leq \expo{3\sqrt{n}}$
	\end{lemma}

\theoremADOSuniversal*
\begin{proof}
Let $\bff$ be the property we wish to estimate, $\bp$ be the underlying distribution and $x^n, \phi$ are the observed sequence and profile. Set $\alpha = \eta$ ($\eta$ is a constant and so is $\alpha$) and let $\hat{\bff}$ be the estimator returned by Lemma \ref{lemchange}. The bias of estimator $\hat{\bff}$ is $$|\bff(p)-\E[\hat{\bff}(x^n)]| \leq \epsilon$$
By McDiarmid’s inequality we get:
$$\bbP\left(|\E[\hat{\bff}(x^n)]-\hat{\bff}(x^n)|\geq \epsilon \right)\leq \expo{-\frac{2\epsilon^2}{nc^2_{*}}}$$
where $c_{*}$ is the change in $\hat{\bff}$ when one of the samples is changed. Using these inequalities we get:
\begin{align*}
\bbP \left( |\bff(p)-\hat{\bff}(x^n)|\geq  2\epsilon \right) & \leq \bbP \left(|\bff(p)-\E[\hat{\bff}(x^n)]|+|\E[\hat{\bff}(x^n)]-\hat{\bff}(x^n)| \geq 2\epsilon \right)\\ 
& \leq  \bbP \left(|\E[\hat{\bff}(x^n)]-\hat{\bff}(x^n)| \geq \epsilon \right)\\
& \leq \expo{-\frac{2\epsilon^2}{nc^2_{*}}} = \expo{-\frac{2\epsilon^2}{n \left( \frac{c n^{\alpha}}{n} \right)^2}}=\expo{-\frac{2\epsilon^2 }{c^2}n^{1-2 \alpha}}
\end{align*}
In the derivation above we used $c_{*}\leq c \cdot \frac{n^{\alpha}}{n}$ (\Cref{lemchange}).
Invoke Theorem \ref{thmbeta} with $\delta=\expo{-\frac{2\epsilon^2 }{c^2}n^{1-2 \alpha}}$ we get:
\begin{align*}
\bbP\left(|\bff(\bp)-\bff(\bp_{pml,\phi})|\geq 4\epsilon \right) & \leq \frac{\delta |\Phi^n|}{\beta}  \leq \frac{\expo{-\frac{2\epsilon^2 }{c^2}n^{1-2 \alpha}}\expo{3\sqrt[]{n}}}{\expo{-O(n^{\frac{2}{3}}\log^3n)}}\\ 
&\leq \expo{-5n^{\frac{2}{3}+\eta}}\expo{O(n^{\frac{2}{3}}\log^3n)}\\
& \leq \expo{-n^{\frac{2}{3}}}
\end{align*}
In the first inequality we used \Cref{lem:partition}.
\end{proof}

%% file: appendix.tex
\appendix
\section{Minimum Probability}\label{app:minprob}
Here we provide the proof for our first technical lemma that gives a lower bound of $\Omega(\frac{1}{n^2})$ for the minimum non-zero probability value of a $\expo{-6}$-approximate PML distribution. To show such a result we use an independent rounding algorithm that is described in the lemma below. We need the following simple claim for the proof of our next lemma.
\begin{claim}\label{clm1}
For any non-negative and non-zero vector $\vvec$ and a profile $\phi \in \Phi^n$,
$$\probpml(\vvec,\phi)\leq (\|\vvec\|_1)^n \probpml(\bp_{pml,\phi},\phi)$$
\end{claim}
 \begin{proof}
 $$\probpml(\vvec,\phi) = (\|\vvec\|_1)^n\probpml\left(\frac{\vvec}{\|\vvec\|_1},\phi\right)\leq (\|\vvec\|_1)^n \probpml(\bp_{pml,\phi},\phi)$$
 \end{proof}
 
\lemmamin*
 \begin{proof}
 We do independent rounding to show the existence of such a distribution. For notational convenience we use $\bp_{pml,\phi}(x)$ to denote the probability of symbol $x$ in the PML distribution $\bp_{pml,\phi}$. Let $\bS\defeq\{ x \in \bX ~|~ \bp_{pml,\phi}(x) <\frac{1}{n^2} \}$ and for all $x \in \bS$ we define a random variable $Y_x$ as follows:
 $$Y_x\defeq
 \begin{cases}
 \frac{1}{n^2} \quad \text{with probability } ~n^2\bp_{pml,\phi}(x)\\
 0 \quad ~~\text{otherwise} \\
 \end{cases}
 $$
 Clearly $\forall x \in S$, 
 \begin{equation}\label{eqS1}
 \E \left[ Y_x\right]=\bp_{pml,\phi}(x)
 \end{equation}
 and in general for any integer power $i$ of random variable $Y_x$ we have:
 \begin{equation}\label{eqS2}
 \E \left[ Y_x^i\right ] \geq  \bp_{pml,\phi}^i(x) \quad \forall i=2,\dots 
 \end{equation}
For the remaining $x \in \bar{\bS}$ ($\bar{\bS}\defeq\bX \backslash \bS$) with $\bp_{pml,\phi}(x) \geq \frac{1}{n^2}$ we define:
 $$Z_x \defeq \bp_{pml,\phi}(x) \quad \text{with probability } ~1$$
 Define $\bY\defeq(Y_x)_{x\in \bS}$ and $\bZ\defeq(Z_x)_{x\in \bar{\bS}}$. 
 $$\mu_{\bS}\defeq\E \left[ \|\bY\|_1\right]= \E \left[\sum_{x \in \bS} Y_x\right]=\sum_{x \in \bS}\E \left[ Y_x\right]=\sum_{x \in \bS}\bp_{pml,\phi}(x)$$ 
 $$\mu_{\bar{\bS}}\defeq\E \left[ \|\bZ\|_1\right]= \E \left[\sum_{x \in \bar{\bS}} Z_x\right]=\sum_{x \in \bar{\bS}}\E \left[ Z_x\right]=\sum_{x \in \bar{\bS}}\bp_{pml,\phi}(x)$$
 $$\mu_{\bS}+\mu_{\bar{\bS}}=1$$
 
Define $\bp\defeq(\bY,\bZ)$ to be the concatenation of random vectors $\bY$ and $\bZ$. All random variables $Y_x,Z_x$ are mutually independent and we have:
\begin{equation}
\E \left[ \probpml(\bp,\phi)\right ] \geq \probpml(\bp_{pml,\phi},\phi)
\label{eq:expected_w1}
\end{equation}
(From Equation~\ref{eqS1},\ref{eqS2} and the fact that $Z_x$ is a constant random variable).
 
 When we generate a random sample $\bp$ from this distribution, we have a lower bound on the expected value of $\probpml(\bp,\phi)$ 
 but this is misleading since $\bp$ may not be a distribution. Scaling $\bp$ to 1 could significantly reduce the value of $\probpml(\bp,\phi)$ if $\|\bp\|_1$ is large. However, we show that a constant fraction of the expectation of $\probpml(\bp,\phi)$ comes from the sample space with bounded $\|\bp\|_1\leq 1+\frac{c}{n}$. Here $c$ is a constant and assume $c\geq 3$. Note that:
 %$$\E \left[ \probpml(\bp,\phi) ~\Big| ~\|\bp\|_1 \leq 1+\frac{c}{n} \right]=\E \left[ \probpml(\bp,\phi) ~\Big| ~\|\bY\|_1 + \|\bZ\|_1 \leq 1+\frac{c}{n} \right]=\E \left[ \probpml(\bp,\phi) ~\Big| ~\|\bY\|_1 \leq \mu_{\bS}+\frac{c}{n} \right]$$
 $$\|\bp\|_1 \leq 1+\frac{c}{n} \Leftrightarrow \|\bY\|_1 + \|\bZ\|_1 \leq 1+\frac{c}{n} \Leftrightarrow \|\bY\|_1 \leq \mu_{\bS}+\frac{c}{n}$$
 The last inequality follows because $\bZ$ is a constant random vector.
\begin{equation}\label{eqcondexp}
\begin{split}
\E \left[ \probpml(\bp,\phi) ~\Big| ~\|\bY\|_1 \leq \mu_{\bS}+\frac{c}{n} \right] & \bbP \left[ \|\bY\|_1 \leq \mu_{\bS}+\frac{c}{n} \right] + \E \left[ \probpml(\bp,\phi) ~\Big| ~\|\bY\|_1 > \mu_{\bS}+\frac{c}{n} \right]\bbP \left[ \|\bY\|_1 > \mu_{\bS}+\frac{c}{n} \right]\\
&=\E \left[ \probpml(\bp,\phi)\right ] \geq \probpml(\bp_{pml,\phi},\phi)
\end{split}
\end{equation}
To argue that a constant fraction of the expectation comes from the sample space with small $\|\bp\|_1$ we need a tight upper bound for:
$$\E \left[ \probpml(\bp,\phi) ~\Big| ~\|\bY\|_1 > \mu_{\bS}+\frac{c}{n} \right]\bbP \left[ \|\bY\|_1 > \mu_{\bS}+\frac{c}{n} \right]$$
For $t \geq c$, we first upper bound the probability term: 
 $$\bbP \left[ \|\bY\|_1 \geq \mu_{\bS}+\frac{t}{n} \right]$$
 We will use Chernoff bounds here and to apply them, we convert the $Y_x$ random variables into $\{0,1\}$ Bernoulli random variables. Define $\forall x \in \bS$,
 $$Y'_x\defeq n^2Y_x$$
 Equivalently:
  $$Y'_x\defeq
 \begin{cases}
 1 \quad \text{with probability} ~n^2\bp_{pml,\phi}(x)\\
 0 \quad ~~\text{otherwise} \\
 \end{cases}
 $$
 Define $\bY'\defeq(Y'_x)_{x \in \bS}$ and $\mu'_{\bS}\defeq\E \left[\|\bY'\|_1 \right]=n^2\mu_{\bS} \leq n^2$. For any $t>0$,
 %$$\bbP \left[ \|\bY\|_1 \geq \mu_{\bS}+\frac{c'}{n} \right]=\bbP \left[ \|\bY'\|_1 \geq n^2\mu_{\bS}+c'n \right]=\bbP \left[ \|\bY'\|_1 \geq \mu'_{\bS}+c'n \right]$$
 $$\|\bY\|_1 \geq \mu_{\bS}+\frac{t}{n} \Leftrightarrow \|\bY'\|_1 \geq n^2\mu_{\bS}+tn \Leftrightarrow \|\bY'\|_1 \geq \mu'_{\bS}+tn$$
 
 Since $\|\bY'\|_1$ is a sum of Bernoulli random variables, by Chernoff bounds:
\begin{equation}\label{eqchernoff}
\bbP \left[ \|\bY'\|_1 \geq \mu'_{\bS}+tn \right]=\bbP \left[ \|\bY'\|_1 \geq \left( 1+\frac{tn}{\mu'_{\bS}} \right) \mu'_{\bS} \right]\leq \expo{-\frac{t^{2}n^2}{3\mu'^2_{\bS}}\mu'_{\bS}}=\expo{-\frac{t^2n^2}{3\mu'_{\bS}}} \leq \expo{\frac{-t^2}{3}}
\end{equation}
%Observe that any $\bp$ generated by the above random process always has $\|\bp\|_1$ of the form $\mu_{\bar{\bS}}+\frac{i}{n}$ for some non-negative $i \in \indexset$, where $\indexset \defeq \{\frac{k}{n} ~|~ k \in \Z_{\geq 0} \}$.
Note from \Cref{clm1} that:
\begin{equation}\label{eq:massbound}
\E \left[ \probpml(\bp,\phi) ~\Big| \|\bY\|_1 \leq \mu_{\bS}+\frac{t}{n}\right] \leq \probpml(\bp_{pml,\phi},\phi) \left( 1+\frac{t}{n}\right)^{n} \leq \probpml(\bp_{pml,\phi},\phi) \cdot e^t \triangleq H(t)
\end{equation}
\begin{align*}
\bbP \left[ \|\bY\|_1 > \mu_{\bS}+\frac{c}{n} \right] & \E \left[ \probpml(\bp,\phi) ~\Big| ~\|\bY\|_1 > \mu_{\bS}+\frac{c}{n}  \right]\\
&= \int_{t=c}^{\infty} \E \left[ \probpml(\bp,\phi) ~\Big| \|\bY\|_1 = \mu_{\bS}+\frac{t}{n}\right] \bbP \left[ \|\bY\|_1 = \mu_{\bS}+\frac{t}{n} \right] dt\\
&\leq  \int_{t=c}^{\infty} H(t) \bbP \left[ \|\bY\|_1 = \mu_{\bS}+\frac{t}{n} \right] dt \quad \text{(By \Cref{eq:massbound})}\\
& \leq \int_{t=c}^{\infty} \frac{d H(t)}{dt} \bbP \left[ \|\bY\|_1 > \mu_{\bS}+\frac{t}{n} \right] dt\\
& = \probpml(\bp_{pml,\phi},\phi) \int_{t=c}^\infty e^t \expo{\frac{-t^2}{3}} dt\\
&= \probpml(\bp_{pml,\phi},\phi) \frac{\expo{3/4} \sqrt{3 \pi}}{2} \left( 1-\mathrm{erf}\left(\frac{2c-3}{2 \sqrt{3}}\right)\right)\\
&\leq 0.75 \cdot \probpml(\bp_{pml,\phi},\phi) \quad \mbox{for $c \geq 3$}
%& \leq \sum_{\{ i \in \indexset\}}\E \left[ \probpml(\bp,\phi) ~\Big| ~\frac{i}{n} \leq \|\bY\|_1-\left(\mu_{\bS}+\frac{c}{n}\right) \leq \frac{i+1}{n} \right]\bbP \left[ \frac{i}{n} \leq \|\bY\|_1 -\left(\mu_{\bS}+\frac{c}{n}\right) \leq \frac{i+1}{n} \right]\\
%& \leq \sum_{\{ i \in \indexset\}}\E \left[ \probpml(\bp,\phi) ~\Big| ~\|\bp\|_1 \leq 1+\frac{c}{n}+\frac{i+1}{n} \right]\bbP \left[ \frac{i}{n} \leq \|\bY\|_1 -\left(\mu_{\bS}+\frac{c}{n}\right) \leq \frac{i+1}{n} \right]\\
%& \leq \probpml(\bp_{pml,\phi},\phi) \sum_{\{ i \in  \indexset\}} \left( 1+\frac{c+i+1}{n}\right)^{n}  \bbP \left[ \frac{i}{n} \leq \|\bY\|_1 -\left(\mu_{\bS}+\frac{c}{n}\right) \leq \frac{i+1}{n} \right] (\because \text{Claim \ref{clm1} and Eq. \ref{eqchernoff} })\\
%& {\color{red}\leq \probpml(\bp_{pml,\phi},\phi) \sum_{\{ i \in \indexset\}}\expo{c+i+1}  \bbP \left[ \frac{i}{n} \leq \|\bY\|_1 -\left(\mu_{\bS}+\frac{c}{n}\right) \leq \frac{i+1}{n} \right]}\\
%& {\color{red}\leq \frac{1}{2}\probpml(\bp_{pml,\phi},\phi)}   
\end{align*}
 Substituting back in Equation \ref{eqcondexp} we have (for $c \geq 3$),
 $$\E \left[ \probpml(\bp,\phi) ~\Big| ~\|\bY\|_1 \leq \mu_{\bS}+\frac{c}{n} \right]  \bbP \left[ \|\bY\|_1 \leq \mu_{\bS}+\frac{c}{n} \right] \geq \frac{1}{4}\probpml(\bp_{pml,\phi},\phi)$$
 $$\Rightarrow \E \left[ \probpml(\bp,\phi) ~\Big| ~\|\bY\|_1 \leq \mu_{\bS}+\frac{c}{n} \right] \geq \frac{1}{4}\probpml(\bp_{pml,\phi},\phi)$$
 $$\Rightarrow \E \left[ \probpml(\bp,\phi) ~\Big| ~\|\bp\|_1 \leq 1+\frac{c}{n} \right] \geq \frac{1}{4}\probpml(\bp_{pml,\phi},\phi)$$
 The above inequality implies existence of a $\bp'$ with $\probpml(\bp',\phi)\geq \frac{1}{4}\probpml(\bp_{pml,\phi},\phi)$ and $ \|\bp'\|_1 \leq 1+\frac{c}{n} $. Define $\bp''\defeq \bp'/\|\bp\|_1$,
 $$\bp''=\frac{\bp'}{\|\bp'\|_1}$$
 $$\probpml(\bp'',\phi) = \|\bp'\|_1 ^{-n} \probpml(\bp',\phi) \geq (1+\frac{c}{n})^{-n}\frac{1}{4}\probpml(\bp_{pml,\phi},\phi)\geq \frac{\expo{-c}}{4}\probpml(\bp_{pml,\phi},\phi)$$
 In the final inequality substitute $c=3$ and observe $\frac{\expo{-c}}{4}\geq 1/100$. Also our rounding procedure always ensures that minimum non-zero entry of $\bp'$ is $\geq \frac{1}{n^2}$ that further implies a lower bound on the minimum non-zero probability value of $\bp''$ to be $\frac{1}{n^2}\frac{1}{\|\bp'\|_1}=\frac{1}{n^2}\frac{1}{1+c/n}\geq \frac{1}{2n^2}$. Hence $\bp''$ is our final distribution satisfying the conditions of lemma.
 \end{proof}

%% file: appendixprofile.tex
\section{Profile Discretization Lemma}\label{app:profiledisc}
Here we prove our profile discretization lemma. We first introduce a new definition called discrete type and then provide new formulations which help us in our proof.
 \begin{defn}[Discrete type] 
For a sequence $y^n \in \bX^{n}$, its \emph{discrete} type $\phihd=\Phihd(y^n) \in \bM^{\bX}$ is:
$$\phihd_x=\ceil{\bff(y^n,x)}_{\bM}$$
\end{defn}
For a sequence $y^n \in \bX^{n}$ let $\setd=\{ \bff(y^n,x) \}_{x \in \bX} \cup \{1,\dots  \ceil{\frac{1}{\epst}}\}$ be the set of all its distinct frequencies plus all integers less than $\ceil{\frac{1}{\epst}}$ and $\eled_1<\eled_2<\dots< \eled_{|\setd|}$ be elements of the set $\setd$. For this extended set $\setd$, the definition of profile $\phi=(\phi_{j})_{j=1\dots |\setd|}$ is still the same and $\phi_{j}=|\{x \in \bX~|~\bff(y^n,x)=\eled_{j} \}|$. In this extended definition there might be indices $j \in [1,|\setd|]$ with $\phi_{j}=0$ and this extended definition help us write cleaner proof for the next lemma. We first state an equivalent formulation for {the} probability of its profile $\phi=\Phi(y^n)$ (from Equation 20 in \cite{OSZ03}, Equation 15 in \cite{PJW17}) in terms of its type $\phih=\Phih(y^n)$:

\begin{equation}\label{eqpml2}
\probpml(\bp,\phi)=\left(\prod_{j=0\dots |\setd|}\frac{1}{\F_i!}\right)\binom{n}{\phih}\sum_{\sigma \in S_{\bX}}\prod_{x\in X}\bp_{x}^{\phih_{\sigma(x)}}=\left(\prod_{j=0\dots |\setd|}\frac{1}{\F_i!}\right)\cphi \sum_{\sigma \in S_{\bX}}\prod_{x\in X}\bp_{x}^{\phih_{\sigma(x)}}
%perm \left(\underbrace{\left( \bp_x^{\phih_{x'}}\right)_{x,x' \in \bX}}_{\Q}\right)
\end{equation}
%where $perm(\Q)$ is the permanent of matrix $\Q$. $$perm(\Q)=\sum_{\sigma \in S_{\bX}}\prod_{x\in X}\Q_{x,\sigma(x)}=\sum_{\sigma \in S_{\bX}}\prod_{x\in X}\bp_{x}^{\phih_{\sigma(x)}} $$
where $S_{\bX}$ is the set of all permutations of domain set $\bX$ and $\phi_0$ is the number of unseen domain elements. The difference between \Cref{eqpml2} and \Cref{eqlabeled} is the index set over which they are summed.

\lemmaprofiledisc*
 \begin{proof} 
 Let $\phih=\Phih(y^n)$ and $\phihd=\Phihd(y^n)$ be the type and discrete type of sequence $y^n$ respectively. By \Cref{eqpml2}:
 $$\probpml(\bp,\phi)=\left(\prod_{j=0}^{|\setd|}\frac{1}{\F_i!}\right)\cphi \left(\sum_{\sigma \in S_{\bX}}\prod_{x\in X}\bp_{x}^{\phih_{\sigma(x)}}\right)$$
Similarly: 
 $$\probpml(\bp,\phid)=\left(\prod_{j=0}^{|\bM|}\frac{1}{\Fd_i!}\right)\cphid \left(\sum_{\sigma \in S_{\bX}}\prod_{x\in X}\bp_{x}^{\phihd_{\sigma(x)}}\right),$$
where $\phid_{0}$ is the number of unseen domain elements in profile $\phid$. Note $\phid_{0}=\phi_{0}$ because our discretization procedure does not change the number of unseen domain elements. We now analyze both objectives term by term. For any permutation $\sigma \in S_{\bX}$
\begin{align*}
\prod_{x \in \bX}\bp_{x}^{\phihd_{\sigma(x)}} & \geq \prod_{x \in \bX}\bp_{x}^{\phih_{\sigma(x)}(1+\epst)}=\prod_{x \in \bX}\bp_{x}^{\phih_{\sigma(x)}}\prod_{x \in \bX}\bp_{x}^{\epst\phih_{\sigma(x)}} \geq \prod_{x \in \bX}\bp_{x}^{\phih_{\sigma(x)}} \left(\frac{1}{2n^2}\right)^{\epst n}\\
&\geq \expo{-3\epst n \log n}\prod_{x \in \bX}\bp_{x}^{\phih_{\sigma(x)}}
\end{align*}
The first inequality above follows because $\phihd_{\sigma(x)} \leq \phih_{\sigma(x)} (1+\epst)$ and using $\phih_{\sigma(x)} \leq \phihd_{\sigma(x)}$ we get the following inequality.
\begin{equation}\label{eqpro1}
\expo{3\epst n \log n}\prod_{x \in \bX}\bp_{x}^{\phihd_{\sigma(x)}} \geq \prod_{x \in \bX}\bp_{x}^{\phih_{\sigma(x)}}\geq \prod_{x \in \bX}\bp_{x}^{\phihd_{\sigma(x)}}
\end{equation}
Lets consider terms $\cphi$ and $\cphid$ next:
$$\frac{\cphi}{\cphid}=\frac{\binom{n}{\phih}}{\binom{\nd}{\phihd}}=\frac{n!}{\nd!}\prod_{x\in X}\frac{\phihd_x!}{\phih_x!}=\frac{n!}{\nd!}\prod_{x\in X}\frac{\ceil{\bff(y^n,x)}_{\bM}!}{\bff(y^n,x)!} \leq \prod_{x\in X}\frac{\floor{\bff(y^n,x)(1+\epst)}!}{\bff(y^n,x)!}\leq \prod_{x \in \bX}(n(1+\epst))^{\epst \bff(y^n,x)}$$
$$= (n(1+\epst))^{\epst n}\leq \expo{2\epst n \log n}$$
Next we lower bound the same quantity:
$$\frac{\cphi}{\cphid}=\frac{\binom{n}{\phih}}{\binom{\nd}{\phihd}}=\frac{n!}{\nd!}\prod_{x\in X}\frac{\phihd_x!}{\phih_x!}\geq \frac{n!}{\nd!} \geq \frac{n!}{\floor{n(1+\epst)}!}\geq (n(1+\epst))^{-\epst n} \geq \expo{-2\epst n \log n}$$
Combining both we get:
\begin{equation}\label{eqpro2}
\expo{-2\epst n \log n}\cphid \leq \cphi\leq \expo{2\epst n \log n}\cphid
\end{equation}
To bound our final term we use the extended definition of $\setd$. In this definition of $\setd$ we included all integers less than $\ceil{\frac{1}{\epst}}$ and we have $\eled_{j} =j$ for all $j \leq \ceil{\frac{1}{\epst}}$. Similarly recall all integers less than $\ceil{\frac{1}{\epst}}$ also belong to set $\bM$ and therefore $\nn_{j} =j$ for all $j \leq \ceil{\frac{1}{\epst}}$. Now observe that any frequency strictly less than $\ceil{\frac{1}{\epst}}$ ($\eled_{j} < \ceil{\frac{1}{\epst}}$) is not discretized and,
$$\Fd_j=\F_j \quad  \text{ for all } j < \ceil{\frac{1}{\epst}}$$
The number of domain symbols $x \in \bX$ with $\bff(y^n,x)\geq\ceil{\frac{1}{\epst}}$ is at most $\epst n$ and $\sum_{j \geq\ceil{\frac{1}{\epst}}}{\F_j}\leq \epst n$. This further implies, $\sum_{j \geq  \ceil{\frac{1}{\epst}}}{\phid_j} \leq \epst n$. Hence the ratio evaluates to:
$$1\leq \prod_{j=0}^{|\bM|}\Fd_j! \prod_{j=0}^{|\setd|}\frac{1}{\F_j!}=\prod_{j=0}^{\ceil{\frac{1}{\epst}}-1}\frac{\Fd_j!}{\F_j!}\prod_{j=\ceil{\frac{1}{\epst}}}^{|\bM|}\Fd_j! \prod_{j=\ceil{\frac{1}{\epst}}}^{|\setd|}\frac{1}{\F_j!}\leq \left(\sum_{j\geq \ceil{\frac{1}{\epst}}}{\Fd_j}\right)!\leq \ceil{\epst n}! \leq \expo{\epst n \log n}$$
Rewriting the final inequality:
\begin{equation}\label{eqpro3}
1\leq \prod_{j=0}^{|\bM|}\Fd_j! \prod_{j=0}^{|\setd|}\frac{1}{\F_j!} \leq \expo{2\epst n \log n}
\end{equation}
Combining all \cref{eqpro1,eqpro2,eqpro3} we have our result.
 \end{proof}

%% file: appendixstructure.tex
 \section{Remaining proofs for \Cref{sec:structuremain}}\label{app:structuremain}
Here we prove multiple lemmas associated with our functions $\probsdpml(\cdot)$ and $\bg(\cdot)$. Our first lemma shows that functions $\probsdpml(\cdot)$ and $\bg(\cdot)$ approximate each other in their values and later we also show that function $\bg(\bk)$ is log-concave in $\bk$. To help readability of this section lets recall definitions of functions $\probsdpml(\cdot)$ and $\bg(\cdot)$. For any $\bk \in \bK^{f}_{\phid}$,
 \begin{align*}
 \bg(\bk)=\expo{\log(\pvec)^{T}\bk \mvec + \sum_{i=1}^{\boo}(\bk \onevec)_i \log (\bk \onevec)_i-\sum_{i=1}^{\boo}\sum_{j=0}^{\btt}\bk_{ij}\log \bk_{ij}}\\
 \end{align*}
 Also for any $\bk \in \bK_{\phid}$,
 \begin{align*}
 \probsdpml(\bk)= \prod_{i=1}^{\boo}\left(\pp_i^{(\bk \mvec)_i}\frac{(\bk \onevec)_i!}{\prod_{j=0}^{\btt}\bk_{ij}!}\right)
 \end{align*}

\lemmastirlingsapproxmain*
 \begin{proof}
 By Stirling's approximation for all integer $n \geq 1$:
 $$\sqrt{2\pi}\leq \frac{n!}{\sqrt{n}\expo{n\log n-n}} \leq e $$
 We slightly use a weaker version of this inequality that holds all integers $n\geq 0$,
 $$1\leq \frac{n!}{\expo{n\log n-n}} \leq e \sqrt{n+1}$$
 $$\frac{\probsdpml(\bk)}{\bg(\bk)}=\prod_{i=1}^{\boo}\left( \frac{(\bk \onevec)_i!}{\expo{(\bk \onevec)_i\log (\bk \onevec)_i-(\bk \onevec)_i}} \prod_{j=0}^{\btt}\frac{\expo{\bk_{ij}\log \bk_{ij} -\bk_{ij}}}{\bk_{ij}!} \right) $$
 $$\leq \prod_{i=1}^{\boo}e \sqrt{1+(\bk \onevec)_i} \leq (e\sqrt{1+\pmin})^{\boo} \leq \expo{O(\log n) \boo}$$
In the above expression we used the fact that each $i \in \otboo$, $(\bk \onevec)_{i}\leq \pmin$ (Lemma \ref{lemminmain} combined with the constraint $\pvec^{T}\bk \onevec\leq 1$ ensures this fact). Also,
 \begin{align*}
 \frac{\probsdpml(\bk)}{\bg(\bk)}&\geq \prod_{i=1}^{\boo}\prod_{j=0}^{\btt}\frac{\expo{\bk_{ij}\log \bk_{ij} -\bk_{ij}}}{\bk_{ij}!} \geq \prod_{i=1}^{\boo}\prod_{j=0}^{\btt} \frac{1}{e\sqrt{1+\bk_{ij}}}\geq \left(\frac{1}{e\sqrt{1+\pmin}}\right)^{\boo (\btt+1)}\\
 &\geq \expo{-O(\log n) \boo \btt}
 \end{align*}
 \end{proof}
 
Next we show that function $\bg(\bk)$ is log-concave in $\bk$ and we need the following lemma to prove it.
 \begin{lemma}\label{lemfconvex}
 The function $h : \R^l_{\geq 0} \rightarrow \R$ defined for all $\ba \in \R^{l}_{\geq 0}$ by 
 \[
\bh(\ba)\defeq\sum_{i \in [l]}\ba_i \log \ba_i- \ba^{T}\onevec \log \ba^{T}\onevec 
\]
is convex.
 \end{lemma}
 \begin{proof}
 Let $\bA \defeq \ba^{T}\onevec$. Direct calculation reveals that for all $i \in [l]$,
 $$\frac{\partial}{\partial \ba_i}\bh(\ba) = 1+\log \ba_i-\log \bA-1=\log \ba_i -\log A
 ~. 
$$
 The Hessian matrix $\bH$ is:
 $$\bH(i,j)=\frac{\bd}{\bd \ba_i \ba_j}\bh=
 \begin{cases}
 \frac{1}{\ba_i}-\frac{1}{\bA} \quad if \quad i=j\\
 -\frac{1}{\bA} \quad if \quad i\neq j
 \end{cases}
 $$
 Let $\bD_{\ba}=diag(\ba)$ and also $\ba^{\frac{1}{2}}$ be the entry wise square root vector,
 $$\bH=\bD_{\ba}^{-1}-\frac{1}{\bA}\1\1^T$$
 $$\bD_{\ba}^{\frac{1}{2}}\bH\bD_{\ba}^{\frac{1}{2}}=I-\frac{1}{\bA}\bD_{\ba}^{\frac{1}{2}}\1\1^T\bD_{\ba}^{\frac{1}{2}}$$
 $$\bD_{\ba}^{\frac{1}{2}}\bH\bD_{\ba}^{\frac{1}{2}}=I-\frac{1}{\bA}\ba^{\frac{1}{2}}\ba^{\frac{1}{2}T} \succeq \textbf{0}\Rightarrow \bH \succeq 0$$
 The last inequality holds because $\frac{1}{\bA}\ba^{\frac{1}{2}}\ba^{\frac{1}{2}T}$ is a rank one matrix and its spectral norm is equal to 1:
 $$\left\|\frac{1}{\bA}\ba^{\frac{1}{2}}\ba^{\frac{1}{2}T} \right\|_2
= 
\tr\left(\frac{1}{\bA}\ba^{\frac{1}{2}}\ba^{\frac{1}{2}T}\right)
=\frac{1}{\bA} \tr(\ba^{\frac{1}{2}T}\ba^{\frac{1}{2}})=\frac{1}{\bA}\bA=1$$
 ~.
 \end{proof}

% \begin{lemma}\label{app:logconcave}
% For any profile $\phi$, the function $\bg(\bk)$ is log-concave in $\bk$.
% \end{lemma}
\lemmalogconcave*
 \begin{proof}
 Recall the definition of $\bg(\bk)$:
  \begin{align*}
 \bg(\bk)=\expo{\log(\pvec)^{T}\bk \mvec + \sum_{i=1}^{\boo}(\bk \onevec)_i \log (\bk \onevec)_i-\sum_{i=1}^{\boo}\sum_{j=0}^{\btt}\bk_{ij}\log \bk_{ij}}\\
 \end{align*}
 Taking $\log$ on both sides:
 $$\log \bg(\bk)=\log(\pvec)^{T}\bk \mvec + \sum_{i=1}^{\boo}(\bk \onevec)_i \log (\bk \onevec)_i-\sum_{i=1}^{\boo}\sum_{j=0}^{\btt}\bk_{ij}\log \bk_{ij}$$
 The first term is linear in $\bk$ and we consider the negative of second and third term and show it is convex. 
 \begin{align*}
 \bh(\bk)&=\sum_{i=1}^{\boo}\left( (\bk \onevec)_i \log (\bk \onevec)_i - \sum_{j=0}^{\btt}\bk_{ij}\log \bk_{ij} \right)\\
 &=\sum_{i=1}^{\boo} \bh_i(\bk_i)
 \end{align*}
 In the above expression $\bk_i \in \R^{\btt}$ is the $i$'th column of matrix $\bk$. By Lemma \ref{lemfconvex} each of the functions $\bh_i(\bk_i)$ is convex and $\bh(\bk)=\sum_{i=1}^{\boo} \bh_i(\bk_i)$ is also convex  ($-\bh(\bk)$ is concave). $\bg(\bk)$ is sum of a linear and a concave function, and is concave.
 \end{proof}

\newcommand{\termo}{T1}
\newcommand{\termt}{T2}
In the remaining part of this section, we prove our final result of this section that is used to bound the approximation guarantee of our rounding procedure. Recall our rounding procedure introduces new probability values resulting in a extended discretized probability space $\bP'$, where $\bP' \defeq  \bP \cup \{\pvec_{\boo+j} \}_{j \in \otbtt}$. To derive the relation between solution $\bk$ and PML objective value we defined extended sets $\bK^{ext}_{\bq,\phid}$ and $\bK^{ext}_{\phid}$.
Further for any $\bk \in \bK^{ext}_{\bq,\phid}$, recall that functions $\probsdpml(\cdot)$ and $\bg(\cdot)$ are defined as follows,
$$\probsdpml(\bk)\defeq \prod_{i=1}^{\boo+\btt}\left(\pp_i^{(\bk \mvec)_i}\frac{(\bk \onevec)_i!}{\prod_{j=0}^{\btt}\bk_{ij}!}\right)$$
$$ \bg(\bk)\defeq\prod_{i=1}^{\boo+\btt}\left(\pp_i^{(\bk \mvec)_i}\frac{\expo{(\bk \onevec)_i \log (\bk \onevec)_i-(\bk \onevec)_i}}{\prod_{j=0}^{\btt}\expo{\bk_{ij}\log \bk_{ij}-\bk_{ij}}}\right)$$
In the following lemma we show that for any $\bk \in \bK^{ext}_{\bq,\phid}$ returned by our rounding procedure the functions $\probsdpml(\bk)$ and $\bg(\bk)$ approximate each other in their values.
 \lemroundster*
 \begin{proof}
For all integers $n\geq 0$, recall the weaker version of sterlings approximation we used earlier ,
 $$1\leq \frac{n!}{\expo{n\log n-n}} \leq e \sqrt{n+1}$$
 Now, 
% $$\termo=\prod_{i=1}^{\boo+\btt}\left( \frac{(\bk \onevec)_i!}{\expo{(\bk \onevec)_i\log (\bk \onevec)_i-(\bk \onevec)_i}} \prod_{j=1}^{\btt}\frac{\expo{\bk_{ij}\log \bk_{ij} -\bk_{ij}}}{\bk_{ij}!} \right)$$
 \begin{align*}
 \frac{\probsdpml(\bk)}{\bg(\bk)}&=\prod_{i=1}^{\boo+\btt}\left( \frac{(\bk \onevec)_i!}{\expo{(\bk \onevec)_i\log (\bk \onevec)_i-(\bk \onevec)_i}} \prod_{j=0}^{\btt}\frac{\expo{\bk_{ij}\log \bk_{ij} -\bk_{ij}}}{\bk_{ij}!} \right)
% &=\termo \times \termt\\
 \end{align*}
 and 
 $$ \frac{\probsdpml(\bk)}{\bg(\bk)}\leq \prod_{i=1}^{\boo+\btt}e \sqrt{1+(\bk \onevec)_i} \leq (e\sqrt{1+\pmin})^{\boo+\btt} \leq \expo{O(\log n) (\boo+\btt)}$$
Now $\bP'=\bP \cup \{ \pp_{\boo+j}\}_{j \in \otbtt}$ and for any $j \in \otbtt$, $\pp_{\boo+j}$ is a convex combination of elements in $\bP$ and therefore $\pp_{\boo+j} \geq 1/2n^2$. In the above expression we used the fact that each $i \in \otboo$, $(\bk \onevec)_{i}\leq \pmin$ (For any $i \in [1,\boo+\btt]$, $\pp_{i} \geq 1/2n^2$ and further combined with the constraint $\pvecext^{T}\bk \onevec\leq 1$ (because $\bk \in \bK^{ext}_{\phid}$) ensures this fact). Also,
 \begin{align*}
 \frac{\probsdpml(\bk)}{\bg(\bk)}& \geq \prod_{i=1}^{\boo+\btt}\prod_{j=0}^{\btt}\frac{\expo{\bk_{ij}\log \bk_{ij} -\bk_{ij}}}{\bk_{ij}!}\\ 
&  \geq \left(\prod_{i=1}^{\boo}\prod_{j=0}^{\btt} \frac{1}{e\sqrt{1+\bk_{ij}}}\right) \left( \prod_{j=1}^{\btt}\frac{1}{e\sqrt{1+\bk_{\boo+j,j}}}\right)\\
& \geq \left(\frac{1}{e\sqrt{1+\pmin}}\right)^{\boo \bttpo +\btt}\\
 &\geq \expo{-O(\log n) \boo \btt}
 \end{align*}
 In the second inequality we used the fact that solution $\bk$ returned by our rounding procedure always satisfies $\bk_{\boo+j,k}=0$ for all $j \in \otbtt$, $k \in \ztbtt$ and $k \neq j$.
 \end{proof}

%% file: oracle_temp.tex
\section{Algorithm for solving our convex program}\label{app:cuttingplane}
To make this section self readable we start by recalling our original SDPML objective.
\begin{equation}\label{eq:w2}
\argmax_{\vx \in \bK_{\phid}}\bw_2(\vx)
\end{equation}
We relaxed it to: 
\begin{equation}\label{eq:g}
\argmax_{\vx\in \bK^{f}_{\phid}}\bg(\vx) \Leftrightarrow \argmax_{\vx\in \bK^{f}_{\phid}}\log \bg(\vx)
\end{equation}
where function $\bg(\vx)$ is defined as:
\begin{equation}\label{contsdpml}
 \bg(\vx)=\expo{\logpvec^{T}\vx \mvec + \sum_{i=1}^{\boo}(\vx \onevec)_i \log (\vx \onevec)_i-\sum_{i=1}^{\boo}\sum_{j=0}^{\btt}\vx_{ij}\log \vx_{ij}}
 \end{equation}
 
For $\ff(\vx)\defeq \log \bg(\vx)$ the optimization problem can be formulated equivalently as:
\begin{equation}\label{eq:convopt}
\argmax_{\vx\in \bK^{f}_{\phid}}\ff(\vx)
\end{equation}
where the constraint set $\bK^{f}_{\phid}$ is given by
\begin{equation}\label{setrelax}
\bK^{f}_{\phid} = \left\{\vx \in \R_{\geq 0}^{\boo \times \bttpo}~\big|~ (\vx^{T} \onevec)_{\otbtt} = \Fd, \text{ and } \logpvec^{T}\vx \onevec \leq  1 \right\} ~.
\end{equation}
and function $\ff(\vx)$ is:
$$\ff(\vx) \defeq \logpvec^{T}\vx \mvec + \sum_{i=1}^{\boo}(\vx \onevec)_i \log (\vx \onevec)_i-\sum_{i=1}^{\boo}\sum_{j=0}^{\btt}\vx_{ij}\log \vx_{ij} ~.$$

Our constraint set $\bK^{f}_{\phid}$ is bounded and for any $\vx \in \bK^{f}_{\phid}$, 
$$\|\vx\|_{F}^{2}=\sum_{i,j}\vx_{i,j}^{2}\leq \left( \sum_{i,j}\vx_{i,j}\right)^{2}= \left(\sum_{j}\Fd_{n_j} \right)^{2}=\nd^{2} \leq O(n^2)$$ 
\begin{comment}
However on the other hand our function $\ff(\vx)$ is not well behaved, the boundedness of $\ff$ doesn't imply any good polynomial bound on $\|\vx \|_{F}^{2}$.  We leverage the fact that our feasible set is bounded to define a new function which is close to our original function $\ff$ inside the feasible region and is also well behaved outside it.
Define 
$$\fnf(\vx) \defeq \sum_{i}\left( \sum_{j}\vx_{i,j} \log \left( \sum_{j}\vx_{i,j}\right) \right)-\left(\sum_{i}\sum_{j}\vC_{i,j}\vx_{i,j} +\sum_{i}\sum_{j}\vx_{i,j} \log \vx_{i,j}\right) -\frac{\reg}{n^3}\|\vx\|_{F}^{2}~.$$
For any $\vx \in \bK^{f}_{\phi}$:
\end{comment}
However on the other hand our function $\ff(\vx)$ is not well behaved as the boundedness of $\ff$ doesn't imply any good polynomial bound on $\|\vx \|_{F}^{2}$.  We leverage the fact that our feasible set is bounded to define a new function which is close to our original function $\ff$ inside the feasible region and is also well behaved outside it.
Define:
$$\fnf(\vx) \defeq \vC \cdot \vx + \sum_{i=1}^{\boo}(\vx \onevec)_i \log (\vx \onevec)_i-\sum_{i=1}^{\boo}\sum_{j=0}^{\btt}\vx_{ij}\log \vx_{ij} -\frac{\reg}{n^3}\|\vx\|_{F}^{2}~.$$
where $\vC=\mvec \logpvec^{T}$ and for any $\vx \in \bK^{f}_{\phi}$: $|\ff(\vx)-\fnf(\vx)| \leq o(\reg)$.
Hence optimizing $\ff(\vx)$ is equivalent to optimzing $\fnf(\vx)$ in an approximate sense:
\begin{equation}\label{eq:convoptapprox}
\argmax_{\vx\in \bK^{f}_{\phid}}\ff(\vx) \stackrel{\epsilon}{\approx} \argmax_{\vx\in \bK^{f}_{\phid}}\fnf(\vx)
\end{equation}
Let $\vx_{pml,\phi} $ be the matrix $\vx \in \bK^{f}_{\phid}$ which corresponds to distribution $p_{pml,\phi}$. Recall the maximum PML objective $\bw_1(p_{pml,\phi},\phi)$ is a probability term and is not hard to see that it is always between $[\expo{-n\log n}, 1]$ (lower bound comes from uniform distribution on $[n]$) and $\bw_2(\vx_{pml,\phi})$, $\bg(\vx_{pml,\phi}) \in [\expo{-2n\log n}, \expo{n \log n}]$ (using a crude approximation) because they approximate the value of $\bw_1(p_{pml,\phi},\phi)$. Combining all we get that optimum value of both optimization problems in Equation \ref{eq:convoptapprox} are always greater than $-n^2$.

In the rest of the section we show how to solve the optimization problem:
\begin{align*}
\argmax_{\vx\in \bK^{f}_{\phid}}\fnf(\vx)
\end{align*}
which can be equivalently written as:
\begin{equation}\label{convoptapp}
\argmax_{(\vx,\st) \in \conset}\st \text{ subject to } (\vx^{T} \onevec)_{\otbtt} = \Fd, \text{ and } \logpvec^{T}\vx \onevec \leq  1
\end{equation}
where the convex set $\conset \defeq \{(\vx,\st) \in \left( \R^{\boo \times \btt}, \R \right) ~|~ \fnf(\vx) \geq \st \text{ and } \st \geq -\lt \}$.

First we show how to solve a simple optimization problem which in turn will act as an oracle to solve our main optimization problem \ref{convoptapp} using cutting plane method from \cite{LSW15}. The simple optimization problem which we will refer to as oracle here on is stated next:
\begin{equation}\label{eq:oracle}
\opt=\max_{(\vx,\st) \in \setK} \vd \cdot \vx+\sd \st-\lambda \left( \|\vx\|_{F}^{2}+\st^2\right)
\end{equation}
where $\vd \in \R^{\boo \times \btt}$, $\sd \in \R$, $\setK$ is the same convex set and $\fnf(\cdot)$ is the same convex function defined above.

\begin{comment}
is a convex set defined as follows: 
$$\setK \defeq \{(\vx,\st) : \fnf(\vx)\geq \st \}$$
$$\fnf(\vx) \defeq \sum_{i}\left( \sum_{j}\vx_{i,j} \log \left( \sum_{j}\vx_{i,j}\right) \right)-\left(\sum_{i}\sum_{j}\vC_{i,j}\vx_{i,j} +\sum_{i}\sum_{j}\vx_{i,j} \log \vx_{i,j}\right) ~.$$
\end{comment}

We implement the oracle, that is, solve optimization problem \ref{eq:oracle}, by solving a sequence of unconstrained problems that penalize leaving the set $\setK$. Formally, for all $\alpha \in \R_{\geq 0}$ we define:
\begin{equation}\label{eq:fngxta}
\fnga(\vx,\st)\defeq\vd \cdot \vx+\sd\st-\lambda \left( \|\vx\|_{F}^{2} + \st^2\right) + \alpha \left(\fnf(\vx)-\st \right) ~.
\end{equation}
To implement our oracle we will show how solve the following to high precision 
\begin{equation}\label{eq:fixka}
\fnHa(\alpha)\defeq\max_{\{(\vx,\st)~|~\st \geq -\lt \}} \fnga(\vx,\st) ~.
\end{equation}
Our result will then follow by performing binary search on $\alpha$ and invoking this subroutine. 

For any $\alpha$ let $(\vxa,\sta)$ be the optimal solution for optimization problem \ref{eq:fixka} and also let $(\vxopt,\stopt)$ be the optimal solution to \ref{eq:oracle}. It is clear that:
\begin{align*}
\fnga(\vxa,\sta)
&\geq \fnga(\vxopt,\stopt) =\vd \cdot \vxopt+\sd\stopt-\lambda \left( \|\vxopt\|_{F}^2 + \left(\stopt\right)^2\right) + \alpha \left(\fnf(\vxopt)-\stopt \right) \\
&\geq \vd \cdot \vxopt+\sd\stopt-\lambda \left( \|\vx\|_{F}^{2} + \left(\stopt\right)^2\right)=\opt
\end{align*}
The second to last inequality follows because $\fnf(\vxopt)\geq \stopt$. Hence we have:
\begin{equation}\label{eq:optlowerbound}
\fnga(\vxa,\sta)\geq \opt
\end{equation}
Higher the value of $\alpha$ more incentive is to satisfy the constraint.
\begin{lemma}\label{lem:gradg}
For all $\alpha>0$ the following holds 
\[
(\lambda+\frac{\alpha \epsilon}{n^2})\|\vxa\|_{F}^{2} \leq \frac{1}{4\lambda}\|\vd\|_{F}^{2} +\frac{c^2}{4\lambda}-min(\frac{(c-\alpha)^2}{4\lambda},\alpha n^2- \lambda n^4- c n^2)
+ \alpha \left(
\fnf(\vxa) - \sta
\right)
\]
where $(\vxa,\sta)$ is the optimum solution pair for optimization problem \ref{eq:fixka}.
\end{lemma}
\begin{proof}
Direct calculation shows that the following derivatives for $\fnga$ hold for all input:
\begin{align*}
\dxij \fnga(\vx,\st)=\vd_{ij}-2 \lambda \vx_{ij} + \alpha \left(\log (\vx \onevec)_{i}-\vC_{ij}-\log \vx_{ij} -\frac{2\reg}{n^3}\vx_{i,j}\right)
\end{align*}
\begin{align*}
\dxt \fnga(\vx,\st)=\sd-2 \lambda \st - \alpha 
\end{align*}
By the optimality of $\vxa$ and $\sta$ we know these derivatives are $0$ at $(\vxa, \sta)$ and therefore:
\begin{equation}\label{eq:optimality}
\frac{2 \lambda \vxa_{ij}-\vd_{ij}}{  \alpha}+ \frac{\reg}{n^3}\vxa_{i,j}= \log (\vxa \onevec)_{i}-\vC_{ij}-\log \vxa_{ij} -\frac{\reg}{n^3}\vxa_{i,j}
\text{ and }
\sta=\max(-\lt,\frac{\sd-\alpha}{2\lambda}) ~.
\end{equation}
Consequently, 
\begin{align*}
\alpha \fnf(\vxa) &=  \alpha \left(  \sum_{i}\sum_{j}\vxa_{i,j} \left( \log (\vxa \onevec)_{i} - \vC_{i,j} -\log \vxa_{i,j}-\frac{\reg}{n^3}\vxa_{i,j}\right)  \right)\\
&= (2\lambda + \frac{\alpha \reg}{n^3})\| \vxa \|_{F}^{2}  - \vd \cdot \vxa ~.
\end{align*}
and substituting this and the value of $\sta$ into the formula for $\fnHa$ yields
\begin{align}\label{eq:eval}
\fnHa(\alpha) &= \fnga(\vxa,\sta) = \vd \cdot \vxa+\sd\sta-\lambda \left( \|\vxa\|_{F}^{2} + \left(\sta \right)^2\right) + \alpha \left(\fnf(\vxa)-\sta \right)\\
&= (\lambda + \frac{\alpha \reg}{n^3})\|\vxa\|_{F}^{2} + c \sta - \lambda \left(\sta\right)^2 - \alpha \sta \\ 
&= \min (\frac{(\sd-\alpha)^2}{4\lambda},\alpha n^2- \lambda n^4- c n^2)+(\lambda+\frac{\alpha \reg}{n^3}) \|\vxa\|_{F}^{2} \label{eq:end}
\end{align}
Combining this equality with the following upper bound for $\fnHa(\alpha)$ yields the result:
\[
\fnHa(\alpha) \leq
\max_{(\vx,\st)}\vd \cdot \vx+\sd\st-\lambda \left( \|\vx\|_{F}^{2}+\st^2\right)
+ \alpha \left(\fnf(\vxa)-\sta\right)
=\frac{1}{4\lambda}\|\vd\|_{F}^{2} +\frac{\sd^2}{4\lambda}
+ \alpha \left(\fnf(\vxa)-\sta\right)
~.
\]
\end{proof}

\begin{cor}\label{cor:bounded}
For any $\delta>0$ and $\alpha>\bamaxo$, where $\bamaxo=\boundalpha$
$$\fnf(\vxa) \geq \sta+\frac{\delta}{4\alpha}$$
\end{cor}
\begin{proof}
Suppose $\fnf(\vxa) < \sta+\epsone$, then by Lemma \ref{lem:gradg}, it holds that:
$$0 \leq (\lambda+\frac{\alpha \reg}{n^3}) \|\vxa\|_{F}^{2} \leq \frac{1}{4\lambda^2}\|\vd\|_{F}^{2} +\frac{\sd^2}{4\lambda^2}-\min(\frac{(c-\alpha)^2}{4\lambda},\alpha n^2- \lambda n^4- c n^2) +\frac{\delta}{4}< 0 ~.$$
The final inequality follows from the conditions of the corollary.
\end{proof}
Next we show that $\vxa$ is differentiable with respect to $\alpha$ and therefore, $\fnHa$, $\fnf(\vxa) - \sta$, and $\|\vxa\|_{F}^{2}$ are continuous with respect to $\alpha$. The crux is the following, simple, possibly well known fact whose proof is a slight modification of that in (cite geometric median).

\begin{lemma}
\label{lem:differentiability_of_paths}
Let $\fnf : \R^{n + 1} \rightarrow \R$ be a twice differentiable function and for all $\vvx \in \R^n$ and $\alpha \in \R$ define the function $\fnf_\alpha : \R^n \rightarrow \R$ by $\fnf_\alpha (\vvx) = \fnf(\vvx, \alpha)$ and let  $\vvx_\alpha \defeq \argmax_{\vvx \in \R^n} f_\alpha(x)$. If $\fnf_\alpha$ is strictly concave for all $\alpha \in \R$ then $\vvx_\alpha$ is differentiable as a function of $\alpha$. 
\end{lemma}

\begin{proof} 
By the optimality conditions for $\vvx_\alpha$ we know that 
$\nabla \fnf_\alpha(\vvx_\alpha) = \vec{0}$. Consequently, since $\fnf$ is differentiable, differentiating with respect to $\alpha$ yields by chain rule that
\[
\nabla^2 \fnf_\alpha(\vvx_\alpha) \left[\frac{d}{d\alpha} \vvx_\alpha \right] + \left. \frac{d}{d\alpha} \nabla \fnf_\alpha(\vvx) \right|_{\vvx_\alpha}
= \vec{0}
~.
\]
However, since $\fnf$ is strictly concave, all eigenvalues of this matrix are negative and this matrix is invertible yielding the desired result.
\end{proof}

\begin{lemma}\label{lem:cont}
Functions $\fnHa(\alpha)$, $ \fnf(\vxa)-\sta$ and $\|\vxa\|_{F}^{2}$ are continuous in $\alpha$.
\end{lemma}
\begin{proof}
Since $\fnHa$ is twice differentiable and $\fnH$ is strictly concave,  Lemma~\ref{lem:differentiability_of_paths} implies that $\vxa$ is differentiable and therefore continuous as a function of $\alpha$. Since $\fnf$ and $\|\vx\|_{F}^{2}$ are continuous functions the result follows.
\end{proof}

\begin{lemma}\label{lem:inc}
Let $\vxo,\vxt$ be the optimum solutions to Optimization problem \ref{eq:fixka} with respect to $\ao$ and $\at$ respectively. For any $\ao < \at$: 
$$\fnf(\vxo)-\sto >0 \imply \fnHa(\ao) < \fnHa(\at)
\text{ and }
\fnf(\vxt)-\stt <0 \imply \fnHa(\ao) > \fnHa(\at)
~.
$$
\end{lemma}
\begin{proof}
Suppose that $\fnf(\vxo)-\sto >0 $ as the proof for when $\fnf(\vxt)-\stt <0$ is analogous. Then since $\ao < \at$ we have $\ao (\fnf(\vxo)-\stt) < \at (\fnf(\vxo)-\stt)$ and
\begin{align*}
\fngo(\vxo,\sto)&= \vd \cdot \vxo+\sd\sto-\lambda \left( \|\vxo\|_{F}^{2}+(\sto)^2\right) + \ao \left(\fnf(\vxo)-\sto \right)\\
& < \vd \cdot \vxo+\sd\sto-\lambda \left( \|\vxo\|_{F}^{2}+(\sto)^2\right) + \at \left(\fnf(\vxo)-\sto \right)
= \fngt(\vxo,\sto) 
\end{align*}
The result follows as $\fnHa(\ao) =\fngo(\vxo,\sto)$ and $\fngt(\vxo,\sto) \leq \fngt(\vxt,\stt)=\fnHa(\at)$.
\end{proof}

\begin{cor}\label{cor:staypositive}
Let $\vxo,\vxt$ be the optimum solutions to Optimization problem \ref{eq:fixka} with respect to $\ao$ and $\at$ respectively. For any $\ao < \at$: 
$$\fnf(\vxo)-\sto >0 \imply \fnf(\vxt)-\stt >0$$
\end{cor}
\begin{proof}
Given $\ao < \at$ and $\fnf(\vxo)-\sto >0 $. By first part of the Lemma \ref{lem:inc} $\fnHa(\ao) < \fnHa(\at)$. Suppose $\fnf(\vxt)-\stt <0$ by the second part of same Lemma \ref{lem:inc} we have $ \fnHa(\ao) > \fnHa(\at)$ A contradiction!
\end{proof}

\begin{lemma}\label{lem:xbounded} For any $\alpha >0$,
$$\|\vxa\|_{F}^{2} \leq \bxmax$$
where $\bxmax=\max(\boundx,1)$
\end{lemma}
\begin{proof}
Observe that we can optimize problem \ref{eq:fixka} with respect to $\vx$ and $\st$ independently. Lets look at the function behaviour $\fnHa(\alpha)$ with respect to $\vx$. From equation \ref{eq:eval}-\ref{eq:end} we have:
$$\max_{\vx}\vd \cdot \vx-\lambda  \|\vx\|_{F}^{2}  + \alpha \fnf(\vx)=(\lambda+\frac{\alpha \reg}{n^3}) \|\vxa\|_{F}^{2}$$
Also note that $\fnf(\vx )<0$ for $\|\vx\|_{F}^{2} \geq \frac{n^6}{\reg}$ because the term $\frac{\reg}{n^3}\|\vx\|_{F}^{2}$ dominates and also there is a trivial solution with $\fnf(0)=0$. Combining all we get $\max_{\vx} \fnf(\vx) =\max_{\{\vx~|~\|\vx\|_{F}^{2} \leq \frac{n^6}{\reg}\}}  \fnf(\vx)$ and the function $\fnf(\vx) \leq O(\frac{n^6}{\reg^2})$ because all $|C_{i,j}| \leq O(n\log n)$ are bounded.
\begin{align*}
(\lambda+\frac{\alpha \reg}{n^3}) \|\vxa\|_{F}^{2} & = \max_{\vx}\vd \cdot \vx-\lambda  \|\vx\|_{F}^{2}  + \alpha \fnf(\vx)\\
& \leq \max_{\vx}\vd \cdot \vx-\lambda  \|\vx\|_{F}^{2}  + \max_{\vx} \alpha \fnf(\vx)\\
&=\max_{\vx}\vd \cdot \vx-\lambda  \|\vx\|_{F}^{2}  + \max_{\vx} \alpha \fnf(\vx) \\
&\leq \frac{\|\vd\|_{F}^{2}}{4\lambda} + O(\frac{n^6}{\reg^2})
\end{align*}

\end{proof}

%%%%%%%%%%%%%%%%%%%%%%%%%%%%%%%%%%%%%%%%%%%%%%%%comments

\newcommand{\timealpha}{O(\boo \btt \log(\|\vd\| c / (\lambda \delta))}

\begin{lemma}\label{lem:solvealpha}
For any $ \alpha >0$, we can find a solution $(\vxe,\ste)$ such that $ \|\vxe-\vxa\|_{1} \leq \epstwo \text{ and } \ste=\sta$ in time $O(\boo\cdot \btt\log \left(\frac{\bxmax}{\epstwo} \right))$.
\end{lemma}
\begin{proof}
Lets recall the objective of optimization problem \ref{eq:fixka}:
$$\fnga(\vx,\st)\defeq\vd \cdot \vx+\sd\st-\lambda \left( \|\vx\|_{F}^{2} + \st^2\right) + \alpha \left(\fnf(\vx)-\st \right)$$
Lets recall the optimality conditions from Equation~\ref{eq:optimality}:
\begin{equation}\label{eq:optdummyy}
\frac{2 \lambda \vxa_{ij}-\vd_{ij}}{  \alpha}= \log (\vxa \onevec)_{i}-\vC_{ij}-\log \vxa_{ij} -\frac{2\reg}{n^3}\vx_{ij}
\text{ and }
\sta=\max(\frac{\sd-\alpha}{2\lambda},-n^2) ~.
\end{equation}
%
%
%%%%%%%%%%%%%%%
Rearranging terms and taking exponential on Equation~\ref{eq:optdummyy} yields:
\begin{equation}\label{eq:optdd}
\sa^{\vxa_{ij}}\vxa_{ij}=\sbb_{ij}(\vxa \onevec)_{i}
\end{equation}
where $\sa \defeq \expo{\frac{2\lambda}{\alpha}+\frac{2\reg}{n^3}} > 1$ and $\sbb_{ij}=\expo{\frac{\vd_{ij}-\vC_{ij}}{\alpha}}$. Let $\sKai \defeq (\vxa \onevec)_{i}$ and we define new variables $\vya_{ij}$ which satisfy the following conditions, 
$$\vxa_{ij}=\vya_{ij}\sKai$$
and we know that $(\vya \onevec)_{i}\sKai=(\vxa \onevec)_{i}=\sKai$ and $\vya_{ij}$ should satisfy $(\vya \onevec)_{i}=1$. Lets rewrite Equation \ref{eq:optdd} in terms of $\vya_{ij}$ variables:
$$\sa^{\sKai \vya_{ij}}\vya_{ij}\sKai=\sbb_{ij}\sKai$$
This can be written equivalently as:
\begin{equation}\label{eq:optfixK}
\left( \sa^{\sKai}\right)^{\vya_{ij}}\vya_{ij}=\sbb_{ij}
\end{equation}
From Lemma \ref{lem:xbounded}, we can do binary search in $[0,\bxmax]$ to guess $\sKai$. Let $\lb$ and $\ub$ be the current lower and upper bounds for the value of $\sKai$: Assign $\sK=\frac{\lb+\ub}{2}$ and we can do binary search to find $\vyk_{ij}$ such that $\left( \sa^{\sK}\right)^{\vyk_{ij}}\vyk_{ij}=\sbb_{ij}$ because for fixed $\sa$ and $\sK$ function $\left( \sa^{\sK}\right)^{\vy}\vy$ is monotone (increasing) in $\vy$ ($\because \sa^{\sK}>1$).
\begin{enumerate}
\item If $(\vyk \onevec)_{i}=1$, assign $\vxa_{ij}=\vyk_{ij}\sK$ and Equation \ref{eq:optdd} is satisfied and we are done.
\item If $(\vyk \onevec)_{i}<1$, update $\ub=\sK$ that is decrease our guess for $\sKai$ to $\frac{\lb+\sK}{2}$ and observe that next iteration values of all $\vyk_{ij}$ increase as $\sbb_{ij}$ is fixed.
\item Else If $(\vyk \onevec)_{i}>1$, update $\lb=\sK$ because of the similar analysis as case above.
\item Assign $\sK=\frac{\lb+\ub}{2}$ and repeat.
\end{enumerate}

Note we never have to work with $\vy_{ij}$ variables, we introduced them to better understand our binary search procedure. From Lemma \ref{lem:xbounded} we have a good bound on $\|\vxa\|_{F}^{2}$ and the above procedure finds a solution $(\vxe,\ste)$ such that $\|\vxe-\vxa\|_{1} \leq \epstwo$ and $\ste=\sta$ (because we have closed form expression for $\sta$)  in time $O(\boo\cdot \btt\log \left(\frac{\bxmax}{\epstwo} \right))$.  
\end{proof}

\begin{lemma}\label{alg1}
Optimization problem \ref{eq:oracle} can be solved to $\delta$ accuracy in time $\timeoracle$.
\end{lemma}
\begin{proof}
First we show that solving optimization problem \ref{eq:fixka} for $\astar$ for which the solution pair $(\vxstar,\ststar)$ satisfies $\epsilon_{1}<f(\vxstar)-\ststar < 2\epsilon_{1}$ for $\epsilon_{1}=\frac{\delta}{4\alpha}$ solves our main problem \ref{eq:oracle}. Observe that the solution pair $(\vxstar,\ststar)$ satisfies our constraint and also our objective value for problem \ref{eq:fixka} at $(\vxstar,\ststar)$ is greater than $\opt-\frac{\delta}{2}$ as shown below:\\
$$\vd \cdot \vxstar+\sd\ststar-\lambda \left( \|\vxstar\|_{F}^{2} + (\ststar)^2\right)=\fng^{(\alpha^{*})}(\vxstar,\ststar)-2\epsilon_{1}\alpha \geq \opt -\frac{\delta}{2}$$
The first inequality follows because $\epsilon_{1}<f(\vxstar)-\ststar < 2\epsilon_{1}$ and the later one follows from Equation \ref{eq:optlowerbound}.
\begin{comment}
Solve the following 
\begin{equation}
\fnHa(\astar)=\min_{\alpha}\fnHa(\alpha)
\end{equation}
\end{comment}
By similar reasoning we are also done if at $\alpha=0$ the optimal solution pair $(\vxa,\sta)$ (has closed form solution) satisfies the constraint $\fnf(\vxa)>\sta$ and it is interesting if this constraint is not satisfied at $\alpha=0$. In such a case existence of $\astar$ such that $\epsilon_{1}<f(\vxstar)-\ststar < 2\epsilon_{1}$ follows from continuity (Lemma \ref{lem:cont}) and boundedness of $\alpha$ (Corollary \ref{cor:bounded}) for which the constraint $\epsilon_{1}<f(\vxstar)-\ststar < 2\epsilon_{1}$ is satisfied. Corollary \ref{cor:staypositive}, and \ref{lem:inc} suggests that we can find an $\alpha$ by binary search over the interval $(0,\bamaxo]$ such that $\epsilon_{1}<f(\vxstar)-\ststar < 2\epsilon_{1}$ and Lemma \ref{lem:solvealpha} finds a solution $\vxe$ such that $\|\vxe-\vxa\|_1 \leq \epsilon_{2}$.
Choose $\epsilon_{2} < \frac{\epssec}{poly(\bxmax,\bamax)n^{10}} $. \\
$\bullet$ If $(\vxa \onevec)_{i} \leq \frac{\epssec}{n^5}$ then so is $(\vxe \onevec)_{i} \leq \frac{2\epssec}{n^5}$ and value of $| \vxa_{i,j} \log (\vxa \onevec)_{i}-\vxe_{i,j} \log (\vxe \onevec)_{i}| \leq | \vxa_{i,j} \log (\vxa \onevec)_{i}|+|\vxe_{i,j} \log (\vxe \onevec)_{i}| \leq O(\frac{|\epssec \log \epssec|}{n^5})$.\\
$\bullet$ Else $(\vxa \onevec)_{i} > \frac{\epssec}{n^5}$ then so is $(\vxe \onevec)_{i} > \frac{\epssec}{2n^5}$ and value of $| \vxa_{i,j} \log (\vxa \onevec)_{i}-\vxe_{i,j} \log (\vxe \onevec)_{i}| \leq | \vxa_{i,j} \log (\frac{(\vxa \onevec)_{i}}{(\vxe \onevec)_{i}})|+|\epsilon_{2} \log (\frac{(\vxa \onevec)_{i}}{(\vxe \onevec)_{i}})| \leq | \vxa_{i,j} \log (1\pm\frac{\epsilon_{2}}{(\vxe \onevec)_{i}})|+|\epsilon_{2} \log (1\pm\frac{\epsilon_{2}}{(\vxe \onevec)_{i}})| \leq |\vxa_{i,j}\frac{\epsilon_{2}}{(\vxe \onevec)_{i}}|+|\epsilon_{2}\frac{\epsilon_{2}}{(\vxe \onevec)_{i}}| \leq O(\epsilon_{2})$.

We can do similar analysis for other terms in $\fnf(\vx)$ and the boundedness of $|\fnf(\vxa)-\fnf(\vxe)|$ follows because:
 \begin{align*}
|\fnf(\vxa)-\fnf(\vxe)| & \leq \sum_{i}\Big | \sum_{j}\vxa_{i,j} \log (\vxa \onevec)_{i} - \sum_{j}\vxe_{i,j} \log (\vxe \onevec)_{i} \Big| 
+\sum_{i}\sum_{j}|\vC_{i,j}||\vxa_{i,j}-\vxe_{i,j}| \\
&+\sum_{i}\sum_{j}|\vxa_{i,j} \log \vxa_{i,j}-\vxe_{i,j} \log \vxe_{i,j}| -\frac{\reg}{n^3}(\|\vxa\|_{F}^{2}-\|\vxe\|_{F}^{2})\\
& \leq \sum_{i} \left(O(\epsilon_{2})+O(\frac{|\epssec \log \epssec|}{n^5})\right)+O(n\log n) \epsilon_{2}+\sum_{i} \sum_{j}\left(O(\epsilon_{2})+O(\frac{|\epssec \log \epssec|}{n^5})\right)\\
&+\frac{\reg}{n^3}\sum_{i}\sum_{j}(\epsilon_{2}^2+2\epsilon_{2}\vxa_{ij})\\
&\leq o(\epssec)
\end{align*}
Recall $\ste=\sta$ and combined with inequality above and $\epsilon_{1}<f(\vxstar)-\ststar < 2\epsilon_{1}$ implies: 
$$\fnf(\vxe)>\ste$$
Now all that remains is to bound the objective value of optimization problem \ref{eq:oracle} $\vd \cdot \vxe+\sd \ste-\lambda \left( \|\vxe\|_{F}^{2} + (\ste)^2 \right)$.
\begin{align*}
\vd \cdot \vxe+\sd \ste-\lambda \left( \|\vxe\|_{F}^{2} + (\ste)^2 \right) &=\opt-\frac{\delta}{2}+ \vd \cdot (\vxe-\vxa)+\sd (\ste-\sta)\\
&-\lambda \left( \|\vxe\|_{F}^{2}-\|\vxa\|_{F}^{2} + (\ste)^2-(\sta)^2 \right)\\
&\geq \opt-\frac{\delta}{2} -\|\vd\|_{F}^{2}\epsilon_{2}-\lambda(|\|\vxe\|_{F}^{2}-\|\vxa\|_{F}^{2}|)\\
& \geq \opt-\frac{\delta}{2} -\|\vd\|_{F}^{2}\epsilon_{2} - \sum_{i}\sum_{j}(\epsilon_{2}^2+2\epsilon_{2}\vxa_{ij})\\
&\geq \opt-\frac{\delta}{2}-o(\epsilon_{1}) \\
& \geq \opt-\delta
\end{align*}
The whole procedure can be implemented in time $O(\boo\cdot \btt \log(\frac{\bxmax}{\epsilon_{2}})\log(\bamaxo))$
\end{proof}

Now we are in good shape to solve our main optimization problem \ref{convoptapp}. First we write our optimization problem in vector form:
\begin{equation}\label{convoptapp}
\argmax_{(x,\st) \in \conset}\st \text{ subject to } \ma x=\vb
\end{equation}
where the convex set $\conset \defeq \{(x,\st) \in (\R^{\boo \cdot \bttpo},\R) ~|~ \fnf(x) \geq \st \text{ and } \st \geq -\lt \}$  and our matrix $\ma \in \R^{\bttpo \times \boo \cdot \bttpo} $\footnote{Our matrix $\ma$ is a sparse matrix and matrix vector product with it can be computed in time $O(\boo \cdot \btt)$} and with vector $\vb \in \R^{\btt+1}$ represent the linear constraints in the set $\sK^{f}_{\phid}$.

\[
   \ma=
  \left[ {\begin{matrix}
 &  1 \dots 1  &  0\dots 0 & 0 \dots 0 & 0 \dots 0  \\
  & 0 \dots 0 & 1 \dots 1 & 0 \dots 0 & 0 \dots 0 \\
& \vdots & \ddots &\ddots & &\\
& 0 \dots 0 & 0 \dots 0 & 1 \dots 1 & 0 \dots 0\\
& \pp_1,\dots \pp_n & \dots & \pp_1,\dots \pp_n & \pp_1,\dots \pp_n\\
  \end{matrix} } \right]
\text{ and }
  \vb=\left[ {\begin{matrix}
  \Fd_{1}\\
  \Fd_{2}\\
  \vdots\\
  \Fd_{{\btt}}\\
  1
  \end{matrix}} \right]
\]
Formulation above is in the form of a general optimization problem $(11.14 )$ in \cite{LSW15}. For convenience we redefine the optimization problem $(11.14)$ from \cite{LSW15}:
\begin{equation}\label{eq:lsopt}
\max_{\vecx \in \conset \text{ and } \ma\vecx=\vecb} \vecc^T\vecx
\end{equation}
where $\setK$ is a convex set. To invoke the algorithm to solve this general optimization problem algorithm in \cite{LSW15} requires to implement a $\del$-2nd-order-optimization oracle which is define below:
\begin{defn}
Given a convex set $\conset$ and $\del > 0$. A $\del$-2nd-order-optimization oracle for $\conset$ is a function on $\R^{\dimn}$ such that for any input $\vecc \in \R^{\dimn}$ and $\lambda > 0$, it outputs $\vecy$ such that
$$\max_{\vecx \in \conset} \left(\vecc^T\vecx -\lambda \|\vecx\|^2\right) \leq \del+\vecc^T\vecy -\lambda \|\vecy\|^2$$
We denote by $\timeora$ the time complexity of this oracle
\end{defn}

Our simple optimization problem \ref{eq:oracle} is exactly the $\del$-2nd-order-optimization oracle for our main optimization problem \ref{convoptapp}. Consequently, all the remains to solve optimization problem \ref{convoptapp} is to bound the eigenvalues of $\ma \ma^\top$ and put together the results of this section to obtain our desired running time. We do this in Lemma~\ref{lem:eigenlower} and Theorem~\ref{oraclethm} respectively.

\begin{lemma}\label{lem:eigenlower}
The eigenvalues of matrix $\ma \ma^\top$ are either $\boo$ or of the form
\[
(\boo  + \bttpo \|\pvec\|_2^2)  \frac{1 \pm \sqrt{1 -\frac{4\boo \bttpo \|\pvec\|_2^2 - 4 \btt \|\pvec\|_1^2}{(\boo  + \bttpo \|\pvec\|_2^2)^2}}}{2}
\]
and therefore the smallest eigenvalue of $\ma \ma^\top$ is at least
\[
\frac{2\boo \bttpo \|\pvec\|_2^2 - 2 \btt \|\pvec\|_1^2}{\boo  + \bttpo \|\pvec\|_2^2}=\Omega\left( \boo \right)
~.
\] 
\end{lemma}

\begin{proof}
Direct calculation shows that if $\vec{1}_{\btt} \in \R^{\btt}$ is $\btt$-dimensional all ones vector, $I_{\btt} \in \R^{\btt \times \btt}$ is the $\btt$-dimensional identity matrix and $\pvec \in \R^{\boo}$ with $\pvec_i = \frac{1}{2n^2}(1+\epso)^{i}$ then for all $x \in \R^{\btt}$ and $\alpha \in \R$ we have
\[
\ma \ma^\top 
\left(
\begin{matrix}
x \\
\alpha
  \end{matrix} 
\right)
=
\left[ 
\begin{matrix}
 &  \boo I_{\btt}  & \|\pvec\|_1 \vec{1}_{\btt}  \\
  & \|\pvec\|_1 \vec{1}_{\btt}^\top  & \bttpo \|\pvec\|_2^2 
  \end{matrix} \right]
\left(
\begin{matrix}
x \\
\alpha
  \end{matrix} 
\right)
=
\left(
\begin{matrix}
\boo x + \alpha \|\pvec\|_1 \vec{1}_{\btt} \\
\|\pvec\|_1 \vec{1}_{\btt}^T x + \alpha \bttpo \|\pvec\|_2^2
  \end{matrix} 
\right)
~.
\]
Consequently $v = (x, \alpha)^T$ is an eigenvector of $\ma \ma^\top $ with eigenvalue $\lambda$ if and only if
\[
\boo x + \alpha \|\pvec\|_1 \vec{1}_{\btt} = \lambda x
\text{ and }
\|\pvec\|_1 \vec{1}_{\btt}^T x + \alpha \bttpo \|\pvec\|_2^2 = \lambda \alpha
\]
Now if $x \perp \vec{1}_{\btt}$ then we see the $v$ is an eigenvector if and only if $\alpha = 0$ in which case the eigenvalue is $\boo$. On the other hand if $x = \vec{1}_{\btt}$ then we see $v$ is an eigenvector of eigenvalue $\lambda$ if and only if 
\[
\lambda = \boo + \alpha \|\pvec\|_1
\text{ and }
\btt \|\pvec\|_1 + \alpha \bttpo \|\pvec\|_2^2 = \lambda \alpha ~.
\]
When this happens we have $\btt \|\pvec\|_1 + \alpha \bttpo \|\pvec\|_2^2 = \boo \alpha + \alpha^2 \|\pvec\|_1$ and solving for $\alpha$ yields that
\[
\alpha = \frac{\bttpo \|\pvec\|_2^2 - \boo \pm \sqrt{(\boo  - \bttpo \|\pvec\|_2^2)^2 + 4 \btt \|\pvec\|_1^2}}{2 \|\pvec\|_1}
\]
Substituting this into $\lambda = \boo + \alpha \|\pvec\|_1$ yields the eigenvalues. 
\begin{align*}
\lambda&=\frac{\bttpo \|\pvec\|_2^2 + \boo \pm \sqrt{(\boo  - \bttpo \|\pvec\|_2^2)^2 + 4 \btt \|\pvec\|_1^2}}{2}\\
&=\frac{\bttpo \|\pvec\|_2^2 + \boo \pm \sqrt{(\boo  + \bttpo \|\pvec\|_2^2)^2 -[4\boo \bttpo \|\pvec\|_2^2 - 4 \btt \|\pvec\|_1^2]}}{2}\\
&=(\boo  + \bttpo \|\pvec\|_2^2)  \frac{1 \pm \sqrt{1 -\frac{4\boo \bttpo \|\pvec\|_2^2 - 4 \btt \|\pvec\|_1^2}{(\boo  + \bttpo \|\pvec\|_2^2)^2}}}{2}
\end{align*}
The lower bound follows from the fact that for  $\sqrt{1 - a} \leq \sqrt{(1 - a / 2)^2} = 1 - a / 2$ when $a>0$ and therefore 
\[
\sqrt{1 -\frac{4\boo \bttpo \|\pvec\|_2^2 - 4 \btt \|\pvec\|_1^2}{(\boo  + \bttpo \|\pvec\|_2^2)^2}}
\leq  \left[1 - \frac{2\boo \bttpo \|\pvec\|_2^2 - 2 \btt \|\pvec\|_1^2}{(\boo  + \bttpo \|\pvec\|_2^2)^2}\right]
\]
The smallest eigenvalue is at least $\frac{2\boo \bttpo \|\pvec\|_2^2 - 2 \btt \|\pvec\|_1^2}{\boo  + \bttpo \|\pvec\|_2^2}$. Recall $\boo=\theta(\frac{\log n}{\epso})$ and $\boo$ is such that $\frac{1}{n^2}(1+\epso)^{\boo}\geq 1$ and  $\frac{1}{n^2}(1+\epso)^{\boo-1}<1$. Lemma statement follows because  $\|\pvec\|_1=\sum_{i=0}^{\boo}\frac{1}{2n^2}(1+\epso)^i=\theta(\frac{1}{\epso})$, $\|\pvec\|_2^2=\sum_{i=0}^{\boo}\frac{1}{4n^4}(1+\epso)^{2i}=\theta(\frac{1}{\epso})$.

\end{proof}

Below is the theorem we invoke to solve the optimization problem.

\begin{thm}[Theorem 56 from \cite{LSW15_arxiv}]\label{oraclethm}
Assume that $\max_{\vecx\in \conset}\normFull{\vecx}_{2}<M$, $\norm{\vecb}_{2}<M$,
$\norm{\vecc}_{2}<M$, $\norm{\ma}_{2}<M$ and $\lambda_{\min}(\ma)>1/M$.
Assume that $\conset \cap\{\ma\vecx=\vecb\}\neq \emptyset$ and we have $\epsilon$-2nd-order-optimization
oracle for every $\epsilon>0$. For $0<\del<1$, we can find
$\vecz\in \conset$ such that 
\[
\max_{\vecx\in \conset\text{ and }\ma\vecx=\vecb}\vecc^T\vecx\leq\delta+\vecc^T\vecz 
\]
and $\norm{\ma\vecz-\vecb}_{2}\leq\del$. This algorithm takes time
\[
O\left(r\timeora \log\left(\frac{nM}{\del}\right)+r^{3}\log^{O(1)}\left(\frac{nM}{\delta}\right)\right)
\]
where $r$ is the number of rows in $\ma$, $\eta=\left(\frac{\del}{nM}\right)^{\Theta(1)}$
and $\lambda=\left(\frac{\del}{nM}\right)^{\Theta(1)}$.
\end{thm}

\begin{comment}
To simply our analysis assume $\btt =\Omega(1)$ ($\btt$ being large constant doesn't hurt us anywhere) and $\btt \slt-\btt-\slo$ is always greater than zero and we next bound the eigenvalues of rank 2 matrix $\slo \cdot \onevec_{\btt+1}\onevec_{\btt+1}^T- \slo \cdot \onevec_{\btt}\onevec_{\btt}^T$. One can write the closed form solution solution for the eigenvalues of $\mb=\slo \cdot \onevec_{\btt+1}\onevec_{\btt +1}^T- \slo \cdot \onevec_{\btt}\onevec_{\btt}^T$ bu just looking at the 

\end{comment}

\begin{thm}
Optimization problem \ref{convoptapp} can be solved in time $\runningtime$
\end{thm}
\begin{proof}
The proof follows by combining Lemmas \ref{oraclethm}, \ref{lem:solvealpha}, \ref{alg1} and noting that all the parameters in the running time $\|d\|_2,|c|$, $1/\lambda$ are all bounded by $O(poly(\boo,\btt))$ and we only pay logarithm in these terms.
\end{proof}

%% file: multipmlappendix.tex
\section{Proofs for multidimensional PML}\label{app:multipmlstart}
Here we show how our techniques built throughout this paper apply to a general setting. In particular, we provide an efficient algorithm for computing approximate PML in higher dimensions when the dimension is constant. The proofs and techniques are analogous to one dimensional PML but there are few places such as, minimum probability lemma proof, singular value lower bound for the constraint matrix (for optimization) where we require general proofs.

\input{multipmlpreliminaries.tex}
\input{multipmlstructure.tex}

\input{multipmlsymmetric.tex}
\input{multipmlminprob.tex}
\input{multipmleigenbound.tex}

%% file: multipmlpreliminaries.tex
\renewcommand{\simplexd}{\Delta^{\bX,\md}}
\newcommand{\psimplexd}{\Delta_{\mathrm{pseudo}}^{\bX,\md}}
\renewcommand{\vsimplexd}{\Delta_{\mathrm{vector}}^{\bX,\md}}
\newcommand{\dsimplexd}{\Delta_{\mathrm{discrete}}^{\bX,\md}}
\renewcommand{\bpd}{\textbf{p}}

\renewcommand{\disclossd}{\sum_{k=1}^{\md}\epsok \nkk+ \sum_{k=1}^{\md}\epstk \nkk}
\newcommand{\problossd}{\sum_{k=1}^{\md}\epsok \nkk}
\newcommand{\multlossd}{\left(\sum_{k=1}^{\md}\epstk \nkk \right)\log \left( \sum_{k=1}^{\md} \epstk \nkk \right)}
\newcommand{\otmultlossd}{\sum_{k=1}^{\md}\epstk \nkk )}
\renewcommand{\grouplossepsd}{\prod_{k=1}^{\md}\frac{\log^3 \nkk}{\epsok \epstk}}
\newcommand{\grouplossepsdtilde}{\prod_{k=1}^{\md}\frac{1}{\epsok \epstk}}
\newcommand{\sdpmllossd}{\prod_{k=1}^{\md}\frac{\log^3 \nkk}{\epsok \epstk}}
\newcommand{\otsdpmllossd}{\prod_{k=1}^{\md}\frac{1}{\epsok \epstk}}

\renewcommand{\bp}{\bpd}
\renewcommand{\bq}{\bqd}
\newcommand{\minploss}{O(\md)}

\renewcommand{\vvec}{\textbf{v}}
\renewcommand{\vveco}{\vvec{(1)}}
\renewcommand{\vvecd}{\vvec{(\md)}}
\renewcommand{\vveck}{\vvec{(k)}}
\renewcommand{\vveci}{\vvec{(i)}}
\renewcommand{\vvecj}{\vvec{(j)}}

\newcommand{\bpdo}{\bpd{(1)}}
\newcommand{\bpdd}{\bpd{(\md)}}
\newcommand{\bpdk}[1]{\bpd_{#1}(k)}
\newcommand{\bpdi}{\bpd{(i)}}
\newcommand{\bpdj}{\bpd{(j)}}
\renewcommand{\bqd}{\textbf{q}}
\newcommand{\bqdo}{\bqd{(1)}}
\newcommand{\bqdd}{\bqd{(\md)}}
\newcommand{\bqdi}{\bqd{(i)}}
\newcommand{\bqdk}{\bqd{(k)}}
\newcommand{\bqdj}{\bqd{(j)}}

\renewcommand{\n}{\textbf{n}}
\renewcommand{\nd}{\textbf{n}'}
\newcommand{\noo}{\n{(1)}}
\newcommand{\nii}{\n{(i)}}
\newcommand{\njj}{\n{(j)}}
\newcommand{\ndd}{\n{(\md)}}
\renewcommand{\pvecext}{\pvec_{\mathrm{ext}}}
\newcommand{\ndoo}{\nd{(1)}}
\newcommand{\ndii}{\nd{(i)}}
\newcommand{\ndjj}{\nd{(j)}}
\newcommand{\nddd}{\nd{(\md)}}
\newcommand{\ndkk}{\nd(k)}

\renewcommand{\yn}{\textbf{y}^{\n}}
\renewcommand{\xn}{\textbf{x}^{\n}}
\renewcommand{\ynk}{\textbf{y}^{\nkk}}
\renewcommand{\yni}{\textbf{y}^{\nii}}
\renewcommand{\ynj}{\textbf{y}^{\njj}}
\renewcommand{\yno}{\textbf{y}^{\noo}}
\renewcommand{\ynt}{\textbf{y}^{\n(2)}}
\renewcommand{\ynd}{\textbf{y}^{\ndd}}
\renewcommand{\bffi}{\textbf{f}_{i}}
\renewcommand{\bffk}{\textbf{f}_{k}}

\newcommand{\phiho}{\lphi(1)}
\newcommand{\phihi}{\lphi(i)}
\newcommand{\phihk}[1]{\lphi_{#1}(k)}
\newcommand{\phihdk}[1]{\phihd_{#1}(k)}
\newcommand{\phihj}{\lphi(j)}
\newcommand{\phihdd}{\lphi(\md)}

\renewcommand{\freq}{\textbf{f}}
\renewcommand{\freqk}{\textbf{f}(k)}
\renewcommand{\freqj}{\textbf{f}_{j}}
\renewcommand{\Freq}{\textbf{f}}
\renewcommand{\Freqn}{\textbf{F}^{\n}}
\renewcommand{\Freqnd}{\textbf{F}^{\nd}}
\renewcommand{\otmd}{[1,\md]}

\renewcommand{\cc}{\textbf{c}}
\newcommand{\op}{\mathrm{op}}

\renewcommand{\epso}{\epsilon}
\newcommand{\epsoi}{\epso(i)}
\newcommand{\epsoj}{\epso(j)}
\renewcommand{\epsok}{\epso(k)}
\newcommand{\epsod}{\epso(\md)}
\newcommand{\epsoo}{\epso(1)}

\renewcommand{\epst}{\gamma}
\newcommand{\epsti}{\epst(i)}
\newcommand{\epstj}{\epst(j)}
\renewcommand{\epstk}{\epst(k)}
\newcommand{\epstd}{\epst(\md)}
\newcommand{\epsto}{\epst(1)}

\renewcommand{\boo}{\textbf{b}}
\newcommand{\booo}{\textbf{b}_{1}}
\newcommand{\bood}{\textbf{b}_{\md}}
\newcommand{\booi}{\textbf{b}_{i}}
\newcommand{\booj}{\textbf{b}_{j}}
\newcommand{\book}{\textbf{b}_{k}}

\renewcommand{\btt}{\textbf{e}}
\newcommand{\btto}{\btt_{1}}
\newcommand{\bttd}{\btt_{\md}}
\newcommand{\btti}{\btt_{i}}
\newcommand{\bttj}{\btt_{j}}
\newcommand{\bttk}{\btt_{k}}

\newcommand{\ik}{i_{k}}
\newcommand{\io}{i_{1}}
\newcommand{\id}{i_{\md}}

\newcommand{\otbook}{[1,\book]}

\subsection{Preliminaries for $\md$-dimensional objects}\label{sec:prelimsd}

\noindent \textbf{$\md$-tuple}: $\cc$ is a $\md$-tuple if $\cc \in \R^{\md}$. For all $k \in \otmd$, we use $\cc(k)$ to denote its $k$'th element.\\

\noindent \textbf{Arithmetic operations on $\md$-tuples}: For any two $\md$-tuples $\cc$, $\cc'$ and an arithmetic operator $\op \in  \{+, \times, -,/\}$, the operation $\cc~ \op ~\cc'$ denotes element wise operation, meaning it outputs another $\md$-tuple equal to $(\cc(1) ~\op~ \cc'(1), \dots , \cc(\md)~\op~\cc'(\md))$. Further for any $\md$-tuple $\cc$ and scalar $s$, the operation $\cc~ \op ~ s$ denotes element wise scalar operation, meaning it outputs another $\md$-tuple equal to $(\cc(1) ~\op~ s, \dots , \cc(\md)~\op~s)$. Just in the case of power operation $\cc^{\cc'}$ we return a scalar value and is equal to:
$$\cc^{\cc'} \defeq \prod_{k=1}^{\md} \cc(k)^{\cc'(k)}$$
Also for a $\md$-tuple $\cc$ and scalar $s$ we define:
$$\cc^{s} \defeq \prod_{k=1}^{\md} \cc(k)^{s} \text{ and }s^{\cc} \defeq \prod_{k=1}^{\md} s^{\cc(k)}$$

\noindent \textbf{Logic operations on $\md$-tuples}: For any two $\md$-tuples $\cc$ and $\cc'$ and a logic operator $\op \in  \{\leq, \geq,=\}$, the operation $\cc ~\op~ \cc'$ is true if and only if $\cc(k) ~\op~ \cc'(k)$ is true for all $k \in \otmd$. Further for any $\md$-tuple $\cc$ and scalar $s$, the logic operation $\cc~ \op ~ s$ is true iff $\cc(k) ~\op~ s$ is true for all $k \in \otmd$.\\

\noindent \textbf{Floor and ceil operations on $\md$-tuples}: For a $\md$-tuple $\cc$ and set $\bS$ of $\md$-tuples we use the notation $\floor{\cc}_{\bS}$ and $\ceil{\cc}_{\bS}$ to denote the following $\md$-tuples:\\
$$ \floor{\cc}_{\bS}\defeq\max_{\cc' \in \bS: \cc' \leq \cc}\cc' \quad \text{ and } \quad \ceil{\cc}_{\bS}\defeq\min_{\cc' \in \bS: \cc' \geq \cc}\cc'$$

We next recall (defined in \Cref{subsec:higherPML}) the setting for higher dimensions.

\noindent \textbf{Setting for higher dimension}: For each $k \in \otmd$, we receive a sequence $\ynk$ that consists of $\nkk$ independent samples drawn from an underlying distribution $\bpd(k)$ supported on same domain $\bX$, further $\ynk$ is independent of other sequences $\textbf{y}^{\n(k')}$ for $k' \in \otmd$ and $k' \neq k$. We call $\yn=(\yno,\dots \ynd)$ a $\md$-sequence and $\n=(\n(1),\dots, \n(\md))$ its $\md$-length. Let $\bX^{\n}$ be the set of all $\md$-sequences of $\md$-length equal to $\n$. We use $\bpd_{x}(k)$ to denote the probability of domain element $x$ in distribution $\bpd(k)$. We also refer $\bpd=(\bpd(1),\dots,\bpd(\md))$ as a $\md$-distribution and let $\simplexd$ be the set of all $\md$-distributions.

For any $\md$-distribution $\bp \in \simplexd$, the probability of a $\md$-sequence $\yn$ is defined as:
$$\bbP(\bpd,\yn) \defeq \prod_{k=1}^{\md}\prod_{x \in \bX}(\bpd_{x}(k))^{\bff(\ynk,x)}$$

Recall for each $k \in \otmd$, $\bff(\ynk,x)$ is the frequency of domain element $x$ in sequence $\ynk$. For any $\md$-sequence $\yn$, we call $\bff(\yn,x)=(\bff(\yno,x),\dots, \bff(\ynd,x))$ the $\md$-frequency of domain element $x$ in $\yn$. Let $\Freqn$ be the set of all $\md$-frequencies generated by different domain elements in all possible $\md$-sequences in $\bX^{\n}$ and we use $\freqj \in \Freqn$ to denote its $j$th element.\\ 

We next define few more $\md$-dimensional objects of interest.
%\noindent \textbf{$\md$-distribution}: $\bpd=(\bpdo,\dots,\bpdd)$ is a $\md$-distribution if for each element $k \in \otmd$, $\bpdk{}$ is a distribution supported on the same domain $\bX$. Let $\simplexd$ be the set of all $\md$-distributions.\\

\noindent \textbf{$\md$-vector}: $\vvec=(\vveco,\dots,\vvecd)$ is a $\md$-vector if for each element $k \in \otmd$, $\vveck$ is a vector supported on the same domain $\bX$. We use $\vvec_{x}$ to denote the row corresponding to domain element $x$ and $\vvec(k)$ to denote its $k$'th column. Let $\vsimplexd$ be the set of all $\md$-vectors and note that $\md$-distribution is a $\md$-vector.\\

\noindent \textbf{Norm of $\md$-vectors}: For a $\md$-vector $\vvec$, its norm denoted by $\|\vvec\|$ is a $\md$-tuple equal to $(\|\vveco\|,\dots,\|\vvecd\|)$.\\

\noindent \textbf{$\md$-pseudodistribution}: $\bqd=(\bqdo,\dots,\bqdd)$ is a $\md$-pseudodistribution if for each element $k \in \otmd$, $\bqdk$ is a pseudo-distribution supported on the same domain $\bX$ or equivalently $\|\bqd\|_{1} \leq 1$. Let $\psimplexd$ be the set of all $\md$-pseudodistributions and $\simplexd \subset \psimplexd \subset \vsimplexd$.\\

\noindent \textbf{$\md$-level set}: For a $\md$-distribution $\bpd$ and $\md$-pseudodistribution $\bqd$, we call $\bpd_{x}$ and $\bqd_{x}$ $\md$-level sets corresponding to $x$ respectively.\\

%\noindent \textbf{$\md$-length}: $\n=(\noo,\dots, \ndd)$ is a $\md$-length if for each $k \in \otmd$, $\nkk$ represents length of a sequence.\\

%\noindent \textbf{$\md$-sequence}: $\yn=(\yno,\dots \ynd) $ is a $\md$-sequence of $\md$-length equal to $\n$ if for each $k \in \otmd$, $\ynk$ is a sequence of $\nkk$ independent draws from some distribution $\bpd(k)$. Note sequence $\ynk$ is independent of other sequences $\textbf{y}^{\n(k')}$ for $k' \in \otmd$ and $k' \neq k$. Let $\bX^{\n}$ be the set of all $\md$-sequences of $\md$-length equal to $\n$.\\

%\noindent \textbf{$\md$-frequency}: $\freq=(\freq(1),\dots, \freq({\md}))$ is a $\md$-frequency if for each $k \in \otmd$, $\freqk$ represents a frequency. For any $\md$-sequence $\yn$, $(\bff(\yno,x),\dots, \bff(\ynd,x))$ represents $\md$-frequency of domain element $x$ in $\yn$, where for each $k \in \otmd$, $\bff(\ynk,x)$ is the frequency of domain element $x$ in sequence $\ynk$. We overload notation and also use $\Freq(\yn,x)$ to denote a function $\Freq(\cdot,\cdot)$ that takes $\md$-sequence and domain element as input and outputs $(\bff(\yno,x),\dots, \bff(\ynd,x))$ the $\md$-frequency of $x$ in $\yn$. 
%We use shorthand notation $\Freq_{x}=\Freq(\yn,x)$ when the $\md$-sequence is clear from the context. 
%Let $\Freqn$ be the set of all $\md$-frequencies generated by different domain elements in all possible $\md$-sequences in $\bX^{\n}$ and we use $\freqj \in \Freqn$ to denote its $j$th element.\\

\noindent \textbf{$\md$-Type}: For any $\md$-sequence $\yn$, $\phih=(\Phih(\yno),\dots, \Phih(\ynd))$ represents $\md$-type of $\yn$ and we call $\n$ its $\md$-length. Recall $\phih(k)\defeq\Phih(\ynk)$ is the type of sequence $\ynk$ and we overload notation and let $\phih=\Phih(\yn)$ denote $(\Phih(\yno),\dots, \Phih(\ynd))$. We use $\phih_{x}=\bff(\yn,x)$ to denote the row corresponding to domain element $x$ and $\phih_{x}(k)=\phih(k)_{x}=\bff(\ynk,x)$ all mean the same thing. Let $\Tn$ be the set of all $\md$-types of $\md$-length equal to $\n$.

For a $\md$-distribution $\bp \in \simplexd$, the probability of a $\md$-type  $\phih \in \Tn$ is:
\[
\bbP(\bp,\phih) \defeq \sum_{\{\yn \in \bX^{\n}~|~ \Phih (\yn)=\phih \}}\bbP(\bp,\yn)=\prod_{k=1}^{\md}\binom{\nkk}{\phihk{}}\prod_{x \in \bX}\bpdk{x}^{\phihk{x}}
\text{ recall, }
\binom{\nkk}{\phihk{}}=\frac{\nkk !}{\prod_{x \in \bX}\phihk{x}!}
\footnote{$0!\defeq1$}
\]
We use the following shorthand notation to denote the counting term in the above expression.
 \[
\binom{\n}{\phih}= \prod_{k=1}^{\md}\binom{\nkk}{\phihk{}}
 \]

\noindent \textbf{$\md$-Profile}: For any $\md$-sequence $\yn \in \bX^{\n}$, $\phi=\Phi(\yn)$ is a $\md$-profile if $\phi=(\F_j)_{j=1\dots |\Freqn|}$ and $\F_j=|\{x\in \bX ~|~\Freq(\yn,x)=\freqj \}|$\footnote{The $\md$-profile does not contain $(0,\dots,0)$ $\md$-frequency element because we don't know the number of unseen domain symbols.} is the number of domain elements with $\md$-frequency $\freqj$.
We call $\n$ the $\md$-length of $\F$ and use $\Phi^n$ to denote the set of all $\md$-profiles of $\md$-length equal to $\n$.\\

For any $\md$-distribution $\bp \in \simplexd$, the probability of a $\md$-profile $\phi \in \Phi^n$ is defined as:
\begin{equation}\label{eqpml1d}
\probpml(\bp,\phi)\defeq\sum_{\{\yn \in \bX^{\n}~|~ \Phi (\yn)=\phi \}} \bbP(\bp,\yn)\\
\end{equation}
We can also define the $\md$-profile of a $\md$-type  $\phih$. We overload notation and use $\phi=\Phi(\phih)$ to denote the $\md$-profile associated with $\md$-type  $\phih$ and $\F_j=\F_j(\phih)\defeq|\{x\in \bX ~|~\phih_x = \freqj \}|$. Consider all types $\lphi$ such that $\Phi(\lphi)=\phi$ and observe that they all have the same $\binom{\n}{\lphi}$ value. We use notation $\cphi$ to represent this quantity:
\begin{equation}\label{cphi}
\cphi \defeq \binom{\n}{\phih}= \prod_{k=1}^{\md}\binom{\nkk}{\phihk{}}
\end{equation}
\begin{equation}\label{eqlabeledd}
\probpml(\bp,\phi)=\sum_{\{\yn \in \bX^{\n}| \Phi (\yn)=\phi \}} \bbP(\bp,\yn) 
%=\sum_{\{\phih \in \Tn ~|~ \Phi(\phih)=\phi\}}~\sum_{\{\yn \in \bX^{\n}~|~ \Phih (\yn)=\phih \}}\bbP(\bp,\yn) \\
=\sum_{\{\phih\in \Tn | \Phi(\phih)=\phi\}}\bbP(\bp,\phih) =\cphi \sum_{\{\phih \in \Tn | \Phi(\phih)=\phi\}}  \prod_{k=1}^{\md} \prod_{x \in \bX}\bpdk{x}^{\phihk{x}}
\end{equation}

\noindent \textbf{Profile maximum likelihood}: For any $\md$-profile $\phi \in \Phi^{n}$, a \emph{Profile Maximum Likelihood} (PML) $\md$-distribution $\bp_{pml,\phi} \in \simplexd$ is:
$$\bp_{pml,\phi}=\argmax_{\bp} \probpml(\bp,\phi)$$
and $\probpml(\bp_{pml,\phi},\phi)$ is the maximum PML objective value.\\

\noindent \textbf{Approximate profile maximum likelihood}: For any $\md$-profile $\phi \in \Phi^{n}$, a $\md$-distribution $\bp^{\beta}_{pml,\phi} \in \simplexd$ is a $\beta$-\emph{approximate PML} $\md$-distribution if 
$$\probpml(\bp^{\beta}_{pml,\phi},\phi)\geq \beta \cdot \probpml(\bp_{pml,\phi},\phi)$$

\noindent \textbf{Note}: As in the case of one dimension, we extend and use the following definition for $\bbP(\vvec,y^{\n})$ for any $\md$-vector. Further, for any probability terms defined in the future involving $\bp$, we assume those expressions are also defined for any $\md$-vector $\vvec$ just by replacing $\bp_{x}(k)$ by $\vvec_{x}(k)$ everywhere and $\vvec(k)_{x}=\vvec_{x}(k)$ mean the same thing.

\noindent \textbf{Probability discretization}: Let $\bP\defeq\{\pp_{i}:i=1,\dots \boo\}$ be the set representing discretization of $\md$-probability space where for each $i \in \otboo$, $\pp_{i}$ is a $\md$-level set. Further all elements in $\bP$ are of the form $((1+\epsoo)^{1-\io},\dots, (1+\epsod)^{1-\id})$ for some fixed $\epso \in \R_{>0}^{1\times \md}$ and for all possible index $\ik \in \otbook$, where for each $k \in \otmd$, $\book$ is such that $(1+\epsok)^{1-\book}\leq \frac{1}{2\n(k)^2}$ and $\boo=\prod_{k=1}^{\md} \book$.

\noindent \textbf{Discrete $\md$-pseudodistribution}: For any $\md$-distribution $\bp \in \simplexd$, its \emph{discrete} $\md$-pseudodistribution $\bq=disc(\bp) \in \psimplexd$ is defined as:
 $$\bq_x\defeq\floor{\bp_x}_{\bP} \quad \forall x \in \bX$$
 We use $\dsimplexd$ to denote the set of all discrete $\md$-pseudodistributions. Note that $\floor{\bp_x}_{\bP} \geq \frac{\bp_x}{1+\epso}$ and $\frac{1}{1+\epso} \leq ||\bqd||_1 \leq 1$.\\
 
\noindent \textbf{Multiplicity discretization}: Let $\bM=\{\nn_{j}:j=1\dots \btt \}$ be the set representing discretization of multiplicity space where each element $\nn_{j}$ represents a $\md$-frequency. Further each element $\nn_{j}$ is of the following form: for each $k \in \otmd$, $\nn_{j}(k) \in \{ 1,\ceil{(1+\epstk/2)^1},\ceil{(1+\epstk/2)^2},\dots , \ceil{(1+\epstk/2)^{\bttk-1}},n \} \cup \{1,2,3,\dots ,\ceil{\frac{1}{\epstk}} \}$ for some fixed $\epst \in \R^{1\times \md}_{>0}$ and $\bttk \in O(\frac{\log \nkk}{\epstk})$ is such that $\ceil{(1+\epstk/2)^{\bttk}}\geq \nkk$, $\ceil{(1+\epstk/2)^{\bttk-1}}< \nkk$ and as before $0<\epstk<1$. Note that $\btt = |\bM|=\prod_{k=1}^{\md} \bttk \in O(\prod_{k=1}^{\md}\frac{\log \nkk}{\epstk})$.\\

\noindent \textbf{Discrete $\md$-type}:
For a sequence $\yn \in \bX^{n}$, $\phihd=\Phihd(\yn) \in \R^{\bX \times \md}$ is its \emph{discrete} $\md$-type if $\phihd_x=\ceil{\bff(\yn,x)}_{\bM}$.\\

\noindent \textbf{Discrete $\md$-profile}: For a $\md$-sequence $\yn \in \bX^{\n}$, $\phid=\Phid(\yn) \in \Z^{\bM}$ is a \emph{discrete} $\md$-profile if $\phid=(\Fd_j)_{j=1\dots \btt} $, where $\Fd_j=|\{x\in \bX ~|~\ceil{\bff(\yn,x)}_{\bM}=\nn_j \}|$ and $\nd=\sum_{x\in \bX}\ceil{\bff(\yn,x)}_{\bM} \leq (1+\epst) \times \n$ is its $\md$-length.\\

%% file: multipmlstructure.tex
\newcommand{\bppdk}[1]{\bpd''_{#1}(k)}
\subsection{Existence of Structured Approximate Solution}\label{sec:structured}
Here we show the existence of an approximate PML $\md$-distribution with a nice structure over the next several lemmas. First, we first show that one can assume the minimum non-zero probability of the PML $\md$-distribution is $\Omega(\frac{1}{\nkk^2})$ for each $k \in \otmd$ by only loosing $\expo{-O(\md)}$ in the PML objective value.

\begin{restatable}[{Minimum probability lemma}]{lemma}{lemmamind}
\label{lemmind}
For any $\md$-profile $\phi \in \Phi^{n}$, there exists a $\md$-distribution $\bp'' \in \simplexd$ such that $\bp''$ is a $\expo{-\minploss}$-approximate PML $\md$-distribution with $\min_{x \in \bX:\bppdk{x} \neq 0}\bppdk{x} \geq \frac{1}{2\nkk^2}$ for all $k\in \otmd$.
\end{restatable} 
 \begin{proof}
See \Cref{app:multipmlmind}.
 \end{proof}
Next we show that working with discrete $\md$-level sets and $\md$-frequencies doesn't significantly decrease the PML objective value. Our next lemma formally proves this statement.

 \begin{lemma}[{Probability discretization lemma}]\label{lem:probdiscd}
For any $\md$-profile $\phi \in \Phi^{n}$ and $\md$-distribution $\bp \in \simplexd$, its discrete $\md$-pseudodistribution $\bq=\disc(\bp)$ satisfies:
$$\probpml(\bp,\phi) \geq \probpml(\bq,\phi) \geq \expo{-\sum_{k=1}^{\md}\epsok \nkk}\probpml(\bp,\phi)$$
\end{lemma}
 \begin{proof}
 The first inequality is immediate because $\bq_x=\floor{\bp_x}_{\bP} \leq \bp_{x}$ for all $x \in \bX$. To show second inequality consider any $\md$-sequence $\yn \in \bX^{\n}$,
 \begin{align*}
 \bbP(\bq,\yn)& =\prod_{x \in \bX}\bq_x^{\bff(\yn,x)}=\prod_{x \in \bX}\floor{\bp_x}_{\bP}^{\bff(\yn,x)} \geq \prod_{x \in \bX} \left( \frac{\bp_x}{1+\epso} \right)^{\bff(\yn,x)} \\
 & = \left[\prod_{i=1}^{\md} \prod_{x \in \bX} \frac{1}{(1+\epsok)}^{\bff(\ynk,x)} \right] \bbP(\bp,\yn) \geq  \expo{-\sum_{k=1}^{\md}\epsok \nkk}\bbP(\bp,\yn)
 \end{align*}
 In the inequality above we use $\sum_{x \in \bX}\bff(\ynk,x)=\nkk$ for all $k \in \otmd$. Now,
 \begin{align*}
 \probpml(\bq,\phi)&=\sum_{\{\yn \in \bX^{\n}: \Phi (\yn)=\phi \}} \bbP(\bq,\yn)\geq \sum_{\{\yn \in \bX^{\n}: \Phi (\yn)=\phi \}}  \expo{-\sum_{k=1}^{\md}\epsok \nkk}\bbP(\bp,\yn)\\
 &\geq  \expo{-\sum_{k=1}^{\md}\epsok \nkk}\probpml(\bp,\phi)
  \end{align*}
 \end{proof}

Our previous lemma showed that we can work in the discretized probability space and in our next lemma we show that discretization of multiplicities also doesn't change our objective value by much. For a $\md$-sequence $\yn \in \bX^{n}$, we first provide an equivalent formulation for {the} probability of its $\md$-profile $\phi=\Phi(\yn)$ (from Equation 20 in \cite{OSZ03}, Equation 15 in \cite{PJW17}) in terms of its $\md$-type $\phih=\Phih(\yn)$. The formulations provided \cite{OSZ03}, \cite{PJW17} are for two dimensions and it is not hard to see these formulations generalize to higher dimension in the following way:

\begin{equation}\label{eqpml2d}
\probpml(\bp,\phi)=\left(\prod_{j=0}^{|\Freqn|}\frac{1}{\F_j!}\right)\binom{\n}{\phih}\sum_{\sigma \in S_{\bX}}\prod_{x\in X}\bp_{x}^{\phih_{\sigma(x)}}
=\left(\prod_{j=0}^{|\Freqn|}\frac{1}{\F_j!}\right)\cphi \sum_{\sigma \in S_{\bX}}\prod_{x\in X}\bp_{x}^{\phih_{\sigma(x)}}
%perm \left(\underbrace{\left( \bp_x^{\phih_{x'}}\right)_{x,x' \in \bX}}_{\Q}\right)
\end{equation}
%where $perm(\Q)$ is the permanent of matrix $\Q$. $$perm(\Q)=\sum_{\sigma \in S_{\bX}}\prod_{x\in X}\Q_{x,\sigma(x)}=\sum_{\sigma \in S_{\bX}}\prod_{x\in X}\bp_{x}^{\phih_{\sigma(x)}} $$
where $S_{\bX}$ is the set of all permutations of domain set $\bX$ and $\phi_0$ is the number of domain elements with frequency $(0,\dots 0)$ (unseen domain elements). The difference between \Cref{eqpml2d} and \Cref{eqlabeledd} is the index set over which they are summed.
 
\begin{restatable}[{Profile discretization lemma}]{lemma}{lemmaprofilediscd}
\label{lemprofilediscd}
For any $\md$-distribution $\bp\in \simplexd$, and a $\md$-sequence $\yn \in \bX^{\n}$:
$$\expo{\otilde\left(\sum_{k=1}^{\md}\epstk \nkk \right)}\probpml(\bp,\phi)\leq \probpml(\bp,\phid) \leq \expo{\otilde\left(\sum_{k=1}^{\md}\epstk \nkk \right)}\probpml(\bp,\phi)\footnote{Recall our $\otilde$ notation in multidimenional setting hides all $\prod_{i=1}^{\md} poly \log \nkk $ factors}$$
where $\phi=\Phi(\yn)$ and $\phid=\Phid(\yn)$ are {the} $\md$-profile and discrete $\md$-profile of $\yn$ respectively.
\end{restatable}
\begin{proof}
 Let $\phih=\Phih(\yn)$ and $\phihd=\Phihd(\yn)$ be $\md$-type and discrete $\md$-type of $\md$-sequence $\yn$ respectively. By \Cref{eqpml2d}:
 $$\probpml(\bp,\phi)=\left(\prod_{j=0}^{|\Freqn|}\frac{1}{\F_j!}\right)\binom{\n}{\phih} \left(\sum_{\sigma \in S_{\bX}}\prod_{x\in X}\bp_{x}^{\phih_{\sigma(x)}}\right)$$
Similarly: 
 $$\probpml(\bp,\phid)\defeq\left(\prod_{j=0}^{|\bM|}\frac{1}{\Fd_j!}\right)\binom{\nd}{\phihd} \left(\sum_{\sigma \in S_{\bX}}\prod_{x\in X}\bp_{x}^{\phihd_{\sigma(x)}}\right)$$
where $\phid_{0}$ is the number of unseen domain elements in profile $\phid$. Note $\phid_{0}=\phi_{0}$ because our discretization procedure does not change the number of unseen domain elements. We now analyze both objectives term by term. For any permutation $\sigma \in S_{\bX}$
\begin{align*}
\prod_{x \in \bX}\bp_{x}^{\phihd_{\sigma(x)}} & \geq \prod_{x \in \bX}\bp_{x}^{\phih_{\sigma(x)}\times(1+\epst)}=\prod_{x \in \bX}\bp_{x}^{\epst\phih_{\sigma(x)}} \prod_{x \in \bX}\bp_{x}^{\phih_{\sigma(x)}} \geq \left[\prod_{k=1}^{\md}\prod_{x \in \bX}\left(\frac{1}{2\nkk^{2}}\right)^{\epstk \phih_{\sigma(x)}(k)}\right] \prod_{x \in \bX}\bp_{x}^{\phih_{\sigma(x)}} \\
& \geq \expo{-3\sum_{k=1}^{d}\epstk \nkk \log \nkk}\prod_{x \in \bX}\bp_{x}^{\phih_{\sigma(x)}}
\end{align*}
The first inequality above follows because $\phihd_{\sigma(x)} \leq \phih_{\sigma(x)} \times (1+\epst)$ and using $\phih_{\sigma(x)} \leq \phihd_{\sigma(x)}$ we get the right hand side of the following inequality.
\begin{equation}\label{eqpro1dd}
\expo{-3\sum_{k=1}^{d}\epstk \nkk \log \nkk}\prod_{x \in \bX}\bp_{x}^{\phihd_{\sigma(x)}} \geq \prod_{x \in \bX}\bp_{x}^{\phih_{\sigma(x)}}\geq \prod_{x \in \bX}\bp_{x}^{\phihd_{\sigma(x)}}
\end{equation}
Lets consider terms $\binom{\n}{\phih}$ and $\binom{\nd}{\phihd}$, we upper bound their ratio next:
\begin{align*}
\frac{\binom{\n}{\phih}}{\binom{\nd}{\phihd}}=&\prod_{k=1}^{\md}\frac{\binom{\nkk}{\phihk{}}}{\binom{\ndkk}{\phihdk{}}}=\prod_{k=1}^{\md}\frac{\nkk!}{\ndkk!}\prod_{x\in \bX}\frac{\phihdk{x}!}{\phihk{x}!} \leq \prod_{k=1}^{\md} \prod_{x\in \bX}\frac{\floor{\bff(\ynk,x)(1+\epstk)}!}{\bff(\ynk,x)!} \\
& \leq\prod_{k=1}^{\md} \prod_{x \in \bX}(\nkk(1+\epstk))^{\epstk \bff(\ynk,x)} = \prod_{k=1}^{\md} (\nkk(1+\epstk))^{\epstk \nkk}\\
& \leq \expo{2\sum_{k=1}^{\md}\epstk \nkk \log \nkk}
\end{align*}
Next we will lower bound the ratio considered above.
\begin{align*}
\frac{\binom{\n}{\phih}}{\binom{\nd}{\phihd}}=&\prod_{k=1}^{\md}\frac{\binom{\nkk}{\phihk{}}}{\binom{\ndkk}{\phihdk{}}}=\prod_{k=1}^{\md}\frac{\nkk!}{\ndkk!}\prod_{x\in \bX}\frac{\phihdk{x}!}{\phihk{x}!} \geq \prod_{k=1}^{\md}\frac{\nkk!}{\ndkk!} \geq \prod_{k=1}^{\md} \frac{\nkk!}{\floor{\nkk(1+\epstk)}!}\\
& \geq \prod_{k=1}^{\md}(\nkk(1+\epstk))^{-\epstk \nkk} \geq \expo{-2\sum_{k=1}^{\md}\epstk \nkk \log \nkk}
\end{align*}
Combining both we get:
\begin{equation}\label{eqpro2d}
\expo{-2\sum_{k=1}^{\md}\epstk \nkk \log \nkk}\binom{\nd}{\phihd} \leq \binom{\n}{\phih}\leq \expo{2\sum_{k=1}^{\md}\epstk \nkk \log \nkk}\binom{\nd}{\phihd}
\end{equation}
For final term consider all $\md$-frequencies generated by domain elements $x$ in $\md$-sequence $\yn$.  Observe that during our discretization procedure all $\md$-frequencies less than $\ceil{\frac{1}{\epst}}$ are never affected and we upper bound the number of $\md$-frequencies that change.

Analogous to proof in one dimension, for each $k \in \otmd$, the number of domain elements $x \in \bX$ with $\bff(\ynk,x)>\ceil{\frac{1}{\epstk}}$ is less than $\epstk \nkk$. Further, the number of domain elements $x \in \bX$ with $\bff(\ynk,x)>\ceil{\frac{1}{\epstk}}$ for any $k \in \otmd$ is less than $\sum_{k=1}^{\md}\epstk \nkk$. The previous statement upper bounds $\sum_{\{j \in \otbtt\} ~|~\exists k \in \otmd \text{ with }\freq_{j}(k)>\ceil{\frac{1}{\epstk}}\} } \phi_{i} \leq \sum_{k=1}^{\md}\epstk \nkk$. This further implies $\sum_{\{j \in \otbtt\} ~|~\exists k \in \otmd \text{ with }\nn_{j}(k)>\ceil{\frac{1}{\epstk}}\} } \phid_{j} \leq \sum_{k=1}^{\md}\epstk \nkk$. Combining the previous reasoning with the fact that all $\md$-frequencies less than $\ceil{\frac{1}{\epst}}$ are never changed we get the following inequality.
\begin{align*}
1\leq \frac{ \prod_{j=0}^{|\bM|} \Fd_j!}{\prod_{j=0}^{|\Freqn|}\F_j!} \leq \prod_{\{j \in \otbtt\} ~|~\exists k \in \otmd \text{ with }\nn_{j}(k)>\ceil{\frac{1}{\epstk}}\}} \Fd_j! \leq \expo{\left(\sum_{k=1}^{\md}\epstk \nkk \right)\log \left( \sum_{k=1}^{\md} \epstk \nkk \right)}
\end{align*}
Combining previous inequality with \cref{eqpro1dd}, \cref{eqpro2d} we have our result.
\end{proof}

Our next corollary captures the impact of discretizing both probabilities and multiplicities.
 \begin{cor}[{Discretization lemma}]\label{cordiscd}
For any $\md$-distribution $\bp\in \simplexd$, and a $\md$-sequence $\yn \in \bX^{\n}$. If $\bq=\disc(\bp)$ is {the} discrete $\md$-distribution of $\bp$ then,
$$\frac{1}{\alpha}\probpml(\bp,\phi)\leq \probpml(\bq,\phid) \leq \alpha \probpml(\bp,\phi) \text{ for }\alpha=\expo{\otilde\left(\disclossd \right)}$$
where $\phi=\Phi (\yn)$ and $\phid=\Phid(\yn)$ are {the} $\md$-profile and {discrete} $\md$-profile of $\yn$ respectively.
\end{cor}
 \begin{proof}
 Corollary follows immediately by combining \Cref{lem:probdiscd} and \Cref{lemprofilediscd}.
 \end{proof}
 
%The Discretization Lemma above motivates us in defining a new objective function to work with which we study next.
{The discretization lemma above motivates the definition of a new objective function which we introduce and study next.}
 \subsection{Discrete PML Optimization}
Here we define a new optimization problem that can be solved efficiently and returns a $\md$-distribution which has a good approximation to the PML objective value. First we define the discrete profile maximum likelihood which is just the PML objective maximized over discrete $\md$-pseudodistributions.
\begin{defn}[Discrete profile maximum likelihood]
Let $\yn \in \bX^{\n}$ be any $\md$-sequence, $\phi=\Phi (\yn)$ and $\phid=\Phid(\yn)$ be its $\md$-profile and discrete $\md$-profile respectively, a \emph{Discrete Profile Maximum Likelihood} (DPML) $\md$-pseudodistribution $\bq_{dpml,\phid}$ is:
 \begin{equation}\label{eqdpmld}
 \bq_{dpml,\phid}\defeq\argmax_{\bq \in \dsimplexd} \probpml(\bq,\phid)
 \end{equation}
$\probpml(\bq_{dpml,\phid},\phid)$ is {the} maximum objective value.
\end{defn}
\begin{cor}[DPML is an approximate PML]\label{corfinaldiscd}
For any $\md$-sequence $\yn \in \bX^{\n}$, $~~~~~\probpml(\bq_{dpml,\phid},\phid) \geq \expo{-\otilde\left(\disclossd \right)} \probpml(\bp_{pml,\phi},\phi)$
\end{cor}
\begin{proof}
Note that $\bq_{pml,\phi}=\disc(\bp_{pml,\phi})$ is a discrete $\md$-pseudodistribution.
The result follows from \Cref{cordiscd} applied to $\bp_{pml,\phi}$.
  \end{proof}

In {the} next two lemmas we rephrase the DPML optimization problem in forms that are amenable to convex relaxation. To do this, we introduce some new notation.\\
\noindent $\bullet$ Let $\pvec \in \R^{\boo \times \md}$ be the matrix with rows indexed between $1$ to $\boo$ and $i$th row is equal to $\md$-level set $\pp_i \in \bP$. Also let $\mvec \in \R^{\bttpo \times \md}$ be the vector with rows indexed between $0$ to $\btt$. Its zeroth row (denoted by $\nn_{0}$) is equal to $\md$-frequency $(0,\dots 0)$ and $j$th row is equal to $\md$-frequency $\nn_j \in \bM$. We use $\nn(k)$ and $\pp(k)$ to denote the $k$th column of matrix $\nn$ and $\pp$ respectively.

\noindent $\bullet$ Let $\bk \in \Z^{\boo \times \bttpo}$ be a variable matrix and we use $\bk_{ij}$ for $i \in \otboo,j \in \ztbtt$ to denote elements of this matrix. As in the case for vector $\nn$, our second index $j$ of variable matrix $\bk$ starts at $0$ and not at $1$. Here the variable $\bk_{ij}$ counts the number of domain elements $x \in \bX$ with $\md$-level set $\pp_i$ and have $\md$-frequency equal to $\nn_j$. $\bk_{i,0}$ is counting the number of domain elements $x \in \bX$ with $\md$-level set $\pp_i$ and $\md$-frequency equal to $(0,\dots 0)$. We use function $\log \pp(k)$ and $\log \pp$ to perform entrywise operations returning entities of same dimension as $\pp(k)$ and $\pp$ respectively with $\log$ applied on every entry.

\noindent $\bullet$ For any matrix $\vvec$ and set $S$, we use $\vvec_{S}$ to denote the matrix with $|S|$ rows corresponding to index set $S$.

\noindent $\bullet$ For a discrete $\md$-profile $\phid=(\Fd_{j})_{j=1\dots \btt}$ (corresponding to $\md$-sequence $\yn$), define:\\
 $~~~~\bK_{\phid}\defeq\{\bk \in \Z^{\boo \times \bttpo}~\Big|~ ~(\bk^{T}\onevec)_{\otbtt} = \Fd, \text{ and } (\bk\onevec)^{T} \pvec \leq 1 \}$\\
 Note in the expression above $(\bk\onevec)^{T} \pvec$ is a $\md$-tuple and $(\bk\onevec)^{T} \pvec \leq 1$ means each entry of this $\md$-tuple is less than 1 (as described in the preliminaries section).
 
\noindent $\bullet$ For a discrete $\md$-profile $\phid=(\Fd_{j})_{j=1\dots \btt}$ (of $\yn$) and a discrete $\md$-pseudodistribution $\bq$, also define:\\
$~~~~\bK_{\bq,\phid}\defeq\{\bk \in \Z^{\boo \times \bttpo}~\Big|~ ~(\bk^{T}\onevec)_{\otbtt} = \Fd, \text{ and } \bk \onevec=\levelq \}$
where $\levelq \in \R^{\boo}$ and $\levelq_i$ denote the number of domain elements with $\md$-level set $\pp_i \in \bP$ in $\md$-pseudodistribution $\bq$.\\

One of the most important advantages of $\md$-level set and $\md$-frequency discretization we described earlier is that many $\md$-types in the set $\{\phih ~|~ \Phi(\phih)=\phid\}$ share the same probability value of being observed and our goal is to group them using the $\bk_{ij}$ variables. Exploiting this idea, we next give a different formulation for {the} DPML objective.
\begin{lemma}[DPML objective reformulation]\label{DPMLequalityd} 
For any discrete $\md$-pseudodistribution $\bq \in \simplexd$ and discrete $\md$-profile $\phid \in \Phi^{\nd}$:
\begin{equation}\label{eq:formulation2d}
\probpml(\bq,\phid) = \cphid \sum_{\bk\in \bK_{\bq,\phid}} \prod_{i=1}^{\boo}\left(\pp_i^{(\bk \mvec)_i}\frac{(\bk\onevec)_i!}{\prod_{j=0}^{\btt}\bk_{ij}!}\right)
\end{equation}
\end{lemma}
\begin{proof}
Recall from \Cref{eqlabeledd}
$$\probpml(\bq,\phid) = \cphid \sum_{\{\phih ~|~ \Phi(\phih)=\phid\}} \prod_{x \in X}\bq_x^{\phih_x}$$
For convenience, we call a $\md$-type $\lphi$ valid if it belongs to set $\{\phih ~|~ \Phi(\phih)=\phid\}$. Recall variable $\bk_{ij}$ represents the number of domain elements with $\md$-level set $\pp_{i}$ and have $\md$-frequency equal to $\nn_j$. In this representation and for the discrete $\md$-pseudodistribution $\bq$, each valid $\md$-type $\lphi$ corresponds to the following unique variable assignment $\bk \in \bK_{\bq,\phid}$:
$$\bk_{ij}=|\{x \in \bX ~| ~\bq_x=\pp_i \text { and } \lphi_x=\nn_j\}|$$ 
and from the expression above it is not hard to write the exact expression for the probability term associated with the valid $\md$-type $\lphi$:
\begin{equation}\label{eq:probtyped}
\begin{split}
\prod_{x \in \bX}\bq_x^{\phih_x}&=\prod_{i=1}^{\boo}\prod_{j=0}^{\btt}\prod_{\{x \in \bX|\bq_{x}=\pp_{i} \text{ and }\phih_{x}=\nn_{j}\}}\prod_{k=1}^{\md}\pp_{i}(k)^{\nn_{j}(k)}=\prod_{i=1}^{\boo}\prod_{j=0}^{\btt} (\prod_{k=1}^{\md}\pp_{i}(k)^{\nn_{j}(k)})^{\bk_{ij}}\\
&=\prod_{i=1}^{\boo}\prod_{j=0}^{\btt} (\pp_{i}^{\nn_{j}})^{\bk_{ij}}=\prod_{i=1}^{\boo} \pp_{i}^{\sum_{j=0}^{\btt}\nn_{j}\bk_{ij}}=\prod_{i=1}^{\boo}\pp_{i}^{(\bk \mvec)_i}
\end{split}
\end{equation}

For any variable assignment $\bk$, it is clear from the middle term in Equation \ref{eq:probtyped} that all valid $\md$-types $\lphi$ associated with $\bk$ share the same probability value of being observed. With this observation, it is now enough to argue about the number of valid $\md$-types associated with a variable assignment $\bk$ to prove our lemma. We make this argument next by constructing all valid $\md$-types associated with $\bk$.

First consider all domain elements with a fixed $\md$-level set $\pp_{i}$ and number of such elements is equal to $\sum_{j=0}^{\btt}\bk_{ij}$. We can now generate part of a valid $\md$-type corresponding to the domain elements with $\md$-level set equal to $\pp_{i}$ by picking any partition of these $\sum_{j=0}^{\btt}\bk_{ij}$ domain elements into groups of sizes $\{ \bk_{ij}\}_{j \in \ztbtt}$. This corresponds to multinomial coefficient and therefore the number of types associated with $\bk$ is just:
$$\frac{(\bk \onevec)_i!}{\prod_{j=0}^{\btt}\bk_{ij}!}$$
Here we only generated partial valid $\md$-types corresponding to domain elements with $\md$-level set equal to $\pp_{i}$. To generate a full valid $\md$-type we just need to combine these partial valid $\md$-types generated for each $\md$-level set $\pp_{i}$. Let $S_{\bk}$ denote all such full valid $\md$-types associated with a variable assignment $\bk$ and generating a full valid $\md$-type corresponds to groups (for each $\md$-level set $\pp_{i}$) of independent possibilities considered conjointly. Further the cardinality of set $S_{\bk}$ is just the multiplication of cardinalities of each of these groups and is explicitly written below, 
$$|S_{\bk}|=\prod_{i=1}^{\boo}\frac{(\bk\onevec)_i!}{\prod_{j=0}^{\btt}\bk_{ij}!}$$
We are almost done and all we do next is formally derive the expression in our lemma statement to complete the proof.  From \Cref{eqlabeledd},
\begin{align*}
\probpml(\bq,\phid) & =\cphid \sum_{\{\phih | \Phi(\phih)=\phid\}} \prod_{x \in X}\bq_x^{\phih_x}=\cphid \sum_{\bk\in \bK_{\bq,\phid}} \sum_{\{\phih \in S_{\bk}\}} \prod_{x \in X}\bq_x^{\phih_x}\\
&=\cphid \sum_{\bk\in \bK_{\bq,\phid}} \sum_{\{\phih \in S_{\bk}\}} \prod_{i=1}^{\boo}\pp_{i}^{(\bk \mvec)_i}=\cphid \sum_{\bk\in \bK_{\bq,\phid}}  |S_{\bk}| \prod_{i=1}^{\boo}\pp_{i}^{(\bk \mvec)_i} \\
& = \cphid \sum_{\bk\in \bK_{\bq,\phid}} \prod_{i=1}^{\boo}\left(\pp_i^{(\bk \mvec)_i}\frac{(\bk\onevec)_i!}{\prod_{j=0}^{\btt}\bk_{ij}!}\right)
\end{align*}
\end{proof}

\begin{lemma}[DPML objective relaxed]\label{lem:DPMLupperboundd}
For any $\md$-sequence $\yn \in \bX^{\n}$, and a discrete $\md$-pseudodistribution $\bq \in \simplexd$ the DPML objective can be upper bounded by: 
\begin{equation}\label{eq:w1d}
\probpml(\bq,\phid) \leq \cphid \sum_{\bk\in \bK_{\phid}} \prod_{i=1}^{\boo}\left(\pp_i^{(\bk \mvec)_i}\frac{(\bk\onevec)_i!}{\prod_{j=0}^{\btt}\bk_{ij}!}\right)
\end{equation}
where $\phid = \Phid(\yn) \in \Phi^{\nd}$ is discrete $\md$-profile of $\yn$.
\end{lemma}
 \begin{proof}
The proof follows because $ \bK_{\bq,\phid}\subseteq \bK_{\phid} $ and invoking \Cref{DPMLequalityd}.
 \end{proof}
 
We are half way through in defining our final optimization problem which exhibits efficient algorithms. In our final optimization problem we just optimize over one term in the set $\bK_{\phid}$ instead of working with summation over all the terms and next two lemmas serve as the motivation for working with single term over the summation of terms by showing that the optimizing $\md$-pseudodistribution of our final optimization problem is still an approximate PML $\md$-distribution.
  \begin{lemma}[Cardinality of $\bK_{\phid}$]\label{lemcardd}
For any $\md$-sequence $\yn \in \bX^{\n}$ and its associated discrete $\md$-profile $\phid = \Phid(\yn)$: $$|\bK_{\phid}| \leq  \expo{O\left(\grouplossepsd \right)}.$$
  \end{lemma}
 \begin{proof}
   $\bK_p$ is a set of vectors in $\Z^{\boo \times \bttpo}$ and because of \Cref{lemmind} combined with the constraint $(\bk\onevec)^{T} \pvec\leq 1$, each $\bk_{ij}$ takes only positive integer values less than $\min_{k \in \otmd}2\nkk^{2}$. The lemma statement follows by substituting the values of $\boo$ and $\btt$.
  \end{proof}

As described earlier \Cref{lemcardd} motivates us to consider the following objective, define:
$$\probsdpml(\bk)\defeq \prod_{i=1}^{\boo}\left(\pp_i^{(\bk \mvec)_i}\frac{(\bk \onevec)_i!}{\prod_{j=0}^{\btt}\bk_{ij}!}\right)$$
It is important to note that there is a discrete $\md$-pseudodistribution $\bq_{\bk}$ that correspond to each variable assignment $\bk\in \bK_{\phid}$. The description of this $\md$-distribution is as follows: For each $i \in \otboo$, the number of domain elements that have $\md$-level set $\pp_{i}$ in $\bq$ is equal to $(\bk \onevec)_{i}$. This description only provides non zero $\md$-level sets and also does not provide any labels, however it is sufficient for estimating all symmetric properties mentioned in this paper.
  \begin{defn}[Single discrete profile maximum likelihood]
  For any $\md$-sequence $\yn \in \bX^{\n}$ and its associated discrete $\md$-profile $\phid = \Phid(\yn) \in \Phi^{\nd}$, a \emph{Single Discrete Profile Maximum Likelihood} (SDPML) $\md$-pseudodistribution $\bq_{sdpml,\phid}$ is:
 \begin{equation}\label{eq:sdpmld}
 \bk_{sdpml,\phid}\defeq\argmax_{\bk\in \bK_{\phid}}\cphid \probsdpml(\bk)=\argmax_{\bk\in \bK_{\phid}}\probsdpml(\bk)
 \end{equation}
 and $\bq_{sdpml,\phid}$ is the $\md$-pseudodistribution corresponding to $\bk_{sdpml,\phid}$.
  \end{defn}
  
  \begin{lemma}[SDPML relationd to PML]\label{lemsinglediscd} 
  For any $\md$-sequence $\yn \in \bX^{\n}$,
  $$\cphid\probsdpml(\bk_{sdpml,\phid}) \geq \expo{-\otilde\left(\disclossd +\grouplossepsdtilde \right) }\probpml(\bp_{pml,\phi},\phi)$$
  where $\phi = \Phi(\yn)$ and $\phid = \Phid(\yn)$ are $\md$-profile and discrete $\md$-profile associated with $\yn$.
  \end{lemma}
  \begin{proof} $~~~~~\cphid \probsdpml(\bk_{sdpml,\phid}) \geq \cphid \probsdpml(\bk_{dpml,\phid}) \geq \expo{-O\left(\grouplossepsd \right)}\probpml(\bq_{dpml,\phid},\phid)~~~~~$\\
  $$\geq  \expo{-\otilde\left(\disclossd +\grouplossepsdtilde \right)}\probpml(\bp_{pml,\phi},\phi)$$
  The second inequality follows from \Cref{lemcardd}, \ref{lem:DPMLupperboundd} and last follows from \Cref{corfinaldiscd}.
 \end{proof}

\subsection{Convex relaxation of SDPML}\label{subsec:convexrelaxd}
We showed in the previous subsection that the SDPML objective is a good approximation to the PML objective. However the objective function of SDPML is defined only over the integers and in this subsection we present a convex relaxation of SDPML.

First, we consider the feasible set $\bK_{\phid}$ of SDPML, which is the following integral polytope 
\[
\bK_{\phid}=\{\bk \in \Z^{\boo \times \bttpo}~\Big|~ ~(\bk^{T}\onevec)_{\otbtt} = \Fd, \text{ and } (\bk\onevec)^{T} \pvec \leq 1 \} ~.
\]
%Instead we work with the following natural relaxation:
We relax the integer constraint on variables $\bk_{ij}$:
\begin{equation}\label{setrelaxd}
\bK^{f}_{\phid}\defeq\{\bk \in \R^{\boo \times \bttpo}~\big|~ (\bk^{T}\onevec)_{\otbtt} = \Fd, \text{ and } (\bk\onevec)^{T} \pvec \leq  1 \} ~.
\end{equation}
In the later subsections, we show how we deal with these fractional solutions by presenting a rounding algorithm with a good approximation ratio.

Secondly, we relax the objective function of SDPML itself. The objective of SDPML is defined only on the integral set. We next define a continuous relaxation of this objective function which is also log-concave. To do so, we use an approximation of the factorial function (similar to Stirling's approximation) which handles $0!$ terms as well. 
%(Note sterlings approximation doesn't give a finite approximation for $0!$) 
%which is summarized in the next lemma. 
We use the following function as the continuous proxy of the SDPML objective (using the convention that $0 \log 0 = 0$):
\begin{equation}\label{contsdpmlorigd}
\begin{split}
 \bg(\bk)&\defeq\prod_{i=1}^{\boo}\left(\pp_i^{(\bk \mvec)_i}\frac{\expo{(\bk \onevec)_i \log (\bk \onevec)_i-(\bk \onevec)_i}}{\prod_{j=0}^{\btt}\expo{\bk_{ij}\log \bk_{ij}-\bk_{ij}}}\right)\\
 &=\left[\prod_{i=1}^{\boo}\prod_{j=0}^{\btt} \prod_{k=1}^{\md}\pp_{i}(k)^{\nn_{j}(k)\bk_{ij}}\right] \expo{\sum_{i=1}^{\boo}(\bk \onevec)_i \log (\bk \onevec)_i-\sum_{i=1}^{\boo}\sum_{j=0}^{\btt}\bk_{ij}\log \bk_{ij}}\\
 &=\left[\prod_{k=1}^{\md}\expo{\pp(k)^{T}\bk \nn(k)}\right] \expo{\sum_{i=1}^{\boo}(\bk \onevec)_i \log (\bk \onevec)_i-\sum_{i=1}^{\boo}\sum_{j=0}^{\btt}\bk_{ij}\log \bk_{ij}}\\
 & =\expo{\sum_{k=1}^{\md}\log \pp(k)^{T}\bk \nn(k) + \sum_{i=1}^{\boo}(\bk \onevec)_i \log (\bk \onevec)_i-\sum_{i=1}^{\boo}\sum_{j=0}^{\btt}\bk_{ij}\log \bk_{ij}}\\
 &=\expo{\tr(\log \pp^{T}\bk \nn) + \sum_{i=1}^{\boo}(\bk \onevec)_i \log (\bk \onevec)_i-\sum_{i=1}^{\boo}\sum_{j=0}^{\btt}\bk_{ij}\log \bk_{ij}}
 \end{split}
 \end{equation}
The lemma below states that continuous version is not far from the actual SDPML objective. 
\begin{restatable}[$\bg(\cdot)$ approximates SDPML objective]{lemma}{lemmastirlingsapproxd} 
\label{app:sterlingsapproxd}
For any $\md$-sequence $\yn \in \bX^{\n}$ and its associated discrete $\md$-profile $\phid = \Phid(\yn) \in \Phi^{\nd}$. If $\bk \in \bK_{\phid}$, then
$$\expo{-O\left(\prod_{k=1}^{\md}\frac{\log^3 \nkk}{\epsok \epstk} \right)} \bg(\bk) \leq \probsdpml(\bk)\leq \expo{O\left(\prod_{k=1}^{\md}\frac{\log^2 \nkk}{\epsok}\right)} \bg(\bk)$$
 \end{restatable}
\begin{proof}
 By Stirling's approximation for all integer $n \geq 1$:
 $$\sqrt{2\pi}\leq \frac{n!}{\sqrt{n}\expo{n\log n-n}} \leq e $$
 We slightly use a weaker version of this inequality that holds all integers $n\geq 0$,
 $$1\leq \frac{n!}{\expo{n\log n-n}} \leq e \sqrt{n+1}$$
 \begin{align*}
 \frac{\probsdpml(\bk)}{\bg(\bk)}&=\prod_{i=1}^{\boo}\left( \frac{(\bk \onevec)_i!}{\expo{(\bk \onevec)_i\log (\bk \onevec)_i-(\bk \onevec)_i}} \prod_{j=0}^{\btt}\frac{\expo{\bk_{ij}\log \bk_{ij} -\bk_{ij}}}{\bk_{ij}!} \right)\\
& \leq \prod_{i=1}^{\boo}e \sqrt{1+(\bk \onevec)_i} \leq  \expo{O\left(\prod_{k=1}^{\md}\frac{\log^2 \nkk}{\epsok} \right)}
\end{align*}
In the final inequality we used the fact that each $i \in \otboo$, $(\bk \onevec)_{i}\leq \min_{k \in \otmd}2\nkk^{2}$ (Lemma \ref{lemmind} combined with the constraint $(\bk\onevec)^{T} \pvec\leq 1$ ensures this fact) and substituted the value of $\boo$. Also,
 \begin{align*}
 \frac{\probsdpml(\bk)}{\bg(\bk)} & \geq \prod_{i=1}^{\boo}\prod_{j=0}^{\btt}\frac{\expo{\bk_{ij}\log \bk_{ij} -\bk_{ij}}}{\bk_{ij}!} \geq \prod_{i=1}^{\boo}\prod_{j=0}^{\btt} \frac{1}{e\sqrt{1+\bk_{ij}}}\\
 &\geq \expo{-O\left(\prod_{k=1}^{\md}\frac{\log^3 \nkk}{\epsok \epstk} \right)}
 \end{align*}
\end{proof}

A key fact about function $\bg(\bk)$ is that it is log-concave, so we can apply optimization machinery from convex optimization.
 
 \begin{restatable}{lemma}{lemmalogconcaved}
 Function $\bg(\bk)$ is log-concave in $\bk$.
 \end{restatable}
\begin{proof}
Taking $\log$ on both sides of \Cref{contsdpmlorigd} we get,
$$\log \bg(\bk) =\tr(\log \pp^{T}\bk \nn) + \sum_{i=1}^{\boo}(\bk \onevec)_i \log (\bk \onevec)_i-\sum_{i=1}^{\boo}\sum_{j=0}^{\btt}\bk_{ij}\log \bk_{ij}$$
The first term $\tr(\log \pp^{T}\bk \nn)$ is linear in $\bk$ and refer \Cref{lemfconvex} for the concavity of the second term. Combining both we get, $\log \bg(\bk)$ is a sum of linear plus concave term and is therefore concave. Therefore, the function $\bg(\bk)$ is $\log$ concave.
\end{proof}
 
Maximizing log concave objective function $\bg(\cdot)$ over the relaxed convex set $\bK^{f}_{\phid}$ is a convex optimization problem and can be solved efficiently. Below is the convex relaxation of our SDPML objective which can be solved efficiently as summarized by our next theorem.
\begin{equation}\label{eq:convoptoriginald}
\argmax_{\bk\in \bK^{f}_{\phid}}\log \bg(\bk)
\end{equation}

\begin{restatable}[Solver for convex relaxation to SDPML]{thm}{theoremcuttingplaned}
\label{thm:cutting_plane_used}
Optimization problem \ref{eq:convoptoriginald} can solved in time $\runningtime$
\end{restatable}
\begin{proof}
The optimization problem \ref{eq:convoptoriginald} is already in the form of optimization problem studied for one dimension in \Cref{app:cuttingplane}. To invoke the result in \Cref{app:cuttingplane} all we need is a lower bound on the minimum eigenvalue of matrix $\ma^{T}\ma$, where $\ma$ is the constraint matrix when the optimization problem \ref{eq:convoptoriginald} is written in the vector form (described in \Cref{app:cuttingplane}). We state this constraint matrix $\ma$ for the optimization problem \ref{eq:convoptoriginald} and provide lower bound on the minimum eigenvalue of matrix $\ma^{T}\ma$ in \Cref{app:eigenlb}. The number of variables in the optimization problem \ref{eq:convoptoriginald} is $\boo \times \btt$ and the number of constraint is $\btt +\md\leq 2\btt$. In this notation of \Cref{app:cuttingplane}, the value of parameters $b_{1}=\boo$ and $b_{2}=\btt$ and the running time we get for the optimization problem \ref{eq:convoptoriginald} is that stated in the lemma statement.
\end{proof} 
%%Please refer Appendix section \ref{app:cuttingplane}.

 \subsection{Algorithm and Runtime Analysis}
In this section we give an algorithm to find a $\md$-distribution that approximates PML objective and our analysis in previous sections suggests that it suffices to find a $\md$-distribution that approximates SDPML objective, which we replaced by a convex proxy. We now present an algorithm that takes an optimal solution to this convex proxy and produces a $\md$-distribution that approximates PML objective. %a $\md$-distribution that approximates SDPML here with continous objective .
Recall that $\bK^{f}_{\phid}\defeq\{\bk \in \R^{\boo \times \bttpo}~\big|~ (\bk^{T}\onevec)_{\otbtt} = \phid, \text{ and } (\bk\onevec)^{T} \pvec \leq  1 \}$.
 \begin{algorithm}[H]
\caption{Algorithm for approximate PML}\label{euclidd}
\begin{algorithmic}[1]
\Procedure{Approximate PML}{}
\State Solve $\bk'=\argmax_{\bk\in \bK^{f}_{\phid}}\bg(\bk)$. 
\State Round the fractional solution $\bk'$ to a integral solution $\bk \in \bK_{\phid}$.
\State Construct the discrete $\md$-pseudodistribution $\bq_{X}$ corresponding to $\bk$.
\State \textbf{return} $\frac{\bq_{X}}{\|\bq_{X}\|_1}$
\EndProcedure
\end{algorithmic}
\end{algorithm}

\begin{algorithm}[H]
\caption{Rounding algorithm}
\begin{algorithmic}[1]
\Procedure{Rounding}{$\bk'$}
\State Define $\bk=\mzero^{(\boo+\btt)\times \bttpo}$
\State $\bk_{ij}= \floor{\bk'_{ij}} \in \Z \quad \forall i \in \otboo,j \in \ztbtt$\Comment{$\bk \notin \bK_{\phid}$ and we fix it next.}
%\State $\br_j=\sum_{i=1}^{\boo}\bk'_{ij}-\sum_{i=1}^{\boo}\bk_{ij}=\F_{\nn_j}-\sum_{i=1}^{\boo}\bk_{ij} \in \Z \quad \forall j \in \otbtt$\Comment{$\br_j <\boo$}
%\State $\bs_i=\sum_{j=0}^{\btt}\bk'_{ij}-\sum_{j=0}^{\btt}\bk_{ij} \quad \forall i \in \otboo$\Comment{$\bs_i < \btt$}.
\For{$j \in \otbtt$}
\State Create a new $\md$-level set $\pvec_{\boo+j}=\frac{\sum_{i=1}^{\boo}(\bk'_{ij}-\bk_{ij})\pvec_{i}}{\sum_{i=1}^{\boo}(\bk'_{ij}-\bk_{ij})}$
\State Assign $\bk_{\boo+j,j}=\sum_{i=1}^{\boo}(\bk'_{ij}-\bk_{ij})=\Fd_{j}-\sum_{i=1}^{\boo}\bk_{ij} \in \Z$
\EndFor
\State \textbf{return} $\bk$
\EndProcedure
\end{algorithmic}
\label{alg:roundingd}
\end{algorithm} 

The solution $\bk$ returned by the rounding procedure is defined on an extended discretized $\md$-probability space $\bP'$, where $\bP' \defeq  \bP \cup \{\pvec_{\boo+j} \}_{j \in \otbtt}$. To derive the relation between solution $\bk$ and PML objective value we need to extend some definitions studied earlier. First, we define $\pvecext$ as the matrix whose rows are exactly the elements of $\bP'$ and we call it the extended $\md$-level set matrix. Note we still use $\pvec_{i}$ for all $i\in [1,\boo+\btt]$ to refer rows of $\pvecext$.
Further, for any $\md$-pseudodistribution $\bq$ with $\bq_{x} \in \bP'$ for all $x \in \bX$ (we call it extended discrete $\md$-pseudodistribution) and discrete $\md$-profile $\phid$, we first define following extensions of sets $\bK_{\bq,\phid}$ and $\bK_{\phid}$,
$$\bK^{ext}_{\bq,\phid}\defeq\{\bk \in \Z^{(\boo+\btt) \times \bttpo}~\Big|~ ~(\bk^{T}\onevec)_{\otbtt} = \Fd, \text{ and } \bk \onevec=\levelq \}$$
$$\bK^{ext}_{\phid}\defeq\{\bk \in \Z^{(\boo+\btt) \times \bttpo}~\Big|~ ~(\bk^{T}\onevec)_{\otbtt} = \Fd, \text{ and } \pvecext^{T} \bk \onevec \leq  1 \}$$
where $\levelq \in \R^{\boo+\btt}$ and $\levelq_i$ denote the number of domain elements with $\md$-level set $\pp_i \in \bP'$. 

 Further by \Cref{DPMLequalityd}, for any extended discrete $\md$-pseudodistribution $\bq$ and a discrete $\md$-profile $\phid$, the following equality holds,
\begin{equation}\label{eq:formulation3d}
\probpml(\bq,\phid) = \cphid \sum_{\bk\in \bK^{ext}_{\bq,\phid}} \prod_{i=1}^{\boo+\btt}\left(\pp_i^{(\bk \mvec)_i}\frac{(\bk\onevec)_i!}{\prod_{j=0}^{\btt}\bk_{ij}!}\right)
\end{equation}
Similarly for any $\bk \in \bK^{ext}_{\bq,\phid}$, below are the natural extension of definitions of functions $\probsdpml(\cdot)$ and $\bg(\cdot)$,
$$\probsdpml(\bk)\defeq \prod_{i=1}^{\boo+\btt}\left(\pp_i^{(\bk \mvec)_i}\frac{(\bk \onevec)_i!}{\prod_{j=0}^{\btt}\bk_{ij}!}\right) \quad \bg(\bk)\defeq\prod_{i=1}^{\boo+\btt}\left(\pp_i^{(\bk \mvec)_i}\frac{\expo{(\bk \onevec)_i \log (\bk \onevec)_i-(\bk \onevec)_i}}{\prod_{j=0}^{\btt}\expo{\bk_{ij}\log \bk_{ij}-\bk_{ij}}}\right)$$

We are now ready to analyze our rounding algorithm. First we provide some interesting properties solution $\bk$ returned by our rounding procedure satisfies,

\begin{clm}\label{clmroundingd} The solution $\bk \in \Z^{(\boo+\btt) \times \bttpo}$ returned by rounding procedure (\ref{alg:rounding}) above satisfies:
\begin{enumerate}
\item $(\bk' \onevec)_{i}-\bttpo \leq (\bk \onevec)_{i} \leq (\bk' \onevec)_{i} \quad \forall i \in \otboo$
\item $\bk \in \bK^{ext}_{\phid}$.
\end{enumerate}
\end{clm}
\begin{proof}
Claims (1) follows because $\bk'_{ij}-1\leq \bk_{ij} \leq \bk'_{ij}$ for all $i \in \otboo,j \in \ztbtt$. Now note $\sum_{i=1}^{\boo+\btt}\bk_{ij}=\sum_{i=1}^{\boo}\bk'_{ij}=\Fd_{\nn_j} \quad \forall j \in \otbtt$ because of the adjustments made by new level sets. Further, 
\begin{align*}
\pvecext^{T} \bk \onevec=\sum_{i=1}^{\boo+\btt}\pvec_{i}(\bk \onevec)_i & = \sum_{i=1}^{\boo}\pvec_{i}(\bk \onevec)_i+\sum_{j=1}^{\btt}\pvec_{\boo+j}(\bk \onevec)_{\boo+j}\\
&=\sum_{i=1}^{\boo}\pvec_{i}(\bk \onevec)_i+\sum_{j=1}^{\btt}\sum_{i=1}^{\boo}(\bk_{ij}'-\bk_{ij})\pvec_i\\
&=\sum_{j=1}^{\btt}\sum_{i=1}^{\boo}\bk_{ij}'\pvec_i =\pvec^{T} \bk' \onevec \leq  1
\end{align*}
The final inequality follows because $\bk' \in \bK^{f}_{\phid}$ and therefore $\bk \in \bK^{ext}_{\phid}$ and Claim (2) follows.
\end{proof}

The solution $\bk$ returned by (\ref{alg:roundingd}) always belongs to $\bK^{ext}_{\phid}$, further values $\probsdpml(\bk)$ and $\bg(\bk)$ are close to each other and we summarize this result in our next lemma.
\begin{restatable}{lemma}{lemroundsterd} 
\label{app:roundsterd}
For any $\bk \in \bK^{ext}_{\phid}$ returned by rounding procedure above satisfies:
\begin{equation}\label{relationd}
\expo{-O\left(\prod_{k=1}^{\md}\frac{\log^3 \nkk}{\epsok \epstk} \right)} \bg(\bk) \leq \probsdpml(\bk)\leq \expo{O\left(\prod_{k=1}^{\md}\frac{\log^2 \nkk}{\epsok}+\prod_{k=1}^{\md}\frac{\log^2 \nkk}{\epstk}\right)} \bg(\bk)
\end{equation}
 \end{restatable}
\begin{proof}
For all integers $n\geq 0$, recall the weaker version of sterlings approximation we used earlier ,
 $$1\leq \frac{n!}{\expo{n\log n-n}} \leq e \sqrt{n+1}$$
 Now, 
% $$\termo=\prod_{i=1}^{\boo+\btt}\left( \frac{(\bk \onevec)_i!}{\expo{(\bk \onevec)_i\log (\bk \onevec)_i-(\bk \onevec)_i}} \prod_{j=1}^{\btt}\frac{\expo{\bk_{ij}\log \bk_{ij} -\bk_{ij}}}{\bk_{ij}!} \right)$$
 \begin{align*}
 \frac{\probsdpml(\bk)}{\bg(\bk)}&=\prod_{i=1}^{\boo+\btt}\left( \frac{(\bk \onevec)_i!}{\expo{(\bk \onevec)_i\log (\bk \onevec)_i-(\bk \onevec)_i}} \prod_{j=0}^{\btt}\frac{\expo{\bk_{ij}\log \bk_{ij} -\bk_{ij}}}{\bk_{ij}!} \right)
% &=\termo \times \termt\\
 \end{align*}
 and 
 $$ \frac{\probsdpml(\bk)}{\bg(\bk)}\leq \prod_{i=1}^{\boo+\btt}e \sqrt{1+(\bk \onevec)_i} \leq \left(e\sqrt{1+\nkk^2} \right)^{\boo+\btt} \leq \expo{O(\log \nkk) (\boo+\btt)}$$
Now $\bP'=\bP \cup \{ \pp_{\boo+j}\}_{j \in \otbtt}$ and for any $j \in \otbtt$, $\pp_{\boo+j}$ is a convex combination of elements in $\bP$ and therefore $\pp_{\boo+j}(k) \geq \frac{1}{2\nkk^2}$ for all $k\in \otmd$. In the above expression we used the fact that each $i \in \otboo$, $(\bk \onevec)_{i}\leq 2\nkk^2$ for all $k \in \otmd$ (For any $i \in [1,\boo+\btt]$, $\pp_{i}(k) \geq 1/2\nkk^2$ and further combined with the constraint $\pvecext^{T}\bk \onevec\leq 1$ (because $\bk \in \bK^{ext}_{\phid}$) ensures this fact). Also,
 \begin{align*}
 \frac{\probsdpml(\bk)}{\bg(\bk)}& \geq \prod_{i=1}^{\boo+\btt}\prod_{j=0}^{\btt}\frac{\expo{\bk_{ij}\log \bk_{ij} -\bk_{ij}}}{\bk_{ij}!}\\ 
&  \geq \left(\prod_{i=1}^{\boo}\prod_{j=0}^{\btt} \frac{1}{e\sqrt{1+\bk_{ij}}}\right) \left( \prod_{j=1}^{\btt}\frac{1}{e\sqrt{1+\bk_{\boo+j,j}}}\right)\\
& \geq \left(\frac{1}{e\sqrt{1+2\nkk^2}}\right)^{\boo \bttpo +\btt}\\
 &\geq \expo{-O(\log \nkk) \boo \btt}
 \end{align*}
 In the second inequality we used the fact that solution $\bk$ returned by our rounding procedure always satisfies $\bk_{\boo+j,k}=0$ for all $j \in \otbtt$, $k \in \ztbtt$ and $k \neq j$.
\end{proof}

Using \Cref{eq:formulation3d}, for any $\bk \in \bK^{ext}_{\phid}$, if $\bq_{\bk}$ is its corresponding extended discrete $\md$-pseudodistribution, then
\begin{equation}\label{othersidedd}
\probpml(\bq_{\bk},\phid) \geq \Cphid \probsdpml(\bk)
\end{equation}

\begin{lemma}\label{lem:roundedsolnd}
The solution $\bk \in \bK^{ext}_{\phid}$ returned by \Cref{alg:roundingd} satisfies:
$$\probsdpml(\bk) \geq \expo{-O(\prod_{k=1}^{\md}\frac{\log^3 \nkk}{\epsok \epstk})}\probsdpml(\bk_{sdpml})$$
\end{lemma}
\begin{proof}
For any $\bk' \in \bK^{f}_{\phid}$ and $\bk \in \bK^{ext}_{\phid}$ returned by our rounding procedure below are the explicit expressions for $\bg(\bk)$ and $\bg(\bk')$:
$$ \bg(\bk)=\left(\prod_{i=1}^{\boo}\pp_i^{(\bk \mvec)_i}\frac{\expo{(\bk \onevec)_i \log (\bk \onevec)_i}}{\prod_{j=0}^{\btt}\expo{\bk_{ij}\log \bk_{ij}}}\right)\left(\prod_{j=1}^{\btt} \pp_{\boo+j}^{\nn_j\bk_{\boo+j,j} }\cdot 1\right)$$
$$ \bg(\bk')=\prod_{i=1}^{\boo}\left(\pp_i^{(\bk' \mvec)_i}\frac{\expo{(\bk ' \onevec)_i \log (\bk ' \onevec)_i}}{\prod_{j=0}^{\btt}\expo{\bk'_{ij}\log \bk'_{ij}}}\right)$$
We first bound the probability term:
\begin{equation}\label{eq:probtermd}
\begin{split}
\prod_{i=1}^{\boo}\pp_i^{(\bk' \mvec)_{i}}&=\left(\prod_{i=1}^{\boo}\pp_i^{(\bk \mvec)_i}\right)\left(\prod_{i=1}^{\boo}\pp_i^{\sum_{j=1}^{\btt}\nn_j(\bk'_{ij}-\bk_{ij})}\right)
=\left(\prod_{i=1}^{\boo}\pp_i^{(\bk \mvec)_i}\right)\left(\prod_{j=1}^{\btt} \prod_{i=1}^{\boo}\pp_i^{\nn_j(\bk'_{ij}-\bk_{ij})}\right)\\
&=\left(\prod_{i=1}^{\boo}\pp_i^{(\bk \mvec)_i}\right)\left(\prod_{j=1}^{\btt} \left(\prod_{i=1}^{\boo}\pp_i^{(\bk'_{ij}-\bk_{ij})}\right)^{\nn_j}\right)\\
&=\left(\prod_{i=1}^{\boo}\pp_i^{(\bk \mvec)_i}\right)\left(\prod_{j=1}^{\btt} \prod_{k=1}^{\md}\left(\prod_{i=1}^{\boo}\pp_i(k)^{(\bk'_{ij}-\bk_{ij})}\right)^{\nn_j(k)}\right)\\
&\leq \left(\prod_{i=1}^{\boo}\pp_i^{(\bk \mvec)_i}\right)\left( \prod_{j=1}^{\btt}\prod_{k=1}^{\md}\left(\frac{\sum_{i=1}^{\boo}\pp_i(k)(\bk'_{ij}-\bk_{ij})}{\sum_{i=1}^{\boo}(\bk'_{ij}-\bk_{ij})}\right)^{\nn_j(k)\sum_{i=1}^{\boo}(\bk'_{ij}-\bk_{ij})}\right) ~ (\because \text{AM-GM}~\forall  k)\\
&\leq \left(\prod_{i=1}^{\boo}\pp_i^{(\bk \mvec)_i}\right)\left(\prod_{j=1}^{\btt}\prod_{k=1}^{\md} \pp_{\boo+j}(k)^{\nn_j(k)\bk_{\boo+j,j}}\right)\\
&= \left(\prod_{i=1}^{\boo}\pp_i^{(\bk \mvec)_i}\right)\left(\prod_{j=1}^{\btt} \pp_{\boo+j}^{\nn_j\bk_{\boo+j,j}}\right)
\end{split}
\end{equation}
Final expression above is the probability term associated with $\bk$ and the equation above shows that our rounding procedure only increases the probability term and all that matters is to bound the counting term that we do next.
\begin{equation}\label{eq:counttermd}
\begin{split}
\frac{\bg(\bk)}{\bg(\bk')} & \geq \prod_{i=1}^{\boo}\frac{\expo{(\bk \onevec)_i \log (\bk \onevec)_i-(\bk' \onevec)_{i} \log (\bk' \onevec)_{i}}}{\prod_{j=0}^{\btt}\expo{\bk_{ij}\log \bk_{ij}-\bk'_{ij}\log \bk'_{ij}}} \geq \prod_{i=1}^{\boo}\expo{(\bk \onevec)_i \log (\bk \onevec)_i-(\bk' \onevec)_{i} \log (\bk' \onevec)_{i}}\\
&\geq \prod_{i=1}^{\boo} \expo{-O(\bttpo \log (\bk' \onevec)_{i}}) \geq \expo{-O(\prod_{k=1}^{\md}\frac{\log^3 \nkk}{\epsok \epstk})}
\end{split}
\end{equation}
In the derivation above we used (1) in Claim \ref{clmroundingd} and $(\bk' \onevec)_{i} \leq \min_{k \in \otmd}2\nkk^2$. It remains now to lower bound the quantity $\probsdpml(\bk)$:
\begin{align*}
\probsdpml(\bk) & \geq \expo{-O(\prod_{k=1}^{\md}\frac{\log^3 \nkk}{\epsok \epstk})}\bg(\bk) \geq \expo{-O(\prod_{k=1}^{\md}\frac{\log^3 \nkk}{\epsok \epstk})}\bg(\bk')\\
& \geq \expo{-O(\prod_{k=1}^{\md}\frac{\log^3 \nkk}{\epsok \epstk})}\bg(\bk_{sdpml})\geq \expo{-O(\prod_{k=1}^{\md}\frac{\log^3 \nkk}{\epsok \epstk})}\probsdpml(\bk_{sdpml})
\end{align*}

The first and second inequality follow from \Cref{app:roundsterd} and \Cref{eq:counttermd} respectively. In the third inequality we used $\bg(\bk') \geq \bg(\bk_{sdpml})$ because $\bk'$ is the optimal solution over the relaxed constraint set $\bK^{f}_{\phid}$ and finally invoked \Cref{app:sterlingsapproxd} to relate $\probsdpml$ and $\bg$.
\end{proof}

Now construct the $\md$-pseudodistribution $\pk$ corresponding to the solution $\bk$ returned by \Cref{alg:roundingd} by assigning $(\bk \onevec)_{i}$ elements to $\md$-level set $\pp_{i}$ $(\forall i \in [\boo+\btt])$. Our next theorem proves that the $\md$-distribution $\frac{\pk}{\|\pk\|_1}$ is an approximate PML $\md$-distribution.

\begin{thm}[Efficient and approximate PML for higher dimension]\label{thm:approxpmld}
Let $\md$ be a constant and $\yn$ be a $\md$-sequence of $\md$-length $\n=(\n(1),\dots,\n(\md))$. Let $\epsilon,\gamma\in \R^{1\times \md}$ be $\md$-tuples such that for each $k\in \otmd$, $\frac{1}{poly(\nkk)}<\epsok <1$, $\frac{1}{poly(\nkk)}<\epstk<1$, we can compute an $\exps{-\otilde\left(\otdisclossd+\otgrouplossepsd \right)}$-approximate PML $\md$-distribution $\bp_{approx}$ in time $\otilde\left(\writeprofile+ \runningtimed \right)$.
\end{thm}
\begin{proof}
Let $\pk$ be the $\md$-pseudodistribution corresponding to solution $\bk$ returned by \Cref{alg:roundingd}. Set $\bp_{approx}=\frac{\pk}{\|\pk\|_1}$, then:
\begin{align*}
\probpml(\frac{\pk}{\|\pk\|_1},\phi) &\geq \probpml(\pk,\phi) \\
 &\geq \expo{-\otilde\left(\otmultlossd \right)} \probpml(\pk,\phid) \\
 &\geq  \expo{-\otilde\left(\otmultlossd \right)} \Cphid \probsdpml(\bk) \\
&\geq \expo{-\otilde\left(\otsdpmllossd+\otmultlossd \right)}\Cphid \probsdpml(\bk_{sdpml}) \\
& \geq \expo{-\otilde\left(\otdisclossd+\otgrouplossepsd \right) }\probpml(\bp_{pml,\phi},\phi)
\end{align*}

The first inequality follows because $\|\pk\|_1\leq 1$, second inequality from \Cref{lemprofilediscd}, third inequality follows because $\bk \in \bK^{ext}_{\pk,\phid}$ (because we constructed $\pk$ from $\bk$) and $\probsdpml(\bk)$ computes just one term in the summation over $\bK^{ext}_{\pk,\phid}$ (look at the representation of $\probpml(\pk,\phid)$ as summation over $\bK^{ext}_{\pk,\phid}$ from \Cref{othersidedd}), fourth inequality comes from \Cref{lem:roundedsolnd} and last inequality follows from \Cref{lemsinglediscd}.

The total running time of our algorithms is the following: Given a $\md$-sequence $\yn$, it takes $\otilde(\writeprofile + \prod_{k=1}^{\md}\frac{1}{\epstk})$ to write down the discrete $\md$-profile $\phid$, then we need to solve the convex optimization problem \ref{eq:convoptoriginald} which further takes $\otilde\left(\runningtimed \right)$ and our final rounding algorithm can be implemented in time $\otilde(\md\prod_{k=1}^{\md}\frac{1}{\epsok\epstk} )$ ($=O(\md\boo\btt)$). The total running time combining all three steps in summarized in the lemma statement.
\end{proof}

To simplify the expression, for each $k \in [1,\md]$ substitute $\epsok=\epstk=\nkk^{-1/(2\md+1)}$ in the theorem above and in this parameter setting we achieve our best possible approximation ratio.
\thmPMLd*

%% file: multipmlsymmetric.tex
\renewcommand{\condo}{\textbf{C1}}
\renewcommand{\condt}{\textbf{C2}}
\renewcommand{\condth}{\textbf{C3}}

\subsection{Optimal sample complexity for KL divergence}\label{app:universald}
In this section we study the connection between optimal estimation of KL divergence and approximate PML $\md$-distribution. We restate theorem of \cite{ADOS16} we use earlier in one dimensional PML in terms of higher dimensional case.
 \begin{thm}[Theorem 4 of \cite{ADOS16}]\label{thmbetad}
For a symmetric property $\bff$, suppose there is an estimator $\hat{\bff}:\Phi^{\n}\rightarrow \R$, such that for any $\bp$ $\md$-distribution and observed $\md$-profile $\phi$, 
$$\bbP(|\bff(\bp)-\hat{\bff}(\phi)|\geq \epsilon) \leq \delta$$
any $\beta$-approximate PML distribution satisfies:
$$\bbP(|\bff(\bp)-\bff(\bp^{\beta}_{pml,\phi})|) \geq 2 \epsilon) \leq \frac{\delta |\Phi^{\n}|}{\beta}$$
\end{thm}

Let $\bp$ be a $2$-distribution, meaning it is $2$ dimensional with two distributions $\bp(1)$ and $\bp(2)$. Let $B$ be such that, $\forall x \in \bX$, $\frac{\bp(1)_{x}}{\bp(2)_{x}} \leq B$. We next define two conditions under which we get the optimal samples complexity for estimating KL divergence of distributions $\bp(1)$ and $\bp(2)$. $\bullet$ $\condo$ $\epsilon$, the estimation error satisfies $\epsilon > \frac{\log^3 N}{N}$. $\bullet$ $\condt$ $B \leq \epsilon^{2.24}N^{0.24}$.

 \begin{lemma}[Theorem 5 of \cite{Acharya18}]\label{KLchange}
Suppose $\condo$ and $\condt$ hold. Let $\alpha >0$ be a fixed (small) constant. There are constant $c_{1}$ and $c_{2}$ such that if $\n=(\n(1),\n(2))$
$$\n(1) \geq c_{1} \frac{N}{\epsilon \log N} \text{ and } \n(2) \geq c_{2} \frac{N\cdot B}{\epsilon \log N}$$
Given $\n(1)$ independent samples $\yno$ from distribution $\bp(1)$ and $\n(2)$ independent samples $\ynt$ from distribution $\bp(2)$, there exists an estimator $\hat{f}$ for estimating KL divergence $KL(\bp(1),\bp(2))$ that satisfies,
$$\probpml(|KL(\bp(1),\bp(2))-\hat{f}(\yno,\ynt)| \geq \epsilon) \leq \expo{-2\epsilon^2\min\{\n(1),\n(2)\}^{1-2\alpha}}$$
\end{lemma}

\begin{theorem}[\cite{Das12},\cite{BPA97}]\label{thm:card}
Let $\md >1$, and $\n=(\noo,\dots ,\ndd)$. The number of $\md$-profiles of $\md$-length equal to $\n$ is upper bounded by
$$|\Phi^{\n}| \leq \expo{3\sum_{k=1}^{\md} \nkk^{\md/(\md+1)}}$$
\end{theorem}

\theoremADOSuniversald*
\begin{proof}
Invoke \Cref{KLchange} with $\alpha=0.01$ and \Cref{thmbetad} with $\delta=\expo{-2\epsilon^2\min\{\n(1),\n(2)\}^{0.98}}$ we get:
\begin{align*}
\bbP\left(|\bff(\bp)-\bff(\bp^{\beta}_{pml,\phi})|\geq 2\epsilon \right) & \leq \frac{\delta |\Phi^{\n}|}{\beta}  \leq \frac{\expo{-2\epsilon^2\min\{\n(1),\n(2)\}^{0.98}}\expo{3 (\n(1)^{2/3}+\n(2)^{2/3})}}{\expo{-\otilde\left(\n(1)^{4/5} + \n(2)^{4/5} \right)}}\\
& \leq \expo{-2\epsilon^2\min\{\n(1),\n(2)\}^{0.98}}\expo{ \otilde\left(\n(1)^{4/5} + \n(2)^{4/5} \right)}\\
& \leq \expo{-2\epsilon^2 \left(\frac{N}{\epsilon \log N} \right)^{0.98}}\expo{\otilde\left(\left(\frac{BN}{\epsilon \log N}\right)^{4/5} \right)}\\
& \leq \expo{-O(N^{4/5})}
\end{align*}
In the first inequality we use \Cref{thm:card}.
\end{proof}

%% file: multipmlminprob.tex
\newcommand{\pmlset}{\textbf{K}}
\newcommand{\setpml}{\textbf{p}^{*}_{\pmlset,\givenp,\phi}}
\newcommand{\pstar}{\textbf{p}^{*}}
\newcommand{\givenp}{\textbf{r}}

\section{Remaining proofs for multidimensional PML}
\subsection{Minimum Probability}\label{app:multipmlmind}
In this section we provide the proof for our first technical lemma which states that one can assume the minimum non-zero probability of the PML distribution is $\Omega(\frac{1}{\n(k')^2})$ by only loosing a constant factor in the PML objective value.
To show such a result we use an independent rounding algorithm described in the lemma below.
\begin{claim}\label{clm1d}
For any non-negative and non-zero $\md$-vector $\vvec$ and a $\md$-profile $\phi \in \Phi^n$,
$$\probpml(\vvec,\phi)\leq \left(\prod_{k=1}^{\md}\|\vvec(k)\|_1^{\nkk}\right) \probpml(\bp_{pml,\phi},\phi)$$
\end{claim}
 \begin{proof}
 $$\probpml(\vvec,\phi) = \|\vvec\|_1^{\n}\probpml\left(\frac{\vvec}{\|\vvec\|_1},\phi\right)\leq \left(\prod_{k=1}^{\md}\|\vvec(k)\|_1^{\nkk}\right) \probpml(\bp_{pml,\phi},\phi)$$
 \end{proof}

For notational convenience we need the following definition of $\pmlset$-profile maximum likelihood $\md$-distribution.
\begin{defn}
For any set $\pmlset \subset \otmd$, $\md$-distribution $\givenp$ and profile $\phi \in \Phi^{\n}$, the $(\pmlset,\givenp)$-profile maximum likelihood $\md$-distribution denote by $\setpml$ is,
$$\setpml=\argmax_{\{\bp \in \simplexd ~|~ \forall k \in \pmlset , \bp(k)=\givenp(k) \}} \probpml(\bp,\phi)$$
\end{defn}

 \begin{restatable}{lemma}{lemmaminK}
\label{lemminK}
For any set $\pmlset \subset \otmd$, $\md$-distribution $\givenp$, index $k' \notin \pmlset$ and profile $\phi \in \Phi^{\n}$, there exists a $\md$-distribution $\bp'' \in \simplexd$ such that, 
\[
\bullet \min_{\{x \in \bX:\bpd''(k')_x \neq 0\}}\bpd''(k')_x \geq \frac{1}{2\n(k')^2}
~\bullet ~
\probpml(\bpd'',\phi) \geq \expo{-6}\probpml(\pstar_{\pmlset,\givenp,\phi},\phi)
~\bullet~
\bp''(k)=\givenp(k)~\forall k \in \pmlset
\]
\end{restatable} 
 
 \begin{proof}
 We do independent rounding to show the existence of such a solution. For notational convenience let $\pstar=\setpml$ and for $k' \in \otmd$ define $\bS_{k'}\defeq\{ x \in \bX ~|~ \pstar(k')_{x} <\frac{1}{\nkk^2} \}$ and we fix all the probability values in these sets next.

For all $x \in \bS_{k'}$ define a random variable $Y_x$ as follows:
 $$Y_x\defeq
 \begin{cases}
 \frac{1}{\n(k')^2} \quad \text{with probability } ~\n(k')^2\pstar(k')_{x}\\
 0 \quad ~~\text{otherwise} \\
 \end{cases}
 $$
 Clearly $\forall x \in S$, 
 \begin{equation}\label{eqS1}
 \E \left[ Y_x\right]=\pstar(k')_{x}
 \end{equation}
 and in general for any integer power $i$ of random variable $Y_x$ we have:
 \begin{equation}\label{eqS2}
 \E \left[ Y_x^i\right ] \geq  \bp_{pml,\phi}^i(x) \quad \forall i=2,\dots 
 \end{equation}
For the remaining $x \in {\bar{\bS}}_{k'}$ (${\bar{\bS}}_{k'}\defeq\bX \backslash \bS$) with $\pstar(k')_{x} \geq \frac{1}{\n(k')^2}$ we define:
 $$Z_x \defeq \pstar(k')_{x} \quad \text{with probability } ~1$$
 Define $\bY\defeq(Y_x)_{x\in \bS}$ and $\bZ\defeq(Z_x)_{x\in {\bar{\bS}}_{k'}}$. 
 $$\mu_{\bS}\defeq\E \left[ \|\bY\|_1\right]= \E \left[\sum_{x \in \bS_{k'}} Y_x\right]=\sum_{x \in \bS_{k'}}\E \left[ Y_x\right]=\sum_{x \in \bS_{k'}}p_{pml,\phi}(x)$$ 
 $$\mu_{{\bar{\bS}}_{k'}}\defeq\E \left[ \|\bZ\|_1\right]= \E \left[\sum_{x \in {\bar{\bS}}_{k'}} Z_x\right]=\sum_{x \in {\bar{\bS}}_{k'}}\E \left[ Z_x\right]=\sum_{x \in {\bar{\bS}}_{k'}}p_{pml,\phi}(x)$$
 $$\mu_{\bS}+\mu_{{\bar{\bS}}_{k'}}=1$$
 
Define $\bp$ as follows:
\[
\bp(k')=(\bY,\bZ)
\text{ and }
\bp(k)=\pstar(k)  \forall k\neq k'
\]
where $(\bY,\bZ)$ is the concatenation of random vectors $\bY$ and $\bZ$. All random variables $Y_x , Z_x$ are mutually independent and we have:

\begin{align*}
\E \left[ \probpml(\bp,\phi)\right ] \geq \probpml(\pstar,\phi)
\label{eq:expected_w1}
\end{align*}
(From Equation~\ref{eqS1},\ref{eqS2} and the fact that $Z_x$ is a constant).
 
We have a lower bound on the expected value of $\probpml(\bp,\phi)$ but this is misleading since $\bp$ may not be a $\md$-distribution as $\|\bp(k')\|_1$ could be greater than 1. Scaling norm of $\bp(k')$ to 1 could significantly reduce the value of $\probpml(\bp,\phi)$ if $\|\bp(k')\|_1$ is large. However, we show that a constant fraction of the expectation of $\probpml(\bp,\phi)$ comes from the sample space with bounded $\|\bp(k')\|_1\leq 1+\frac{c}{\n(k')}$. Here $c$ is a constant and assume $c\geq 3$. Note that:
 %$$\E \left[ \probpml(\bp,\phi) ~\Big| ~\|\bp\|_1 \leq 1+\frac{c}{\n(k')} \right]=\E \left[ \probpml(\bp,\phi) ~\Big| ~\|\bY\|_1 + \|\bZ\|_1 \leq 1+\frac{c}{\n(k')} \right]=\E \left[ \probpml(\bp,\phi) ~\Big| ~\|\bY\|_1 \leq \mu_{\bS}+\frac{c}{\n(k')} \right]$$
 $$\|\bp(k')\|_1 \leq 1+\frac{c}{\n(k')} \Leftrightarrow \|\bY\|_1 + \|\bZ\|_1 \leq 1+\frac{c}{\n(k')} \Leftrightarrow \|\bY\|_1 \leq \mu_{\bS}+\frac{c}{\n(k')}$$
 The last inequality follows because $\bZ$ is a constant random vector.
\begin{equation}\label{eqcondexp}
\begin{split}
\probpml(\pstar,\phi) \leq \E \left[ \probpml(\bp,\phi)\right ]  &=\E \left[ \probpml(\bp,\phi) ~\Big| ~\|\bY\|_1 \leq \mu_{\bS}+\frac{c}{\n(k')} \right]  \bbP \left[ \|\bY\|_1 \leq \mu_{\bS}+\frac{c}{\n(k')} \right] \\
&+ \E \left[ \probpml(\bp,\phi) ~\Big| ~\|\bY\|_1 > \mu_{\bS}+\frac{c}{\n(k')} \right]\bbP \left[ \|\bY\|_1 > \mu_{\bS}+\frac{c}{\n(k')} \right]
\end{split}
\end{equation}
To argue that a constant fraction of the expectation comes from the sample space with small $\|\bp\|_1$ we need a tight upper bound for:
$$\E \left[ \probpml(\bp,\phi) ~\Big| ~\|\bY\|_1 > \mu_{\bS}+\frac{c}{\n(k')} \right]\bbP \left[ \|\bY\|_1 > \mu_{\bS}+\frac{c}{\n(k')} \right]$$
For $t \geq c$, we first upper bound the probability term: 
 $$\bbP \left[ \|\bY\|_1 \geq \mu_{\bS}+\frac{t}{\n(k')} \right]$$
 We will use Chernoff bounds here and to apply them, we convert the $Y_x$ random variables into $\{0,1\}$ Bernoulli random variables. Define $\forall x \in \bS_{k'}$,
 $$Y'_x=\n(k')^2Y_x$$
 Equivalently:
  $$Y'_x\defeq
 \begin{cases}
 1 \quad \text{with probability} ~\n(k')^2\pstar(k')_{x}\\
 0 \quad ~~\text{otherwise} \\
 \end{cases}
 $$
 Define $\bY'\defeq(Y'_x)_{x \in \bS_{k'}}$ and $\mu'_{\bS}\defeq\E \left[\|\bY'\|_1 \right]=\n(k')^2\mu_{\bS} \leq \n(k')^2$. For any $t>0$,
 %$$\bbP \left[ \|\bY\|_1 \geq \mu_{\bS}+\frac{c'}{n} \right]=\bbP \left[ \|\bY'\|_1 \geq \n(k')^2\mu_{\bS}+c'n \right]=\bbP \left[ \|\bY'\|_1 \geq \mu'_{\bS}+c'n \right]$$
 $$\|\bY\|_1 \geq \mu_{\bS}+\frac{t}{\n(k')} \Leftrightarrow \|\bY'\|_1 \geq \n(k')^2\mu_{\bS}+t\n(k') \Leftrightarrow \|\bY'\|_1 \geq \mu'_{\bS}+t\n(k')$$
 
 Since $\|\bY'\|_1$ is a sum of Bernoulli random variables, by Chernoff bounds:
\begin{equation}\label{eqchernoff}
\begin{split}
\bbP \left[ \|\bY'\|_1 \geq \mu'_{\bS}+t\n(k') \right] &= \bbP \left[ \|\bY'\|_1 \geq \left( 1+\frac{t\n(k')}{\mu'_{\bS}} \right) \mu'_{\bS} \right]\leq \expo{-\frac{t^{2}\n(k')^2}{3\mu'^2_{\bS}}\mu'_{\bS}}=\expo{-\frac{t^2\n(k')^2}{3\mu'_{\bS}}}\\
& \leq \expo{\frac{-t^2}{3}}
\end{split}
\end{equation}
%Observe that any $\bp$ generated by the above random process always has $\|\bp\|_1$ of the form $\mu_{{\bar{\bS}}_{k'}}+\frac{i}{n}$ for some non-negative $i \in \indexset$, where $\indexset \defeq \{\frac{k}{n} ~|~ k \in \Z_{\geq 0} \}$.
Note $\|\bp(k)\|_1=1$ for all $k \neq k'$ and further applying \Cref{clm1d} we get:
\begin{equation}\label{eq:massbound}
\E \left[ \probpml(\bp,\phi) ~\Big| \|\bY\|_1 \leq \mu_{\bS}+\frac{t}{\n(k')}\right] \leq \probpml(\pstar,\phi) \left( 1+\frac{t}{\n(k')}\right)^{\n(k')} \leq \probpml(\pstar,\phi) \cdot e^t \triangleq H(t)
\end{equation}
\begin{align*}
\bbP \left[ \|\bY\|_1 > \mu_{\bS}+\frac{c}{\n(k')} \right] & \E \left[ \probpml(\bp,\phi) ~\Big| ~\|\bY\|_1 > \mu_{\bS}+\frac{c}{\n(k')}  \right]\\
& = \int_{t=c}^{\infty} \E \left[ \probpml(\bp,\phi) ~\Big| \|\bY\|_1 = \mu_{\bS}+\frac{t}{n}\right] \bbP \left[ \|\bY\|_1 = \mu_{\bS}+\frac{t}{n} \right] dt\\
&\leq  \int_{t=c}^{\infty} H(t) \bbP \left[ \|\bY\|_1 = \mu_{\bS}+\frac{t}{n} \right] dt \quad \text{(By \Cref{eq:massbound})}\\
& \leq \int_{t=c}^{\infty} \frac{d H(t)}{dt} \bbP \left[ \|\bY\|_1 > \mu_{\bS}+\frac{t}{\n(k')} \right] dt\\
& = \probpml(\pstar,\phi) \int_{t=c}^\infty e^t \expo{\frac{-t^2}{3}} dt\\
&= \probpml(\pstar,\phi) \frac{\expo{3/4} \sqrt{3 \pi}}{2} \left( 1-\mathrm{erf}\left(\frac{2c-3}{2 \sqrt{3}}\right)\right)\\
&\leq 0.75 \cdot \probpml(\pstar,\phi) \quad \mbox{for $c \geq 3$}
%& \leq \sum_{\{ i \in \indexset\}}\E \left[ \probpml(\bp,\phi) ~\Big| ~\frac{i}{n} \leq \|\bY\|_1-\left(\mu_{\bS}+\frac{c}{\n(k')}\right) \leq \frac{i+1}{n} \right]\bbP \left[ \frac{i}{n} \leq \|\bY\|_1 -\left(\mu_{\bS}+\frac{c}{\n(k')}\right) \leq \frac{i+1}{n} \right]\\
%& \leq \sum_{\{ i \in \indexset\}}\E \left[ \probpml(\bp,\phi) ~\Big| ~\|\bp\|_1 \leq 1+\frac{c}{\n(k')}+\frac{i+1}{n} \right]\bbP \left[ \frac{i}{n} \leq \|\bY\|_1 -\left(\mu_{\bS}+\frac{c}{\n(k')}\right) \leq \frac{i+1}{n} \right]\\
%& \leq \probpml(\bp_{pml,\phi},\phi) \sum_{\{ i \in  \indexset\}} \left( 1+\frac{c+i+1}{n}\right)^{n}  \bbP \left[ \frac{i}{n} \leq \|\bY\|_1 -\left(\mu_{\bS}+\frac{c}{\n(k')}\right) \leq \frac{i+1}{n} \right] (\because \text{Claim \ref{clm1d} and Eq. \ref{eqchernoff} })\\
%& {\color{red}\leq \probpml(\bp_{pml,\phi},\phi) \sum_{\{ i \in \indexset\}}\expo{c+i+1}  \bbP \left[ \frac{i}{n} \leq \|\bY\|_1 -\left(\mu_{\bS}+\frac{c}{\n(k')}\right) \leq \frac{i+1}{n} \right]}\\
%& {\color{red}\leq \frac{1}{2}\probpml(\bp_{pml,\phi},\phi)}   
\end{align*}
Substituting back in Equation \ref{eqcondexp} we have (for $c \geq 3$),
 $$\E \left[ \probpml(\bp,\phi) ~\Big| ~\|\bY\|_1 \leq \mu_{\bS}+\frac{c}{\n(k')} \right]  \bbP \left[ \|\bY\|_1 \leq \mu_{\bS}+\frac{c}{\n(k')} \right] \geq \frac{1}{4}\probpml(\pstar,\phi)$$
 $$\Rightarrow \E \left[ \probpml(\bp,\phi) ~\Big| ~\|\bY\|_1 \leq \mu_{\bS}+\frac{c}{\n(k')} \right] \geq \frac{1}{4}\probpml(\pstar,\phi)$$
 $$\Rightarrow \E \left[ \probpml(\bp,\phi) ~\Big| ~\|\bp\|_1 \leq 1+\frac{c}{\n(k')} \right] \geq \frac{1}{4}\probpml(\pstar,\phi)$$
 The above inequality implies existence of a $\bp'$ with $\probpml(\bp',\phi)\geq \frac{1}{4}\probpml(\pstar,\phi)$ and $ \|\bp'\|_1 \leq 1+\frac{c}{\n(k')} $. Define $\bp''$,
 $$\bp''(k')\defeq \frac{\bp'(k')}{\|\bp'(k')\|_1} \text{ and } \bp''(k)\defeq\bp'(k)=\pstar(k)~ \forall k\neq k'$$
 The above inequality further implies,
 $$\bp''(k')=\pstar(k)=\givenp(k)~ \forall k\in \pmlset$$
 $$\probpml(\bp'',\phi) = \|\bp'(k')\|_1 ^{-\n(k')} \probpml(\bp',\phi) \geq (1+\frac{c}{\n(k')})^{-\n(k')}\frac{1}{4}\probpml(\pstar,\phi)\geq \frac{\expo{-c}}{4}\probpml(\pstar,\phi)$$
 In the final inequality substitute $c=3$ and observe $\frac{\expo{-c}}{4}\geq \expo{-6}$. Also our rounding procedure always ensures that minimum non-zero entry of $\bp'$ is $\geq \frac{1}{\n(k')^2}$ that further implies a lower bound on the minimum non-zero probability value of $\bp''$ to be $\frac{1}{\n(k')^2}\frac{1}{\|\bp'\|_1}=\frac{1}{\n(k')^2}\frac{1}{1+3/ \n(k')}\geq \frac{1}{2\n(k')^2}$. Hence $\bp''$ is our final distribution satisfying the conditions of lemma.
 \end{proof}
 
\lemmamind*
\begin{proof}
The Lemma follows by induction and call to \Cref{lemminK}.

\noindent \textbf{Induction statement}: For $i \in \otmd$, let $\bp^{(i)}$ be the $\md$-distribution satisfying $\min_{x \in \bX:\bp^{(i)}(k)_x \neq 0}\bp^{(i)}_x(k) \geq \frac{1}{2\nkk^2}$ for all $k \leq i$ and is a $\expo{-6i}$-approximate PML $\md$-distribution.

\noindent \textbf{Base Case}: Apply \Cref{lemminK} by setting $\pmlset=\{\}$ an empty set, $\givenp=\bp_{pml,\phi}$ and $k'=1$. Note that $\pstar_{\pmlset,\givenp,\phi}=\bp_{pml,\phi}$ and the $\md$-distribution returned by \Cref{lemminK} is $\expo{-6i}$-approximate PML $\md$-distribution.

\noindent \textbf{Induction step for $i+1$}: Apply \Cref{lemminK} by setting $\pmlset=\{[1,i]\}$, $\givenp=\bp^{(i)}$ and $k'=i+1$. Note that $\probpml(\pstar_{\pmlset,\givenp,\phi},\phi)\geq \probpml(\bp^{(i)},\phi) \geq \expo{-6i}\probpml(\bp_{pml,\phi},\phi)$ (By induction step) and the $\md$-distribution returned by \Cref{lemminK} $\bp^{(i+1)}$ further satisfies  $\probpml(\bp^{(i+1)},\phi) \geq \expo{-6}\probpml(\pstar_{\pmlset,\givenp,\phi},\phi)\geq \expo{-6(i+1)}\probpml(\bp_{pml,\phi},\phi)$ and is therefore a $\expo{-6(i+1)}$-approximate PML $\md$-distribution. Also by \Cref{lemminK} $\bp^{(i+1)}(k)=\bp^{(i)}(k)$ for all $k\leq i$ and $\min_{x \in \bX:\bp^{(i+1)}_x \neq 0}\bp^{(i+1)}_x \geq \frac{1}{2\n(i+1)^2}$. Combining everything we satisfy induction step for $i+1$.

Set $\bp''=\bp^{(\md)}$ and by induction we get that induction step holds for $i=\md$ and the lemma statement follows.
\end{proof}

%% file: multipmleigenbound.tex
\newcommand{\mg}{\textbf{G}}
\newcommand{\mv}{\textbf{V}}
\renewcommand{\mu}{\textbf{U}}
\newcommand{\vv}{\textbf{v}}
\newcommand{\vu}{\textbf{u}}
\newcommand{\vvp}{\tilde{\textbf{v}}}
\newcommand{\ppp}{\tilde{\pp}}
\newcommand{\eigengram}{\Omega(\boo\frac{1}{\sum_{k \in \otmd} \log^2 \nkk})}

\subsection{Eigenvalue bounds for Gram matrix}\label{app:eigenlb}
Here we provide a lower bound for the minimum eigenvalue of a invertible Gram matrix. First, in Lemma~\ref{lem:gram} we provide an explicit expression for the trace of inverse of a Gram matrix. Then, leveraging that $\lambda_{\min}(\mg) \geq 1/\tr(\mg^{-1})$  we obtain Corollary~\ref{cor:grameigenlb}, our desired lower bound.

\begin{lemma}
\label{lem:gram}
For an invertible Gram matrix $\mg \in \R^{\md \times \md}$ of a set of vectors $\vv_{1},\dots, \vv_{\md} \in \R^{\boo}$.
$$\tr(\mg^{-1})=\sum_{k=1}^{\md} \frac{1}{\|\vvp_{k}\|_2^2}$$
where $\vvp_{k}$ is the orthogonal projection of $\vv_{k}$ onto $span(\vv_{1},\dots,\vv_{k-1},\vv_{k+1},\dots,\vv_{\md})^{\perp}$.
\end{lemma}
\begin{proof}
Recall,
\begin{equation}\label{eq:gram}
\tr(\mg^{-1})=\sum_{k=1}^{\md}(\mg^{-1})_{kk}
\end{equation}
Let $\mv \in \R^{\boo \times \md}$ be the matrix with columns $\vv_{1},\dots, \vv_{\md}$. For each $k \in \otmd$ we next give explicit formula for scalar $(\mg^{-1})_{kk}$. Let $\mv_{k} \in \R^{\boo \times (\md-1)}$ be the matrix with $k$th column removed from matrix $\mv$. From the definition of $\mg^{-1}$ and for all $k \in \otmd$, the $k$'th diagonal entry of $\mg^{-1}$ is given by:
$$(\mg^{-1})_{kk}=\frac{\det(\mv_{k}^{T}\mv_{k})}{\det(\mv^{T}\mv)}$$
Using Theorem (3) combined with Equation (3.2) in \cite{MR16} we get,
$$\det(\mv^{T}\mv)=\|\vvp_{k}\|_2^2 \det(\mv_{k}^{T}\mv_{k})  \implies (\mg^{-1})_{kk}=\frac{1}{\|\vvp_{k}\|_2^2}$$
The lemma statement follows by substituting value of $(\mg^{-1})_{kk}$ in \Cref{eq:gram}.
\end{proof}

\begin{cor}\label{cor:grameigenlb}
For an invertible Gram matrix $\mg \in \R^{\md \times \md}$ of a set of vectors $\vv_{1},\dots, \vv_{\md} \in \R^{\boo}$.
$$\lambda_{min}(\mg) \geq \frac{1}{\sum_{k=1}^{\md} \frac{1}{\|\vvp_{k}\|_2^2}}$$
where $\vvp_{k}$ is the orthogonal projection of $\vv_{k}$ onto $span(\vv_{1},\dots,\vv_{k-1},\vv_{k+1},\dots,\vv_{\md})^{\perp}$.
\end{cor}

\subsection{Singular value lower bound for constraint matrix}
Here we show a lower bound for the minimum singular value of our constraint matrix $\ma$ for multidimensional PML. First in \Cref{lem:gramnorm}, we give a lower bound on the norm of orthogonal projection of each column onto span of remaining columns for the $\md$-level set matrix $\pp$ (defined in \Cref{sec:structured}). This result combined with \Cref{cor:grameigenlb} gives a lower bound for the minimum singular value for $\pp$. Then in \Cref{lem:eigenlower}, we lower bound the minimum singular value of $\ma$ in terms of minimum singular value of $\pp$ to achieve our desired lower bound.

Now, recall that $\bP$ is the set of all vectors $x \in \R^d$ where $x(k) = (1 + \epsok)^{1 - j}$ for some $j \in [1, \book]$, where $\book$ for each $k \in \otbook$ is such that $(1+\epsok)^{1-\book}\leq \frac{1}{2\n(k)^2}$ and $\boo=\prod_{k=1}^{\md} \book$. Further, the $\md$-level set matrix $\pvec \in \R^{\boo \times \md}$ is the defined as the matrix whose rows are exactly the elements of $\bP$. 

\begin{lemma}\label{lem:gramnorm}
For $\pp \in \R^{\boo \times \md}$ and $k \in \otmd$, if $\pp(k)$ is its $k$'th column, then the following inequality holds,
$$\|\ppp(k)\|_2^{2} \geq \Omega(\frac{\boo}{\log^2 \nkk})$$
where $\ppp(k)$ is the orthogonal projection of $\pp(k)$ onto $span(\pp(1),\dots,\pp(k-1),\pp(k+1),\dots,\pp(\md))^{\perp}$.
\end{lemma}
\begin{proof}
For each index $k \in \otmd$, there are multiple blocks each of size $\book$ and for each $k_{i}$th block $I_{k_{i}} \subset \otboo$ and $k' \in \otmd$ and $k'\neq k$,
\[
\pp_{I_{k_{i}}}(k')=c_{k'}\onevec_{\book}
\text{ and }
\pp_{I_{k_{i}}}(k)=(\frac{1}{2\nkk^2},\dots \frac{1}{2\nkk^2}(1+\epsok)^{i},\dots 1)
\]
for each scalar $c_{k'} \in \{\frac{1}{2\n(k')^2},\dots \frac{1}{2\n(k')^2}(1+\epso(k'))^{i},\dots 1\}$
%and this further implies $k_{i} \in [1,\prod_{k' \in \otmd|k'\neq k}\boo_{k'}]$ 
and the number of blocks satisfying above equalities is equal to $\prod_{k' \in \otmd|k'\neq k}\boo_{k'}$. 

Note $span(\pp_{I_{k_{i}}}(1),\dots, \pp_{I_{k_{i}}}(k-1),\pp_{I_{k_{i}}}(k+1),\dots, \pp_{I_{k_{i}}}(\md))^{\perp}$ is same as $span(\onevec_{\book})^{\perp}$ and if $\ppp_{I_{k_{i}}}(k)$ is the orthogonal projection of $\pp_{I_{k_{i}}}(k)$ onto $span(\pp_{I_{k_{i}}}(1),\dots, \pp_{I_{k_{i}}}(k-1),\pp_{I_{k_{i}}}(k+1),\dots, \pp_{I_{k_{i}}}(\md))^{\perp} =span(\onevec_{\book})^{\perp}$, then:
$$\|\ppp_{I_{k_{i}}}(k)\|_2^2 \in \Omega(\frac{\book}{\log^2 \nkk})$$
%\Omega(\|\pp_{I_{k_{i}}}(k)\|_2^2)=\Omega(\|(\frac{1}{2\nkk^2},\dots \frac{1}{2\nkk^2}(1+\epsok)^{i},\dots 1)\|_2^2)
The above result combined with number of such blocks gives:
$$\|\ppp(k)\|_2^2 \geq \Omega(\frac{\book}{\log^2 \nkk}) \times \prod_{k' \in \otmd|k'\neq k}\boo_{k'} \geq \Omega(\frac{\boo}{\log^2 \nkk})$$
\end{proof}

\begin{cor}\label{cor:eigenlbgram}
The minimum eigenvalue of matrix $\pp^{T}\pp$ is at least $\eigengram$.
\end{cor}

\noindent Now lets consider our constraint matrix $\ma \in \R^{(\btt+1+\md) \times \boo \cdot \bttpo}$ for multidimensional PML, \footnote{Our matrix $\ma$ is a sparse matrix and matrix vector product with it can be computed in time $O(\boo \cdot \btt)$}

\[
   \ma=
  \left[ {\begin{matrix}
 &  1 \dots 1  &  0\dots 0 & 0 \dots 0 & 0 \dots 0  \\
  & 0 \dots 0 & 1 \dots 1 & 0 \dots 0 & 0 \dots 0 \\
& \vdots & \ddots &\ddots & &\\
& 0 \dots 0 & 0 \dots 0 & 1 \dots 1 & 0 \dots 0\\
& \pp^\top & \dots & \pp^\top & \pp^\top\\
  \end{matrix} } \right]
\]

\newcommand{\mi}{\textbf{I}}

\begin{lemma}\label{lem:eigenlower}
The eigenvalues of matrix $\ma \ma^\top$ are at least $\Omega(\frac{\boo}{\btt})$.
\end{lemma}

\begin{proof}
Direct calculation shows that if $\vec{1}_{\btt} \in \R^{\btt}$, $\vec{1}_{\boo} \in \R^{\boo}$ are $\btt$,$\boo$ dimensional all ones vector respectively and $\mi_{\btt} \in \R^{\btt \times \btt}$ is the $\btt$-dimensional identity matrix then for all $x \in \R^{\btt}$ and $\alpha \in \R^{\md}$ we have
\[
\ma \ma^\top 
\left(
\begin{matrix}
x \\
\alpha
  \end{matrix} 
\right)
=
\left[ 
\begin{matrix}
 &  \boo \mi_{\btt}  &  \vec{1}_{\btt} \vec{1}_{\boo}^{T}\pp  \\
  &  (\vec{1}_{\btt} \vec{1}_{\boo}^{T}\pp)^\top  & \bttpo \pvec^{T} \pvec 
  \end{matrix} \right]
\left(
\begin{matrix}
x \\
\alpha
  \end{matrix} 
\right)
=
\left(
\begin{matrix}
\boo x +   \vec{1}_{\btt} (\vec{1}_{\boo}^{T}\pp \alpha) \\
 (\vec{1}_{\btt} \vec{1}_{\boo}^{T}\pp)^\top x +  \bttpo \pvec^{T} \pvec \alpha
  \end{matrix} 
\right)
~.
\]
Consequently $v = (x, \alpha)^T$ is an eigenvector of $\ma \ma^\top $ with eigenvalue $\lambda$ if and only if
\[
\boo x +   \vec{1}_{\btt} (\vec{1}_{\boo}^{T}\pp \alpha) = \lambda x
\text{ and }
 (\vec{1}_{\btt} \vec{1}_{\boo}^{T}\pp)^\top x +  \bttpo \pvec^{T} \pvec \alpha = \lambda \alpha
\]
Now if $x \perp \vec{1}_{\btt}$ then we see the $v$ is an eigenvector if and only if $\alpha \perp\pp^{T} \vec{1}_{\boo}$ in which case the eigenvalues are $\boo$. On the other hand if $x = \vec{1}_{\btt}$ then we see $v$ is an eigenvector of eigenvalue $\lambda$ if and only if 
\[
\lambda = \boo +  (\vec{1}_{\boo}^{T}\pp \alpha)
\text{ and }
\btt \pp^{T}\vec{1}_{\boo} +  \bttpo \pvec^{T} \pvec \alpha = \lambda \alpha ~.
\]
When this happens we either have $\lambda \geq \bttpo \lambda_{\min}(\pvec^{T} \pvec)$ or in the case of $\lambda < \bttpo \lambda_{\min}(\pvec^{T} \pvec)$ the following holds,
\[
\alpha=\btt (\lambda \mi_{\md} -\bttpo \pvec^{T} \pvec)^{-1}\pp^{T}\vec{1}_{\boo}
\text{ and }
\lambda = \boo +  \btt \vec{1}_{\boo}^{T}\pp (\lambda \mi_{\md} -\bttpo \pvec^{T} \pvec)^{-1}\pp^{T}\vec{1}_{\boo}
\]
To simplify the expression above, let the following be the SVD for $\pvec$, 
$$\pvec=\sum_{i=k}^{\md}\sigma_{k}\vu_{k}\vv_{k}^{T}$$
where $\sigma_{1}\leq \sigma_{2}\dots \leq \sigma_{\md}$ are singular values and $\sigma_1^2=\lambda_{\min}(\pp^{T} \pp)$. In this notation the eigenvalue decomposition of matrix $\pp (\lambda \mi_{\md} -\bttpo \pvec^{T} \pvec)^{-1}\pp^{T}$ is equal to:
$$\pp (\lambda \mi_{\md} -\bttpo \pvec^{T} \pvec)^{-1}\pp^{T}=\sum_{k=1}^{\md}\frac{\sigma_{k}^2}{\lambda-\bttpo\sigma_{k}^2}\vu_{k}\vu_{k}^{T}$$
Further we can write closed form expression for $\lambda$ in terms of singular values and left singular value vectors of matrix $\pp$.
\begin{equation}\label{eq:eigen}
\lambda + \btt\sum_{k=1}^{\md}\frac{\sigma_{k}^2(\vu_{k}^{T}\vec{1}_{\boo})^2}{\bttpo\sigma_{k}^2-\lambda}= \boo
\end{equation}
We use $\fnh(\lambda)$ to denote the expression on the left hand side,
$$\fnh(\lambda)=\lambda + \btt\sum_{k=1}^{\md}\frac{\sigma_{k}^2(\vu_{k}^{T}\vec{1}_{\boo})^2}{\bttpo\sigma_{k}^2-\lambda}$$
We know that $\lambda \geq 0$ because $\ma \ma^{T}$ is PSD. For $\lambda \in [0, \bttpo \lambda_{\min}(\pvec^{T} \pvec))$, $\fnh(\lambda)>0$ and is strictly increasing in $\lambda$. Further \Cref{eq:eigen} has a unique solution $\lambda^{*}$ (if a solution exists) in the interval $[0, \bttpo \lambda_{\min}(\pvec^{T} \pvec))$.

To give a lower bound of $\ell$ on $\lambda^{*}$, if suffices to find a $\lambda$, such that $\fnh(\lambda) < \boo$ and we get $\ell \geq \lambda$. For $\lambda =\min(\frac{1}{2} \sigma_{1}^2, \frac{\boo}{2(2\btt+1)})$, we have $\btt\sum_{k=1}^{\md}\frac{\sigma_{k}^2(\vu_{k}^{T}\vec{1}_{\boo})^2}{\bttpo\sigma_{k}^2-\lambda} \leq \boo \frac{\btt}{\btt+1/2}$, then later combined with $\lambda \leq \frac{\boo}{2(2\btt+1)}$, we get $\fnh(\lambda) <\boo$ and therefore $\lambda^{*}  \geq \min( \frac{1}{2} \sigma_{1}^2, \frac{\boo}{2(2\btt+1)})$.
Combining all cases together we have that $\lambda_{min}(\ma \ma^{T}) \geq \min(\boo, \bttpo \lambda_{min}(\pp^{T}\pp),\frac{1}{2} \sigma_{1}^2, \frac{\boo}{2(2\btt+1)})$. Combined with \Cref{cor:eigenlbgram} we have our result.

\end{proof}